\pdfoutput=1
\documentclass[acmsmall,screen]{acmart}
\acmJournal{PACMPL}
\acmVolume{1}
\acmNumber{POPL} % CONF = POPL or ICFP or OOPSLA
\acmArticle{1}
\acmYear{2021}
\acmMonth{1}
\acmDOI{} % \acmDOI{10.1145/nnnnnnn.nnnnnnn}
\startPage{1}

\setcopyright{acmcopyright}
\usepackage{booktabs}   %% For formal tables:
\usepackage{subcaption} %% For complex figures with subfigures/subcaptions
\usepackage{tikz}
\usetikzlibrary{shapes.misc,arrows.meta}

\usepackage[utf8]{inputenc} 
\usepackage{amsmath, amsthm}
\usepackage{thmtools}
\usepackage[english]{babel}
\usepackage{pdfpages}
\usepackage{paralist}
\usepackage{graphicx}
\usepackage{subcaption}
\usepackage{caption}
\usepackage{longtable}
\usepackage{booktabs}
\usepackage{tabu}
\usepackage[referable]{threeparttablex}
\usepackage{varwidth}
\usepackage{multirow}
\usepackage{wrapfig}
\usepackage{enumerate}
\usepackage[inline]{enumitem}
\usepackage{microtype}
\usepackage{xifthen}

\usepackage{pifont}

\usepackage[ruled,vlined,resetcount,linesnumbered]{algorithm2e}
\usepackage{parskip}

\usepackage{hyperref}
\usepackage{thm-restate}
\usepackage[capitalise]{cleveref}

\crefname{figure}{Figure}{Figure}
\crefformat{footnote}{#2\footnotemark[#1]#3}
\usepackage{tikz}
\usepackage{comment}

\usepackage{todonotes}
\usepackage{tikz}
\usetikzlibrary{arrows,automata,shapes,decorations,decorations.markings,calc, matrix,decorations.pathmorphing, patterns,backgrounds}

\usepackage{bm}
\usepackage{pgfplots}
\usepackage{enumerate,paralist}
\usepackage{pbox}

\usepackage{appendix}
\usepackage{calc}
\usepackage{fp}
\usepackage{multirow}
\usepackage{pbox}

\usepackage[mathscr]{euscript}
\usepackage{mathtools}

\usepackage{footnote}
\usepackage{bbm}
\usepackage{multicol}
\include{macros}

\makeatletter
\g@addto@macro\bfseries{\boldmath}
\makeatother

\renewcommand{\smallskip}{}
\begin{document}

%% Title information
\title[Optimal Prediction of Synchronization-Preserving Races]{Optimal Prediction of Synchronization-Preserving Races}
                                        %% when present, will be used in
                                        %% header instead of Full Title.
%\titlenote{with title note}             %% \titlenote is optional;
                                        %% can be repeated if necessary;
                                        %% contents suppressed with 'anonymous'
%\subtitle{Subtitle}                     %% \subtitle is optional
%\subtitlenote{with subtitle note}       %% \subtitlenote is optional;
                                        %% can be repeated if necessary;
                                        %% contents suppressed with 'anonymous'

%% Author information
%% Contents and number of authors suppressed with 'anonymous'.
%% Each author should be introduced by \author, followed by
%% \authornote (optional), \orcid (optional), \affiliation, and
%% \email.
%% An author may have multiple affiliations and/or emails; repeat the
%% appropriate command.
%% Many elements are not rendered, but should be provided for metadata
%% extraction tools.

\author{Umang Mathur}
\affiliation{
%\position{Position1}
%\department{Department1}             %% \department is recommended
\institution{University of Illinois, Urbana Champaign}            %% \institution is required
%\streetaddress{Street1 Address1}
%\city{Aarhus}
%\state{State1}
%\postcode{Post-Code1}
\country{USA}                    %% \country is recommended
}
\email{umathur3@illinois.edu}          %% \email is recommended

\author{Andreas Pavlogiannis}
\affiliation{
%\position{Position1}
%\department{Department1}             %% \department is recommended
\institution{Aarhus University}            %% \institution is required
%\streetaddress{Street1 Address1}
%\city{Aarhus}
%\state{State1}
%\postcode{Post-Code1}
\country{Denmark}                    %% \country is recommended
}
\email{pavlogiannis@cs.au.dk}          %% \email is recommended

\author{Mahesh Viswanathan}
\affiliation{
%\position{Position1}
%\department{Department1}             %% \department is recommended
\institution{University of Illinois, Urbana Champaign}            %% \institution is required
%\streetaddress{Street1 Address1}
%\city{Aarhus}
%\state{State1}
%\postcode{Post-Code1}
\country{USA}                    %% \country is recommended
}
\email{vmahesh@illinois.edu}          %% \email is recommended

%!TEX root = main.tex

\begin{abstract}
Concurrent programs are notoriously hard to write correctly, as scheduling nondeterminism introduces subtle errors that are both hard to detect and to reproduce.
The most common concurrency errors are \emph{(data) races}, which occur when memory-conflicting actions are executed concurrently.
Consequently, considerable effort has been made towards developing efficient techniques for race detection.
The most common approach is \emph{dynamic race prediction}:~
given an observed, race-free trace $\tr$ of a concurrent program, the task is to decide whether events of $\tr$ can be correctly reordered to a trace $\tr^*$ that witnesses a race hidden in $\tr$.

In this work we introduce the notion of \emph{sync(hronization)-preserving races}.
A sync-preserving race occurs in $\tr$ when there is a witness $\tr^*$ 
in which synchronization operations (e.g., acquisition and release of locks) appear in the same order as in $\tr$.
This is a broad definition that \emph{strictly subsumes} the famous notion of happens-before races.
Our main results are as follows.
% First, we develop a \emph{sound and complete} algorithm for predicting sync-preserving races that runs in $\Otilde(\NumEvents)$ time and space, where $\NumEvents$ is the length of $\tr$, for up to a moderate number of other parameters, such as the number of threads.
First, we develop a sound and complete algorithm for predicting sync-preserving races.
For moderate values of parameters like the number of threads, the algorithm runs in $\Otilde(\NumEvents)$ time and space, where $\NumEvents$ is the length of the trace $\tr$.
Second, we show that the problem has a $\Omega(\NumEvents/\log^2 \NumEvents)$ space lower bound, and thus our algorithm is essentially \emph{time and space optimal}.
Third, we show that predicting races with \emph{even just a single} reversal of two sync operations is $\NP$-complete and even $\W{1}$-hard when parameterized by the number of threads.
Thus, sync-preservation characterizes \emph{exactly} the tractability boundary of race prediction, and our algorithm is nearly \emph{optimal} for the tractable side.
Our experiments show that our algorithm is fast in practice, while sync-preservation characterizes races often missed by state-of-the-art methods.
\end{abstract}

%% 2012 ACM Computing Classification System (CSS) concepts
%% Generate at 'http://dl.acm.org/ccs/ccs.cfm'.
\begin{CCSXML}
<ccs2012>
<concept>
<concept_id>10011007.10011074.10011099</concept_id>
<concept_desc>Software and its engineering~Software verification and validation</concept_desc>
<concept_significance>500</concept_significance>
</concept>
<concept>
<concept_id>10003752.10010070</concept_id>
<concept_desc>Theory of computation~Theory and algorithms for application domains</concept_desc>
<concept_significance>300</concept_significance>
</concept>
<concept>
<concept_id>10003752.10010124.10010138.10010143</concept_id>
<concept_desc>Theory of computation~Program analysis</concept_desc>
<concept_significance>300</concept_significance>
</concept>
</ccs2012>
\end{CCSXML}

\ccsdesc[500]{Software and its engineering~Software verification and validation}
\ccsdesc[300]{Theory of computation~Theory and algorithms for application domains}
\ccsdesc[300]{Theory of computation~Program analysis}

%% End of generated code

%% Keywords
%% comma separated list
\keywords{concurrency, dynamic analysis, race detection, complexity}  %% \keywords are mandatory in final camera-ready submission

%% \maketitle
%% Note: \maketitle command must come after title commands, author
%% commands, abstract environment, Computing Classification System
%% environment and commands, and keywords command.
\maketitle

%!TEX root = main.tex

\section{Introduction}
\seclabel{intro}

The verification of concurrent programs is one of the main challenges in formal methods.
Concurrency adds a dimension of non-determinism to program behavior which stems from inter-process communication.
Accounting for such non-determinism during program development is a challenging mental task, making concurrent programming significantly error-prone.
At the same time, bugs due to concurrency are very hard to reproduce manually, 
and automated techniques for doing so are crucial in enhancing the productivity of software developers.

Data races are the most common form of concurrency errors.
% The most reliable indicators of errors in concurrency are data races.
A data race (sometimes just called a race) occurs when a thread of a multi-threaded 
program accesses a shared memory location 
while another thread is modifying it without proper synchronization. 
The presence of a data race is often symptomatic of a serious bug in the program~\cite{lpsz08};
races have caused data corruption and compilation errors~\cite{boehmbenign2011,racemob,Narayanasamy2007}, 
and significant system errors~\cite{SoftwareErrors2009,evil2012} in the past.
Therefore, considerable research has focused on detecting and preventing races in multi-threaded programs.

One of the most popular approaches to race prediction is via dynamic analysis~\cite{fasttrack,bond2010pacer,Pozniansky:2003:EOD:966049.781529}.
Unlike static analysis, dynamic race prediction is performed at runtime. 
Such techniques determine if an observed execution provides evidence for the existence of a \emph{possibly alternate} program execution that can concurrently perform conflicting data accesses\footnote{Conflicting data accesses  come from different threads, access a common memory location, and at least one is a write.}.
The underlying principle is that a race is present but ``hidden'' in a large number of different program executions; 
hence techniques that uncover such hidden races can accelerate the
process of debugging concurrent programs significantly.
The popularity of dynamic race prediction techniques further stems
(i)~from their scalability to large production software, and 
(ii)~from their ability to produce only sound error reports.

The most popular dynamic race prediction techniques are 
based on Lamport's happens-before partial order~\cite{lamport1978time}. 
These techniques scan the input trace, determine happens-before orderings on-the-fly, 
and report a race on a pair of conflicting data accesses if they are unordered by happens-before. 
This approach is sound, in that the presence of unordered conflicting data accesses ensures the existence of an execution with a race. 
While happens-before based analysis fails to predict races in various cases~\cite{cp2012}, 
its wide deployment is based on the fact that the algorithm is fast, single pass, and runs in linear time. 
The principle that forms the basis of its efficiency is the following.
When reasoning about alternate executions, happens-before analysis does not consider any execution in which the order of synchronization primitives is reversed from that in the observed execution. 
We call such alternate executions \emph{sync(hronization)-preserving executions}.
Other, more powerful race prediction techniques~\cite{cp2012,rv2014,Huang2016,Roemer18,Bond2019,PavlogiannisPOPL20} sacrifice this principle and consider alternate executions that are not sync-preserving.
Naturally, this typically results in performance degradation, as the problem is in general $\NP$-hard~\cite{Mathur20},
and considerable efforts are made towards improving the scalability of such techniques~\cite{Roemer20,roemeronline2019}.

%!TEX root = ../main.tex

\begin{figure}[t]
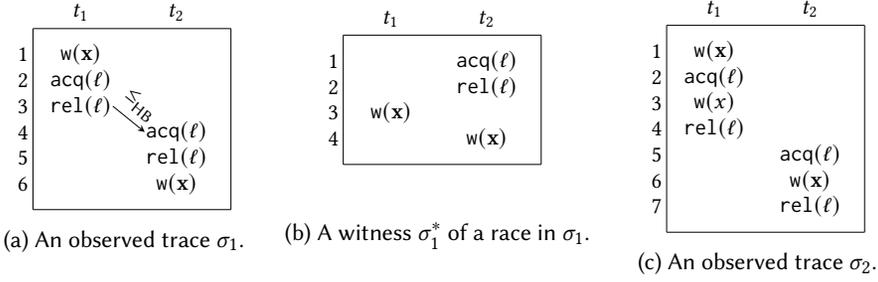

\centering
\begin{subfigure}{.3\textwidth}
	\centering
	\execution{2}{
		\figev{1}{\mathbf{\wt(x)}}
		\figev{1}{\acq(\lk)}
		\figev{1}{\rel(\lk)}
		\figev{2}{\acq(\lk)}
		\figev{2}{\rel(\lk)}
		\figev{2}{\mathbf{\wt(x)}}
		\orderedgewithlabel{1}{3}{0.5}{2}{4}{-0.5}{\small $\hb{}$}{above}
	}
	\caption{An observed trace $\tr_{\two}$.}
	\figlabel{hb-misses-race}
\end{subfigure}%
\begin{subfigure}{.3\textwidth}
  \centering
	\execution{2}{
		\figev{2}{\acq(\lk)}
		\figev{2}{\rel(\lk)}
		\figev{1}{\mathbf{\wt(x)}}
		\figev{2}{\mathbf{\wt(x)}}
	}
	\vspace{0.2in}
  \caption{A witness $\tr_{\two}^{*}$ of a race in $\tr_{\two}$.}
  \figlabel{ZR-race-missed-by-HB}
\end{subfigure}
\begin{subfigure}{.3\textwidth}
  \centering
  \vspace{0.1in}
	\execution{2}{
		\figev{1}{\mathbf{\wt(x)}}
		\figev{1}{\acq(\lk)}
		\figev{1}{\wt(x)}
		\figev{1}{\rel(\lk)}
		\figev{2}{\acq(\lk)}
		\figev{2}{\mathbf{\wt(x)}}
		\figev{2}{\rel(\lk)}
	}
	\caption{An observed trace $\tr_{\one}$.}
	\figlabel{predictable-race-intro}
\end{subfigure}
\caption{
(\ref{fig:hb-misses-race}) shows a trace $\tr_\two$ with a sync-preserving race $(e_1, e_6)$ missed by $\hb{}$.
(\ref{fig:ZR-race-missed-by-HB} ) shows a witness that exposes the race.
(\ref{fig:predictable-race-intro}) shows a sync-preserving race $(e_1,e_6)$ that is non-consecutive, due to the intermediate event $e_3$.
}
\figlabel{compare-hb}
\end{figure}
Although happens-before only detects races whose exposure preserves the ordering of synchronization primitives, it can still miss simple races that adhere to this pattern.
For example, consider the trace $\tr_{\two}$ shown in \figref{hb-misses-race}. 
Let us name the events of this trace based on the order in which they appear in the trace; thus, $e_i$ denotes the $i^\text{th}$ event of the trace. 
Here, the partial order happens-before orders the 
first $\wt(x)$ (event $e_1$) and the last $\wt(x)$ (event $e_6$), and therefore, does not detect any race in this execution. However events $e_1$ and $e_6$ are in race. This can be exposed by the alternate execution shown in \figref{ZR-race-missed-by-HB}, which is obtained by dropping the critical section of lock $\lk$ performed by thread $t_1$. 
Notice that the order of synchronization events (namely, $\acq(\lk)$ and $\rel(\lk)$ events) \emph{that appear in the trace} of \figref{ZR-race-missed-by-HB}, are in the same order as in the trace of \figref{hb-misses-race}, and hence this is a sync-preserving execution.
Thus, the notion of sync-preservation captures races beyond standard happens-before races.

Another important limitation of happens-before and virtually all partial-order methods~\cite{wcp2017,shb2018,Roemer18,Roemer20,cp2012} is highlighted in \figref{predictable-race-intro}.
The trace $\tr_\one$ has a race between $e_1$ and $e_6$, both conflicting on variable $x$.
Notice, however, that the intermediate event $e_3$ also accesses $x$, but is not in race with either $e_1$ or $e_6$.
Partial-order methods for race prediction are limited to capturing races \emph{only between successive} conflicting accesses\footnote{When the earlier access is a read instead of a write, this statement is true \emph{per thread}.}.
Hence, distant races that are interjected with intermediate conflicting but non-racy events, are missed by such methods.
On the other hand, sync-preservation is not bound to such limitations: $(e_1, e_6)$ is characterized as a race under this criterion, regardless of the intermediate, non-racy $e_3$, and is exposed by a witness that omits the critical section on lock $\lk$ in the thread $t_1$.

\Paragraph{Our Contributions.}
Motivated by he above observations, we make the following contributions.
\begin{compactenum}
\item 
We introduce the novel notion of \emph{sync(hronization)-preserving data races.} 
This is a \emph{sound} notion of predictable races, and it \emph{strictly subsumes} the standard notion of happens-before races.
Moreover, it characterizes races between events that can be arbitrarily far apart in the input trace, as opposed to happens-before and other partial-order methods that only characterize races between \emph{successive} conflicting accesses.
Our notion is applicable to all concurrency settings, and interestingly, it is also \emph{complete} for systems with synchronization-deterministic concurrency~\cite{Bocchino09,Cui15,Aguado18,Zhao19}.
\item 
We develop an efficient, single-pass, nearly linear time algorithm $\ZeroRevAlgo$
that, given a trace $\tr$, detects whether $\tr$ contains a sync-preserving race.
In fact, our algorithm  soundly reports \emph{all} events $e_2$ which are in a sync-preserving race with an event $e_1$ that appears earlier in $\tr$.
Given $\NumEvents$ events in $\tr$, our algorithm spends $\Otilde(N)$ time, where $\Otilde$ hides factors poly-logarithmic in $\NumEvents$, when other parameters of the input (e.g., number of threads) are $\Otilde(1)$.
\item 
Although our algorithm performs a single pass of the trace, in the worst case, it might use space that is nearly linear in the length of the trace, i.e., $\Otilde(\NumEvents)$ space.
Hence follows a natural question: is there an efficient algorithm for sync-preserving race prediction that uses considerably less space?
We answer this question in negative, by showing that  \emph{any} single-pass algorithm for detecting even a single sync-preserving race must use nearly linear space.
Hence, our algorithm $\ZeroRevAlgo$ has nearly \emph{optimal} performance in both time and space.
\item 
We next study the complexity of race prediction with respect to the number of synchronization reversals that might occur when constructing a witness that exposes the race.
In the case of synchronization via locks, this number corresponds to the number of critical sections whose order is reversed in the witness trace.
We prove that the problem of predicting races which can be witnessed by \emph{a single} reversal (of two critical sections) is $\NP$-complete and even $\W{1}$-hard when parameterized by the number of threads.
Thus, sync-preservation characterizes \emph{exactly} the tractability boundary of race prediction, and our algorithm is nearly \emph{optimal} for the tractable side.
Moreover, our result shows that \emph{any level} of synchronization suffices to make the problem of race prediction as hard as in the general case.
\item 
Finally, we have implemented our race prediction algorithm $\ZeroRevAlgo$ and evaluated its performance on standard benchmarks.
Our results show that sync-preservation characterizes many races that are missed by state-of-the-art methods,
and $\ZeroRevAlgo$ detects them efficiently.
\end{compactenum}

%!TEX root = main.tex

\section{Preliminaries}\seclabel{prelim}

In this section we establish notation useful throughout of the paper.
The exposition follows other related works in the literature.

\begin{myparagraph}{Traces and events}
Our objective is to develop a dynamic analysis technique
which works over execution traces, or simply \emph{traces} of concurrent programs.
We work with the sequential consistency memory model.
In this setting, traces are sequences of events.
We will use $\tr, \tr', \ldots, \tr_1,\tr_2, \ldots$ to denote traces.
Every event of $\tr$ can be represented as a tuple $e = \ev{i, \thread, \op}$,
where $i$ is a unique identifier of $e$ in $\tr$, 
$t$ is the thread that performs $e$
and $\op$ is the operation performed in the event $e$.
We often omit the unique identifier of such a tuple 
and simply write $e = \ev{\thread, \op}$.
We use $\ThreadOf{e}$ and $\OpOf{e}$ to denote the thread
performing $e$ and the operation performed by $e$.
An operation can be one of read from or write to 
a shared memory location or \emph{variable} $x$,
denoted $\rd(x)$ and $\wt(x)$, 
and acquisition or release of a lock $\lk$, denoted $\acq(\lk)$ or $\rel(\lk)$.
%and fork or join of a thread $u$, denoted $\fork(u)$ or $\join(u)$.
Forks and joins can be naturally handled, but we avoid introducing them here for notational convenience.
We denote by $\events{\tr}$ the set of events in a trace $\tr$.
We use $\threads{\tr}$, $\vars{\tr}$ and $\locks{\tr}$
to denote respectively the threads, variables and locks
 that appear in $\tr$.
%For a variable $x \in \vars{\tr}$, we write $\writes{\tr}(x)$ (resp. $\reads{\tr}(x)$) 
%to denote the set of events $e \in \events{\tr}$ such that $\OpOf{e} = \wt(x)$
%(resp. $\OpOf{e} = \rd(x)$).
Likewise, we use $\acquires{\tr}(\lk)$ and $\releases{\tr}(\lk)$
to denote the set of acquire and release events of $\tr$ on lock $\lk \in \locks{\tr}$.

We require that traces obey lock semantics.
In particular, every lock $\lk$ is released by a thread $\thread$
only if there is an earlier matching acquire event by the same thread
$\thread$, and that each such lock is held by at most one
thread at a time.
Formally, let $\proj{\tr}{\lk}$ denote the projection of $\tr$
to the set of events $\acquires{\tr}(\lk) \cup \releases{\tr}(\lk)$.
We require that for every lock $\lk$, the sequence $\proj{\tr}{\lk}$
is a prefix of some sequence that belongs to the language
of the regular expression 
$\big(\sum\limits_{\thread \in \threads{\tr}} \ev{\thread, \acq(\lk)} \cdot \ev{\thread, \rel(\lk)}\big)^*$.

For an acquire event $e$, we use $\match{\tr}(e)$ to denote the 
matching release event of $e$ if one exists (and $\bot$ otherwise).
Similarly, for a release event $e$, $\match{\tr}(e)$ is the matching acquire
of $e$ on the same lock.
% For an acquire event $e$, the critical section
% protected by $e$ is $\crit{\tr}(e) = \setpred{f}{e \tho{\tr} f \tho{\tr} \match{\tr}(e)}$
% if $\match{\tr}(e)$ exists, and $\crit{\tr}(e) = \setpred{f}{e \tho{\tr} f}$ otherwise.
% For a release event $e$, we have $\crit{\tr}(e) = \crit{\tr}(\match{\tr}(e))$.
For an acquire event $e$, the critical section
protected by $e$, denoted $\crit{\tr}(e)$, is the set of
events $e'$ such that $\ThreadOf{e'} = \ThreadOf{e}$
and $e'$ occurs after $e$ and before the matching release $\match{\tr}(e)$ (if it exists) in $\tr$.
For a release event $e$, we have $\crit{\tr}(e) = \crit{\tr}(\match{\tr}(e))$.
\end{myparagraph}

\begin{myparagraph}{Orders on traces}
A partial order $\ord{\tr}{P}$ defined over a trace $\tr$
is a reflexive, anti-symmetric and transitive binary relation
on $\events{\tr}$; the symbol $\mathsf{P}$ is an optional identifier for 
the partial order.
We write $e_1 \ord{\tr}{P} e_2$ to denote $(e_1, e_2) \in \ord{\tr}{P}$,
where $e_1, e_2 \in \events{\tr}$.
For a partial order $\ord{\tr}{P}$, we use $\strictord{\tr}{P}$
to denote the strict order $\ord{\tr}{P} \setminus \setpred{(e,e)}{e \in \events{\tr}}$.
We write $e_1 \notord{\tr}{P} e_2$ to denote that $(e_1, e_2) \not\in \ord{\tr}{P}$.
Events $e_1, e_2 \in \events{\tr}$ are said to be \emph{unordered} by $\ord{\tr}{P}$,
denoted $\Unordered{e_1}{\tr}{P}{e_2}$ if $e_1 \notord{\tr}{P} e_2$ and $e_2 \notord{\tr}{P} e_1$;
otherwise, we write $\Ordered{e_1}{\tr}{P}{e_2}$, denoting that
$e_1$ and $e_2$ are ordered by $\ord{\tr}{P}$ in one or the other way.
When $\tr$ is clear from context, we will use
$\ord{}{P}$, $\strictord{}{P}$, $\notord{}{P}$, $\Unordered{}{}{P}{}$ and $\Ordered{}{}{P}{}$
instead of respectively $\ord{\tr}{P}$, $\strictord{\tr}{P}$, 
$\notord{\tr}{P}$, $\Unordered{}{\tr}{P}{}$ and $\Ordered{}{\tr}{P}{}$.
For a partial order $\ord{\tr}{P}$, 
a set $S \subseteq \events{\tr}$ is said to
be \emph{downward-closed with respect to} $\ord{\tr}{P}$
if for every $e, e' \in \events{\tr}$, if $e \ord{\tr}{P} e'$
and $e' \in S$, then $e \in S$.

The \emph{trace-order} $\trord{\tr}$ defined by $\tr$ is the total order on
$\events{\tr}$ imposed by the sequence $\tr$,
i.e., $e_1 \trord{\tr} e_2$ iff the event $e_1$ occurs before $e_2$ in $\tr$.
The \emph{thread-order} (or \emph{program-order}) $\tho{\tr}$ of $\tr$
is the partial order on $\events{\tr}$ that orders events in the same thread:
for two events $e_1, e_2\in \events{\tr}$,
$e_1 \tho{\tr} e_2$ iff $e_1 \trord{\tr} e_2$ and $\ThreadOf{e_1} = \ThreadOf{e_2}$.
\begin{comment}
for two events $e_1 \trord{\tr} e_2$, if one of the following hold,
then $e_1 \tho{\tr} e_2$:
\begin{enumerate*}[label=(\alph*)]
\item $\ThreadOf{e_1} = \ThreadOf{e_2}$, or
\item there is a thread $u\in\threads{\tr}$ such that
either $\OpOf{e_1} = \fork(u)$ and $\ThreadOf{e_2} = u$,
or $\ThreadOf{e_1} = \join(u)$ and $\OpOf{e_2} = \join(u)$.
For $e \in \events{\tr}$, we use $\prev{\tr}(e)$ to denote the
last event $e'$ in $\tr$ such that $e' \tho{\tr} e$;
if no such event exists, then $\prev{\tr}(e) = \bot$.
\ucomment{Change this for $\join$ events}
\end{enumerate*}
\end{comment}
\end{myparagraph}

\begin{myparagraph}{Conflicting events and data races}
Let $\tr$ be a trace.
Two events $e_1, e_2 \in \events{\tr}$ 
are said to be \emph{conflicting}, denoted $e_1 \conf e_2$,
% if $\Unordered{e_1}{\tr}{\mathsf{TO}}{e_2}$, 
if $\ThreadOf{e_1} \neq \ThreadOf{e_2}$, 
and there is a common variable $x \in \vars{\tr}$ such that
$\OpOf{e_1}, \OpOf{e_2} \in \set{\rd(x), \wt(x)}$
and at least one of $\OpOf{e_1}$ and $\OpOf{e_2}$ is $\wt(x)$.
Let $\rho$ be a trace with $\events{\rho} \subseteq \events{\tr}$.
An event $e \in \events{\tr}$ is said to be $\tr$-\emph{enabled}
in $\rho$ if $e \not\in \events{\rho}$ and
for all events $e' \in \events{\tr}$ such that $e' \stricttho{\tr} e$, we have $e' \in \events{\rho}$.
A pair of conflicting events $(e_1, e_2)$ in $\tr$ 
is said to be a data race of $\tr$ if $\tr$ has a prefix
$\tr'$ such that both $e_1$ and $e_2$ are
$\tr$-enabled in $\tr'$.
The trace $\tr$ is said to have a data race if there is a pair
of conflicting events $(e_1, e_2)$ in $\tr$ that constitutes a data race of $\tr$.
\end{myparagraph}

\begin{example}
%!TEX root = ../main.tex

\begin{figure}[t]
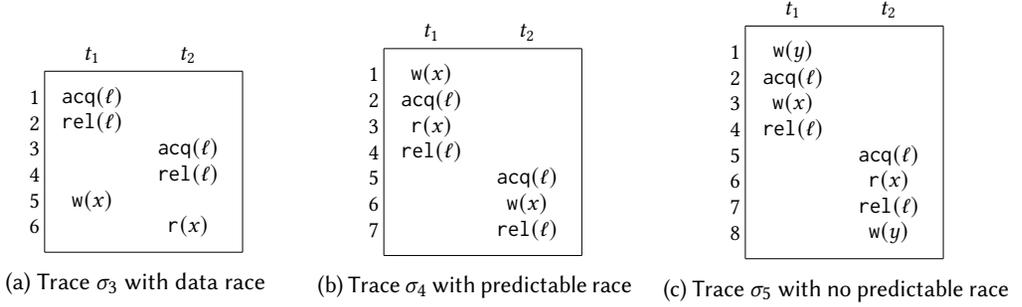

\centering
\begin{subfigure}{.3\textwidth}
	\centering
	\vspace{0.2in}
	\execution{2}{
		\figev{1}{\acq(\lk)}
		\figev{1}{\rel(\lk)}
		\figev{2}{\acq(\lk)}
		\figev{2}{\rel(\lk)}
		\figev{1}{\wt(x)}
		\figev{2}{\rd(x)}
	}
	\caption{Trace $\tr_\exampleone$ with data race}
	\figlabel{data-race}
\end{subfigure}%
\hspace{0.1in}
\begin{subfigure}{.3\textwidth}
  \centering
  \vspace{0.1in}
	\execution{2}{
		\figev{1}{\wt(x)}
		\figev{1}{\acq(\lk)}
		\figev{1}{\rd(x)}
		\figev{1}{\rel(\lk)}
		\figev{2}{\acq(\lk)}
		\figev{2}{\wt(x)}
		\figev{2}{\rel(\lk)}
	}
	\caption{Trace $\tr_\exampletwo$ with predictable race}
	\figlabel{predictable-race}
\end{subfigure}
\hspace{0.1in}
\begin{subfigure}{.33\textwidth}
  \centering
	\execution{2}{
		\figev{1}{\wt(y)}
		\figev{1}{\acq(\lk)}
		\figev{1}{\wt(x)}
		\figev{1}{\rel(\lk)}
		\figev{2}{\acq(\lk)}
		\figev{2}{\rd(x)}
		\figev{2}{\rel(\lk)}
		\figev{2}{\wt(y)}
	}
   \caption{Trace $\tr_\examplethree$ with no predictable race}
  \figlabel{no-predictable-race}
\end{subfigure}
\caption{Traces, data races and predictable data races}
\figlabel{example-traces}
\end{figure}
Consider the trace $\tr_\exampleone$ in~\figref{data-race}.
% The first event in $\tr_\exampleone$ is $e_1 = \ev{t_1, \acq(\lk)}$.
The set of events of $\tr_\exampleone$ is 
$\events{\tr_\exampleone} = \set{e_1, e_2, \ldots, e_6}$,
$\threads{\tr} = \set{t_1, t_2}$, $\vars{\tr} = \set{x}$ and $\locks{\tr} = \set{\lk}$.
For the event $e_1 = \ev{t_1, \acq(\lk)}$, we have $\ThreadOf{e_1} = t_1$ and $\OpOf{e_1} = \acq(\lk)$.
The trace order of this trace is $\trord{\tr_\exampleone} = \setpred{(e_i, e_j)}{i \leq j}$
and the thread-order is $\tho{\tr_\exampleone} = \set{(e_1, e_2), (e_1, e_5), (e_2, e_5), (e_3, e_4), (e_3, e_6), (e_4, e_6)}$.
Events $e_5$ and $e_6$ conflict because they access the same variable $x$
and are performed by different threads.
For the prefix trace $\tr'_\exampleone = e_1{\cdot}e_2{\cdot}e_3{\cdot}e_4$,
both $e_5$ and $e_6$ are $\tr_\exampleone$-enabled in $\tr'_\exampleone$.
Thus, $(e_5, e_6)$ constitutes a data race of $\tr_\exampleone$.
\end{example}

\begin{myparagraph}{Correct reorderings}
Execution traces of concurrent programs are sensitive to thread
scheduling, and looking for a trace with a specific pattern
is like searching for a needle in a haystack.
In terms of data race detection, this means that a dynamic
analysis that looks for executions with enabled conflicting events 
(data races) is likely to miss many data races that might have otherwise been captured
in alternate executions of the same program that arise due to 
slightly different thread scheduling.
% This means that the likelihood of the occurrence of a data race, 
% the way we defined above (as a pair of \emph{consecutive} conflicting events)
% is very low for sufficiently large traces, and limiting a data race
% detection technique to looking for data races in observed executions
% will likely miss many data races that might have otherwise been captured
% in alternate executions of the same program.
The notion of data race \emph{prediction} attempts to alleviate this problem
by capturing a more robust notion of data races.
The idea here is to infer data races that might occur in alternate reorderings
of an observed trace, thereby detecting data races beyond those
in just the execution that was observed.
The set of allowable reorderings of an observed trace $\tr$ is defined in a manner
that ensures that data races can be detected agnostic of the program
that generated $\tr$ in the first place.
Such a notion is captured by a \emph{correct reordering} which we define next.

For a trace $\tr$ and a read event $e$,
we use $\lw{\tr}(e)$ to denote the write event observed by $e$.
That is, $e' = \lw{\tr}(e)$ is the last (according to the trace order $\trord{\tr}$) 
write event $e'$ of $\tr$ such that $e$ and $e'$ access the same variable and $e' \trord{\tr} e$; 
if no such $e'$ exists, then we write $\lw{\tr}(e) = \bot$.

Given the above notation,
a trace $\rho$ is said to be a correct reordering of trace $\tr$
if 
\begin{compactenum}[(a)]
\item $\events{\rho} \subseteq \events{\tr}$
\item $\events{\rho}$ is downward closed with respect to $\tho{\tr}$, 
and further $\tho{\rho} \subseteq \tho{\tr}$,
\item for every read event $e \in \events{\rho}$, 
% if there is an event $e' \in\events{\rho} \setminus \set{e}$ 
% such that $e \tho{\rho} e'$, then
$\lw{\rho}(e) = \lw{\tr}(e)$.
\end{compactenum}
The above definition ensures that if $\rho$ is a correct reordering of $\tr$,
then every program that generates the execution trace $\tr$ also generates $\rho$.
This is because $\rho$ preserves both intra-thread ordering, as well as
the values read by every read occurring in $\rho$, thereby preserving
any control flow that might have been taken by $\tr$.
This style of formalizing alternative executions based on
semantics of concurrent objects was popularized by~\cite{HerlihyWing90}
and by prior race detection works~\cite{maxcausalmodels,Said2011}.
Our definition of correct reordering has been derived from~\cite{cp2012},
which has subsequently also been used in the 
literature~\cite{wcp2017,shb2018,PavlogiannisPOPL20,Mathur20,Roemer18,Bond2019}.
\end{myparagraph}

\begin{myparagraph}{Data race prediction}
Armed with the notion of correct reorderings, 
we can now define a more robust notion of data races.
% 
% \Paragraph{Predictable Data Races.}
%\deflabel{predictable-race}
A pair of conflicting events $(e_1, e_2)$ in $\tr$ is said to be a \emph{predictable}
data race of $\tr$ if there is a correct reordering $\rho$ of $\tr$
such that $e_1, e_2$ are $\tr$-enabled in $\rho$.
We remark that a pair of conflicting events
$(e_1, e_2)$ in trace $\tr$ may not be a data race of $\tr$,
but nevertheless may still be a \emph{predictable} data race of $\tr$.
\begin{example}
Consider the trace $\tr_\exampletwo$ in~\figref{predictable-race}.
Observe that there is no prefix of $\tr_\exampletwo$ in which both $e_1$
and $e_6$ are enabled.
However, $(e_1, e_6)$ is a predictable race of $\tr_\exampletwo$
that is witnessed by the singleton correct reordering $\tr_\exampletwo^\cre = e_5$
in which both $e_1$ and $e_6$ are enabled;
$\tr_\exampletwo^\cre$ is both downward closed with respect to,
and respects $\tho{\tr_\exampletwo}$.
Further, it has no read events and thus vacuously
every read observes the same last write as in $\tr_\exampletwo$.
The other pair of conflicting events in $\tr_\exampletwo$, namely
$(e_3, e_6)$, however, is not a predictable race.
These events are protected
by a common lock, and there is no correct reordering
in which $e_3$ and $e_6$ are simultaneously enabled --- any attempt
at doing so will lead to overlapping
critical sections on $\lk$, thereby violating lock semantics.
\end{example}
\begin{example}
Now, consider $\tr_\examplethree$ in~\figref{no-predictable-race}.
Here, the conflicting pair $(e_3, e_6)$ cannot be a 
predictable race as in the case of $\tr_\exampletwo$--- the lock $\lk$ protects both $e_3$ and $e_6$.
Now consider the other conflicting pair $(e_1, e_8)$.
Let $\rho$ be a correct reordering of $\tr_\examplethree$
in which $e_8$ is enabled.
We must have $e_6 \in \events{\rho}$ 
($\rho$ must be $\tho{\tr_\examplethree}$-downward closed)
and further $e_3 \in \events{\rho}$ (as $e_3 = \lw{\tr_\examplethree}(e_6) = \lw{\rho}(e_6)$).
Clearly, $e_1$ cannot be enabled in any such trace $\rho$,
and thus, the trace $\tr_\examplethree$ has no predictable data race. 
\end{example}

The central theme of race prediction is to solve the problem below.
% In this presentation, we will be focusing on the prediction problem stated below.
% A trace $\tr$ is said to have a predictable data race if there is
% a pair of conflicting events which form a predictable data race of $\tr$.
\begin{problem}[Data Race Prediction]
\problabel{race-prediction}
Given a trace $\tr$, determine if
$\tr$ has a predictable data race.
\end{problem}

\myparagraph{A note on soundness}{
	We say that an algorithm for data race prediction is \emph{sound} if
	whenever the algorithm reports a YES answer, then the
	given trace has a predictable data race.
	Likewise, an algorithm is complete if the algorithm reports YES whenever
	the input trace has a data race.
	Our convention for this nomenclature ensures that
	no false positives are reported by a \emph{sound} algorithm~\cite{soundness-dynamic-analysis}
	and is consistent with prior work
	on data race prediction~\cite{cp2012,wcp2017,Roemer18,Bond2019,PavlogiannisPOPL20}.
	Soundness is often a desirable property for dynamic race predictors
	for widespread adoption~\cite{racerdx2019}.	
}

% We remark that th above problem is, in general, intractable~\cite{Mathur20}.
%\begin{theorem}[\cite{Mathur20}]
%The data race prediction problem is $\NP$-hard.
%% More precisely, it is $\W{1}$-hard when parameterized by the number of threads in the input trace.
%\end{theorem}
\end{myparagraph}

%!TEX root = main.tex

\subsection{Synchronization-Preserving Data Races}
\seclabel{zero_reversals}

In general, the problem of data race prediction is intractable~\cite{Mathur20},
and a sound and complete
algorithm for data race prediction is unlikely to scale beyond programs of even moderate size.
A recent trend in predictive analysis for race detection instead, 
aims to develop techniques that are sound but incomplete, with successively better
prediction power (ability to report more data races) than previous techniques~\cite{cp2012,wcp2017,Roemer18,PavlogiannisPOPL20,Bond2019}.
Most of these techniques are either based on partial orders~\cite{Pozniansky:2003:EOD:966049.781529,cp2012,wcp2017}
or use graph-based algorithms~\cite{Roemer18,PavlogiannisPOPL20}.
In this paper, we characterize a class of predictable data races,
called \emph{sync(hronization)-preserving} races, which we define shortly.
We will later (\secref{detection}) present an algorithm
that reports a race iff the input trace has a 
sync-preserving race.
Since sync-preserving races are predictable races,
our algorithm will be sound for race prediction.

\Paragraph{Sync-preserving correct reordering.}
A correct reordering of a trace is called
\emph{sync(hronization)-preserving} if it does not
reverse the order of synchronization constructs;
in our formalism, traces use locks as synchronization primitives
to enforce mutual exclusion.
Formally, a correct reordering $\rho$ of a given trace $\tr$
is \emph{sync-preserving} with respect to
$\tr$ if for every lock $\lk$ and for any two acquire events 
$e_1, e_2 \in \acquires{\rho}(\lk)$,
we have $e_1 \trord{\rho} e_2$ iff $e_1 \trord{\tr} e_2$.
In other words, the order of two critical sections on the same
lock is the same in $\tr$ and  $\rho$.
Let us illustrate this notion on an example.
\begin{example}
Consider trace $\tr_\examplefour$ in~\figref{tr4}.
This trace has 3 critical sections on lock $\lk$.
Now consider the correct reordering $\tr^{\cre}_\examplefour$ 
(\figref{tr4cr}) of $\tr_\examplefour$.
Here, the critical section in thread $t_1$ is not present.
But, nevertheless, the order amongst the 
remaining critical sections on $\lk$ (in threads $t_2$ and $t_3$)
is the same as in $\tr_\examplefour$,
making $\tr^{\cre}_\examplefour$ a sync-preserving correct reordering
of $\tr_\examplefour$.
% Thus, a $\zeroreversal$ correct reordering of may not 
% contain all critical sections of $\tr$, 
% and is required to preserve the ordering amongst only
% those critical sections (on a give lock) that appear in $\rho$.
This example also demonstrates that the order of read and write 
events may be different in a trace and its sync-preserving correct 
reordering (as in \figref{zeroreversal-cr}).
\end{example}

%!TEX root = ../main.tex

\begin{figure}[t]
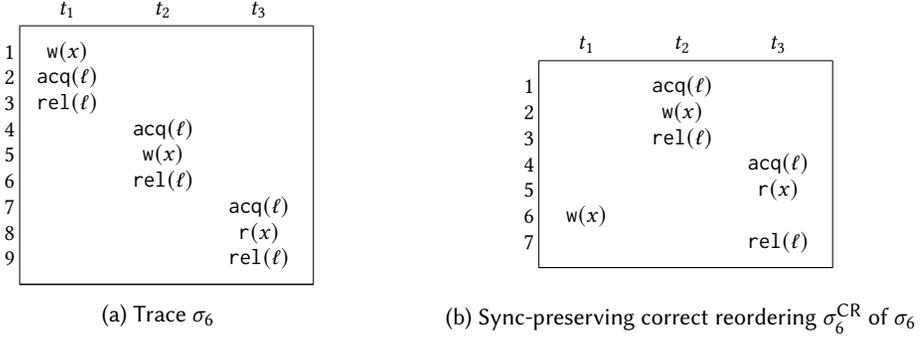

\centering
\begin{subfigure}{.5\textwidth}
	\centering
	\execution{3}{
		\figev{1}{\wt(x)}
		\figev{1}{\acq(\lk)}
		\figev{1}{\rel(\lk)}
		\figev{2}{\acq(\lk)}
		\figev{2}{\wt(x)}
		\figev{2}{\rel(\lk)}
		\figev{3}{\acq(\lk)}
		\figev{3}{\rd(x)}
		\figev{3}{\rel(\lk)}
	}
	\caption{Trace $\tr_\examplefour$}
	\figlabel{tr4}
\end{subfigure}%
% \hspace{0.1in}
\begin{subfigure}{.5\textwidth}
  \centering
  \vspace{0.2in}
		\execution{3}{
		\figev{2}{\acq(\lk)}
		\figev{2}{\wt(x)}
		\figev{2}{\rel(\lk)}
		\figev{3}{\acq(\lk)}
		\figev{3}{\rd(x)}
		\figev{1}{\wt(x)}
		\figev{3}{\rel(\lk)}
	}
	\vspace{0.1in}
	\caption{Sync-preserving correct reordering $\tr^{\cre}_{\examplefour}$ of $\tr_{\examplefour}$}
	\figlabel{tr4cr}
\end{subfigure}
\caption{Sync-preserving correct reordering and sync-preserving races}
\figlabel{zeroreversal-cr}
\end{figure}

A pair of conflicting events $(e_1, e_2)$ of a trace $\tr$ 
is said to be a \emph{sync(hronization)-preserving race} of $\tr$
if there is a sync-preserving correct reordering $\rho$ of $\tr$ 
in which $e_1$ and $e_2$ are $\tr$-enabled.

\begin{example}
Let us again consider traces from~\figref{zeroreversal-cr}.
Events $e_1$ and $e_8$ in $\tr_{\examplefour}$ (\figref{tr4})
correspond respectively to events $e_6$ and $e_7$ in
$\tr^\cre_{\examplefour}$ (\figref{tr4cr}).
These two events are $\tr_{\examplefour}$-enabled in the prefix
$\rho = e_1{\cdot}e_2{\cdot}e_3{\cdot}e_4{\cdot}e_5$
of $\tr^\cre_{\examplefour}$.
As a result, $(e_1, e_8)$ is a sync-preserving race of $\tr_\examplefour$.
Likewise, $(e_1, e_4)$ is also a
sync-preserving race of $\tr_\examplefour$ witnessed by the
singleton sync-preserving correct reordering $\rho' = \ev{t_2, \acq(\lk)}$, in which
both $e_1$ and $e_4$ are enabled.
\end{example}

In this paper we present a linear time algorithm for the following decision problem,
giving a \emph{sound} algorithm for~\probref{race-prediction}.
\begin{problem}[Sync-Preserving Race Prediction]
\problabel{zr-race-prediction}
Given trace $\tr$, determine if there is a pair of
conflicting events $(e_1, e_2)$ in $\tr$ such that
$(e_1, e_2)$ is a sync-preserving data race of $\tr$.
\end{problem}

\Paragraph{Comparison with other approaches.}
Here we briefly compare sync-preserving races with other approaches in the literature for sound dynamic race prediction.
Races reported using the famous \emph{happens-before} \textsf{HB} partial order~\cite{Pozniansky:2003:EOD:966049.781529}, and its extension to \emph{schedulable-happens-before \textsf{SHB}}~\cite{shb2018} 
are strictly subsumed by this notion.
That is, these techniques only compute sync-preserving races, but can also miss simple cases of sync-preservation, as already illustrated in the examples of \figref{compare-hb}.
The \emph{causally precedes} (\textsf{CP}) partial order~\cite{cp2012},  and its extension to
the \emph{weak causally precedes} (\textsf{WCP}) partial order~\cite{wcp2017} are capable of predicting races that reverse critical sections.
However, they are closed under composition with \textsf{HB}, and as such can miss even simple sync-preserving races, even on two-threaded traces.
The \emph{doesn't commute} \textsf{DC} partial order~\cite{Roemer18} is an unsound weakening to \textsf{WCP},
that further undergoes a vindication phase to filter out unsound reports.
Nevertheless, \textsf{DC} is somewhat similar to \textsf{WCP} and also misses sync-preserving races.
The recently introduced partial order \emph{strong-dependently-precedes} 
(\textsf{SDP})~\cite{Bond2019}, while claimed to be sound in that paper,
is, in fact, unsound. 
In \cref{sec:app_sdp}, we show a counter-example
to the soundness theorem of \textsf{SDP}, which we confirmed with the authors~\cite{sdpUnsound2020}.
The partial order \textsf{WDP}~\cite{Bond2019}, is a further unsound weakening of
\textsf{DC}, and and can miss sync-preserving races in the \emph{vindication} phase that it employs for ruling out false positives.
We further refer to \appref{comparison} for a few examples that illustrate the above comparison.

%!TEX root = main.tex

\section{Summary of Main Results}

Here we give an outline of the main results of this paper.
In later sections we present the details, i.e., algorithms, proofs and examples.
Due to limited space, some technical proofs are relegated to the appendix.
Our first result is an algorithm for dynamic prediction of sync-preserving races.
We show the following theorem.

\smallskip
\begin{restatable}{theorem}{thmzeroreversals}
\thmlabel{zero_reversals}
Sync-preserving race prediction is solvable in 
$O(\NumEvents\cdot \NumThreads^2 + \NumAcquires\cdot \NumVariables\cdot \NumThreads^3)$ time and 
$O(\NumEvents+ \NumThreads^3\cdot \NumVariables\cdot \NumLocks)$ space, for a trace $\tr$ with 
length $\NumEvents$, 
$\NumThreads$ threads,
$\NumAcquires$ acquires, and 
$\NumVariables$ variables.
\end{restatable}

In many settings the number of events $\NumEvents$  and number of acquires $\NumAcquires$ are the dominating parameters,
whereas the other parameters are much smaller, i.e., 
$\NumThreads,\NumVariables=\Otilde(1)$, where $\Otilde$ hides poly-logarithmic factors.
Hence, the complexity of our algorithm becomes $\Otilde(\NumEvents)$ 
for both time and space.
Our next result shows that a linear space complexity is essentially unavoidable when predicting sync-preserving races with one-pass streaming algorithms.

\smallskip
\begin{restatable}{theorem}{thmspacelowerbound}\label{thm:space_lowerbound}
Any one-pass algorithm  for sync-preserving race prediction on traces with $\geq 2$ threads, $\NumEvents$ events and $\Omega(\log \NumEvents)$ locks uses $\Omega(\NumEvents/\log^2 \NumEvents)$ space.
\end{restatable}

Clearly, any algorithm must spend linear time, while \cref{thm:space_lowerbound} shows that the algorithm must also use (nearly) linear space. 
As our algorithm uses $\Otilde(\NumEvents)$ time and space, it is optimal for both resources, modulo poly-logarithmic improvements.
Our next theorem shows a combined time-space lower bound for the problem, which highlights that reducing the space usage must lead to an increased running time, given that the algorithm is executed on the Turing Machine model.

\smallskip
\begin{restatable}{theorem}{thmproductlowerbound}\label{thm:product_lowerbound}
Consider the problem of sync-preserving race prediction on traces with $\geq 2$ threads, $\NumEvents$ events and $\Omega(\log \NumEvents)$ locks.
Consider any Turing Machine algorithm for the problem with time and space complexity $T(\NumEvents)$ and $S(\NumEvents)$, respectively.
Then we have $T(\NumEvents)\cdot S(\NumEvents)=\Omega(\NumEvents^2/\log^2\NumEvents)$.
\end{restatable}

Finally, we study the complexity of  general race prediction as a function of the number of reversals of synchronization operations.
Given our positive result in \cref{thm:zero_reversals}, can we relax our restriction of sync-preservation  while retaining a tractable definition of predictable races?
Our next theorem answers this question in negative.

\smallskip
\begin{restatable}{theorem}{thmwonehard}\label{thm:dichotomy}
Dynamic race prediction on traces with a single lock and two critical sections is $\W{1}$-hard parameterized by the number of threads.
\end{restatable}

Note that $\W{1}$-hardness implies $\NP$-hardness. The theorem has two important implications.
\begin{compactenum}
\item Any witness of predictable races in the setting of \cref{thm:dichotomy} is either a sync-preserving reordering, or reverses the order of a single pair of acquire events.
Thus, \cref{thm:zero_reversals} and \cref{thm:dichotomy} establish a \emph{tight dichotomy} on the tractability of the problem, based on the number of synchronization reversals:~the problem is as hard as in the general case for just $1$ reversal, while it is efficiently solvable for no reversals.
\item 
The general problem of dynamic race prediction was shown to be $\W{1}$-hard in~\cite{Mathur20}.
However, that proof requires traces with $\Omega(\NumEvents)$ critical sections, and hence it applies to traces that
essentially comprise of synchronization events entirely.
In contrast, the class of traces in \cref{thm:dichotomy} has the smallest level of synchronization possible, i.e., just a single lock and two critical sections on that lock.
Hence, \cref{thm:dichotomy} shows that \emph{any} amount of lock-based synchronization suffices to make the problem as hard as in the general case.
\end{compactenum}

Together, \cref{thm:zero_reversals}, \cref{thm:space_lowerbound} and \cref{thm:dichotomy} characterize \emph{exactly} the tractability boundary of race prediction, and show that our algorithm is \emph{time and space optimal} for the tractable side.

%!TEX root = main.tex

\section{Detecting Synchronization-Preserving Races}
\seclabel{detection}

In this section, we discuss our algorithm $\ZeroRevAlgo$ for detecting
sync-preserving data races. 
% Before presenting the algorithm in 
% This algorithm runs in linear time and detects races in a streaming fashion, 
% processing events as they are generated in the trace.
% The algorithm uses linear space in the worst case and we establish
% a linear space lower bound for any algorithm that detects $\zeroreversal$ races.
The complete algorithm is presented in~\secref{full-algo}.
The algorithm may appear complex at first glance and 
to make the exposition simple, we first present a high-level overview of the 
algorithm in~\secref{overview}.
In the overview, we highlight important observations and
 algorithmic insights for solving smaller subproblems of the main problem
of sync-preserving race prediction.
\secref{overview-given-pair} and~\secref{overview-existence}
present details of the algorithms for the smaller subproblems,
and pave the way for the final algorithm in~\secref{full-algo}.
In Section~\ref{subsec:space_lowerbound}, 
we present a matching space lowerbound result for detecting sync-preserving races,
thereby showing the optimality of our algorithm.

%!TEX root = main.tex

\subsection{Insights and Overview of the Algorithm}
\seclabel{overview}

Our algorithm, $\ZeroRevAlgo$, relies on several important observations that are 
crucial for detecting sync-preserving races in linear time.
In order to present these observations, it is helpful to 
define intermediate subproblems.
\begin{problem}[Sync-Preserving Race Prediction Given Pair]
\problabel{zr-race-prediction-given-pair}
Given a trace $\tr$ and a pair of conflicting events $(e_1, e_2)$
of $\tr$, determine if $(e_1, e_2)$ is a sync-preserving data race of $\tr$.
\end{problem}
\begin{problem}[Sync-Preserving Race Prediction Given Event and Thread]
\problabel{zr-race-prediction-given-event}
Given a trace $\tr$, an event $e$ in $\tr$ and a thread $t \neq \ThreadOf{e}$, 
check if there is an event $e' \trord{\tr} e$ such that $\ThreadOf{e'} = t$
and $(e', e)$ constitutes a sync-preserving race of $\tr$.
\end{problem}

Observe that a trace with $\NumEvents$ events can have $O(\NumEvents^2)$ conflicting
pairs of events. 
Thus, an algorithm for \probref{zr-race-prediction-given-pair}
that runs in time $O(T(\NumEvents))$ can be used to obtain an
algorithm for~\probref{zr-race-prediction-given-event}
(resp. \probref{zr-race-prediction}) that runs in time 
$O(\NumEvents\cdot T(\NumEvents))$  (resp. $O(\NumEvents^2\cdot T(\NumEvents))$)
by checking if every other event of the given thread $t$
that conflicts with $e$ is also in race with $e$
(resp. every conflicting pair of events is a race). 
We will, however, present algorithms for all three problems that run in $O(\NumEvents)$ time.

\subsubsection{Efficiently solving~\probref{zr-race-prediction-given-pair}}
Our first important observation towards a full fledged solution 
to~\probref{zr-race-prediction-given-pair}
is that when checking for the
existence of a sync-preserving reordering
(of a trace $\tr$) that witnesses a race on a given pair $(e_1, e_2)$, 
it, suffices to only search for those reorderings $\rho$ of $\tr$ which 
impose the same order as in $\tr$, on \emph{all of its events},
and not just on the critical sections.
We formalize this in~\lemref{trace-order-suffices}. % (proof in~\appref{app_detection}).
% \begin{lemma}
\begin{restatable}{lemma}{lemtraceordersuffices}
\lemlabel{trace-order-suffices}
If $(e_1, e_2)$ is a sync-preserving race of $\tr$, then there is a
correct reordering $\rho$ of $\tr$ such that
both $e_1, e_2$ are $\tr$-enabled in $\rho$ and
$\trord{\rho} \subseteq \trord{\tr}$.
% \end{lemma}
\end{restatable}
The implication of~\lemref{trace-order-suffices} is the following.
When we are searching for
a correct reordering $\rho$ of $\tr$ that witnesses a race $(e_1, e_2)$,
and if we already have access to the set of events 
$S \subseteq \events{\tr}$ of a candidate reordering $\rho$,
a simple check suffices --- does $S$ form a correct reordering of $\tr$
when linearized according to $\trord{\tr}$?
In other words, we do not need to enumerate and check all the
(exponentially many) permutations of events in $S$.
Thus, \probref{zr-race-prediction-given-pair} --- `search for a 
sync-preserving correct reordering $\rho$' --- reduces to a simpler 
problem --- `search for an appropriate set of events'.
Of course, not all sets $S \subseteq \events{\tr}$ of events can 
be linearized (according to $\trord{\tr}$)
to obtain a correct reordering of $\tr$.
At the very least, $S$ should satisfy some
sanity conditions which we outline next.

\begin{definition}[Thread-Order and Last-Write Closure]
\deflabel{too-closure}
Let $\tr$ be a trace.
A set $S \subseteq \events{\tr}$ is said to be
$(\tho{\tr}, \lw{\tr})$-closed if
\begin{enumerate*}[label=(\alph*)]
\item~\itmlabel{to-closure} $S$ is downward-closed with respect to $\tho{\tr}$, and
\item~\itmlabel{lw-closure} for any read event $r \in \events{\tr}$,
if $r \in S$ and if $\lw{\tr}(r)$ exists, 
then $\lw{\tr}(r) \in S$.
\end{enumerate*}\\
The $(\tho{\tr}, \lw{\tr})$-closure of a set $S \subseteq \events{\tr}$, denoted
$\TOOClosure{\tr}(S)$ is the smallest set $S' \subseteq \events{\tr}$ 
such that $S \subseteq S'$ and $S'$ is $(\tho{\tr}, \lw{\tr})$-closed.
\end{definition}
We remark that any correct reordering $\rho$ of $\tr$ that contains 
events in the set $S$ must also contain all the events in $\TOOClosure{\tr}(S)$.
%  ---
% if $\rho$ schedules an event $e \in S$, it must
% also schedule all prior events $e' \tho{\tr} e$ (in the same order), and
% similarly, a read event $r$ can only be present in $\rho$ if the corresponding
% last write $\lw{\tr}(r)$ is also present in $\rho$.

\begin{definition}[Sync-Preserving Closure]
\deflabel{zr-closure}
Let $\tr$ be a trace.
A set $S \subseteq \events{\tr}$ is said to be sync-preserving closed if
\begin{compactenum}[(a)]
\item~\itmlabel{too-closure} $S$ is $(\tho{\tr}, \lw{\tr})$-closed, and
\item~\itmlabel{crit-closure} for any two acquire events $a_1, a_2 \in \acquires{\tr}(\lk)$ with $a_1 \trord{\tr} a_2$, 
if both $a_1, a_2 \in S$, then $\match{\tr}(a_1) \in S$
\end{compactenum}
The sync-preserving closure of a set $S \subseteq \events{\tr}$, denoted
$\ZRClosure{\tr}(S)$ is the smallest set $S' \subseteq \events{\tr}$ 
such that $S \subseteq S'$ and $S'$ is sync-preserving closed.
\end{definition}
Intuitively, the set $S' = \ZRClosure{\tr}(S)$
captures the additional set of events that must be present in
any sync-preserving correct reordering $\rho$ of $\tr$
given that $\rho$ contains all events in $S$.
First, any correct reordering of $\tr$ containing $S$ will contain $\TOOClosure{\tr}(S)$
and thus $\TOOClosure{\tr}(S) \subseteq S'$
(Condition~\itmref{too-closure}).
Second, if a correct reordering $\rho$ is sync-preserving
and contains two acquires $a_1 \trord{\tr} a_2$ on the same lock $\lk$,
then we must also have $a_1 \trord{\rho} a_2$.
Then, in order to ensure well-formedness of $\rho$,
$\crit{\tr}(a_1)$ must also 
finish entirely before $a_2$ in $\rho$, and thus $\rho$ must 
contain $\match{\tr}(a_1)$ (Condition~\itmref{crit-closure}).

For two events $e_1, e_2\in \events{\tr}$, we define 
$$\ZRIdeal{\tr}(e_1, e_2) = \ZRClosure{\tr}(\set{\prev{\tr}(e_1)} \cup \set{\prev{\tr}(e_2)}).$$
Here, we use $\prev{\tr}(e)$ to denote the
last event $e'$ in $\tr$ such that $e' \tho{\tr} e$;
if no such event exists, then $\prev{\tr}(e) = \bot$
(and further we let $\set{\bot} = \emptyset$).
In essence, $\ZRIdeal{\tr}(e_1, e_2)$ contains the
necessary set of events that must be present in any sync-preserving correct
reordering that witnesses the race $(e_1, e_2)$.
We next show that, in fact, it is also a sufficient set of events,
given that it is disjoint from $\set{e_1, e_2}$.

\begin{restatable}{lemma}{zridealdisjointness}
% \begin{lemma}
\lemlabel{zrideal-disjointness}
$(e_1, e_2)$ is a sync-preserving race of $\tr$ iff $\set{e_1, e_2} \cap \ZRIdeal{\tr}(e_1, e_2) = \emptyset$.
% \end{lemma}
\end{restatable}
\lemref{zrideal-disjointness} gives us
a straightforward algorithm
for~\probref{zr-race-prediction-given-pair} --- compute
$I = \ZRIdeal{\tr}(e_1, e_2)$ and check if neither $e_1, e_2\not \in I$.
In~\secref{overview-given-pair} we show how to perform this computation in linear time.

\subsubsection{Efficiently Solving~\probref{zr-race-prediction-given-event}}
As noted before, a linear time
algorithm for~\probref{zr-race-prediction-given-pair}
guarantees a \emph{quadratic} time algorithm for~\probref{zr-race-prediction-given-event}.
In order to design a more efficient \emph{linear time}
algorithm, we will exploit \emph{monotonicity} of
$\ZRIdeal{\tr}(\cdot, \cdot)$, which we formalize next.
\begin{restatable}{lemma}{lemmonotonicity}
% \begin{lemma}
\lemlabel{monotonicity}
Let $\tr$ be a trace and let $e_1, e_2, e'_1, e'_2 \in \events{\tr}$
such that $e_1 \tho{\tr} e'_1$ and $e_2 \tho{\tr} e'_2$.
Then, $\ZRIdeal{\tr}(e_1, e_2) \subseteq \ZRIdeal{\tr}(e'_1, e'_2)$.
% \end{lemma}
\end{restatable}
Our linear time algorithm for~\probref{zr-race-prediction-given-event}
 exploits~\lemref{monotonicity} as follows.
Suppose we are checking if a given event $e$ in the given trace $\tr$
is in sync-preserving race with some earlier conflicting event of thread $t$.
To accomplish this, we can scan $\tr$ and enumerate the list $L$
of events that belong to $t$ and conflict with $e$.
When checking for a race with the first such event $e'_\text{first}$, 
we compute $I_\text{first} = \ZRIdeal{\tr}(e'_\text{first}, e)$.
If a race is found, we are done.
Otherwise, we analyze the next event $e'_\text{next}$ in $L$
and compute $I_\text{next} = \ZRIdeal{\tr}(e'_\text{next}, e)$.
Here~\lemref{monotonicity} ensures that $I_\text{first} \subseteq I_\text{next}$.
Our algorithm exploits this observation 
by computing the latter set $I_\text{next}$
incrementally, spending time that is proportional only
to the number of \emph{extra} events (i.e., events in $I_\text{next}\setminus I_\text{first}$).
This principle is applied repeatedly to subsequent events of $L$,
giving us an overall linear time algorithm (\secref{overview-existence}).

\subsubsection{Efficiently solving~\probref{zr-race-prediction}}
A final ingredient in our incremental linear time algorithm
for~\probref{zr-race-prediction} is the 
following observation which builds on~\lemref{monotonicity}.
\begin{restatable}{lemma}{lemconsumed}
% \begin{lemma}
\lemlabel{consumed-events}
Let $\tr$ be a trace and let $e_1, e_2, e'_2 \in \events{\tr}$
such that $e_1 \trord{\tr} e_2 \tho{\tr} e'_2$, 
$e_1 \conf e_2$ and $e_1 \conf e'_2$.
If $(e_1, e_2)$ is not a sync-preserving race, then $(e_1, e'_2)$
is also not a sync-preserving race of $\tr$.
% \end{lemma}
\end{restatable}
Intuitively, this observation suggests the following.
Suppose that, when looking for a sync-preserving race, 
the algorithm determines that $e_2$ is not in race with
any earlier conflicting event.
Then, for an event $e'_2$
(that appears later in the thread of $e_2$), we only need to investigate
if $e'_2$ is in race with conflicting events $e'_1$ that appear
\emph{after} $e_2$ in the trace (i.e., $e_2 \trord{\tr} e'_1 \trord{\tr} e_2$), 
instead of additionally looking for races $(e''_1, e'_2)$ where $e''_1 \trord{} e_2$.

\textbf{High-level overview of \ZeroRevAlgo.}
Equipped with~\lemref{monotonicity} and~\lemref{consumed-events},
we now describe our incremental algorithm 
for~\probref{zr-race-prediction} that works in linear time.
For ease of exposition, let us focus on the question --- is there
a write-write race on some fixed variable $x \in \vars{}$
when accessed in two fixed threads $t_1, t_2 \in \vars{}$.
The algorithm scans the trace in a streaming forward pass
and analyzes every event $e = \ev{t_2, \wt(x)}$, 
checking if there is an earlier conflicting event 
$e' = \ev{t_1, \wt(x)} \trord{} e$ so that $(e', e)$
is a sync-preserving race.
In doing so, it computes $I = \ZRIdeal{\tr}(e', e)$ in linear time
and checks if $e' \not\in I$.
If not, $(e', e)$ is not a race and the algorithm checks if
there is a different event $e'_\text{next} \trord{} e$ so that 
$(e'_\text{next}, e)$ is a race.
This continues
until there are no earlier events remaining that conflict with $e$.
Each time, the ideal computation is performed incrementally, by using
the previously computed ideals.
After this, the algorithm moves to the next write event $e_\text{next} = \ev{t_2, \wt(x)}$ 
in $t_2$ and checks if it is in race with some earlier event 
$e''$ of thread $t_1$,
where this time, $e''$ appears in the trace after the
previously discarded event $e$ of $t_2$.
Again, the closed set $\ZRIdeal{\tr}(e'', e_\text{next})$
is computed incrementally.
We show that all this incremental computation can be performed efficiently
and present an outline for our algorithm for~\probref{zr-race-prediction} 
in~\secref{full-algo}.

Next, we present high-level descriptions of the intermediate steps
that we outlined above, and discuss important algorithmic insights and
data-structures that help achieve efficiency.

%!TEX root = main.tex

\subsection{Checking if a given pair of conflicting events is a sync-preserving race}
\seclabel{overview-given-pair}

%!TEX root = main.tex

%%%%%%%%%%%%%%%%%%%%%%%%%%%%%%%%%%%%%%%%%%%%%%%%%%%%%
% Algorithm1(e1,e2,J)
%     I = J \cup down(prev(e1)) \cup down (prev(e2))
%     Repeat
%         If there are f, g \in I, and l such that f = acq(l), g = acq(l) and f <TR g then
%             I = I \cup down(match(f))
%     Until I does not change
%     If e1 \not\in I then race
%     Else no race.
%%%%%%%%%%%%%%%%%%%%%%%%%%%%%%%%%%%%%%%%%%%%%%%%%%%%%
% \small
\begin{algorithm*}[h]
\Input{Trace $\tr$, Conflicting events $e_1$ and $e_2$ with $e_1 \trord{\tr} e_2$}
\BlankLine
\myfun{\fixpoint{$\tr$, $e_1$, $e_2$, $I_0$}}{
    $I \gets I_0 \cup \TOOClosure{\tr}(\set{\prev{\tr}(e_1)}) \cup \TOOClosure{\tr}(\set{\prev{\tr}(e_2)})$ \; \linelabel{init}
    \Repeat{$I$ does not change}{ \linelabel{repeat_start}
        \If{$\exists \lk \in \locks{\tr}$, $\exists a_1,a_2 \in \acquires{\tr}(\lk)$ such that $a_1 \trord{\tr} a_2$ and $a_1, a_2 \in I$}{ \linelabel{acquirecheck}
            $I \gets I \cup \TOOClosure{\tr}(\set{\match{\tr}(a_1)})$ \; \linelabel{update_ideal}
        }
    } \linelabel{repeat_end}
    \Return $I$ \;
}
$I \gets$ \fixpoint{$\tr$, $e_1$, $e_2$, $\emptyset$} \;
\If{$e_1 \not\in I$}{ \linelabel{check_containment}
    \declare `race'
}
% \vspace*{0.25cm}
\caption{\textit{Checking if a given conflicting pair constitutes a sync-preserving race}}
\algolabel{check-given-pair}
\end{algorithm*}
\normalsize

\algoref{check-given-pair} outlines
our solution to~\probref{zr-race-prediction-given-pair}
(check if a given 
pair of events $(e_1, e_2)$ is a sync-preserving race of $\tr$).
This algorithm computes the closure $\ZRIdeal{\tr}(e_1, e_2)$
in an iterative fashion and checks if it 
contains neither $e_1$ nor $e_2$
(see~\lemref{zrideal-disjointness}).
We remark that when $e_1 \trord{\tr} e_2$, \defref{zr-closure}
ensures that $e_2 \not\in \ZRIdeal{\tr}(e_1, e_2)$.
Consequently, the check `$e_1 \not\in \ZRIdeal{\tr}(e_1, e_2)$'
(\lineref{check_containment} in~\algoref{check-given-pair}) is equivalent to
the condition `$\set{e_1, e_2} \cap \ZRIdeal{\tr}(e_1, e_2) = \emptyset$'
(due to~\lemref{zrideal-disjointness}).
The function \closurefun performs a fixpoint
computation, starting from the set 
$I = \bigcup_{i\in\set{1,2}}\TOOClosure{\tr}(\set{\prev{\tr}(e_i)})$
(when $I_0 = \emptyset$).
The correctness of the algorithm follows from the correctness of
the function \closurefun, which we formalize below.
\begin{lemma}
\lemlabel{computeLCCClosure-correctness}
Let $\tr$ be a trace, $e_1, e_2 \in \events{\tr}$ and
$I_0 \subseteq \events{\tr}$ be a $(\tho{\tr}, \lw{\tr})$-closed set.
Let $I$ be the set returned by \closurefun\!\texttt{(}$\tr, e_1, e_2, I_0$\texttt{)}
in~\algoref{check-given-pair}. Then, 
$I = \ZRClosure{\tr}(I_0 \cup \set{\prev{\tr}(e_1)} \cup \set{\prev{\tr}(e_2)})$.
\end{lemma}

Let us discuss the data-structures we use to
ensure that~\algoref{check-given-pair} runs in linear time and space.

\subsubsection{Vector timestamps.}
Vector timestamps~\cite{Mattern1988,Fidge:1991:LTD:112827.112860} are routinely used in distributed
computing and also in prior work on race prediction~\cite{Pozniansky:2003:EOD:966049.781529,fasttrack,wcp2017,Roemer18}.
We use vector timestamps to represent sets of events
that are $(\tho{\tr}, \lw{\tr})$-closed; a formal definition is deferred to~\secref{full-algo}.
In~\algoref{check-given-pair}, the sets
$\TOOClosure{\tr}(\set{\prev{\tr}(e_i)})$ are $(\tho{},\lw{})$-closed (\lineref{init}).
Further, the initial value $I_0 = \emptyset$ ensures that
all subsequent values of $I$ in \closurefun 
are $(\tho{},\lw{})$-closed.
All these sets can be represented as vector timestamps.
The advantage of using vector timestamps is two-fold.
First, these timestamps provide a succinct representation of
sets --- instead of representing a set explicitly 
as a collection of events, a vector timestamp 
only uses $\NumThreads$ integers (where $\NumThreads = |\threads{\tr}|$).
Second, the vector timestamps for $(\tho{\tr}, \lw{\tr})$-closed
sets can be computed in a streaming fashion, incrementally,
using vector timestamps of smaller subsets.

\subsubsection{Projecting a trace to threads and locks}
Let us consider the check in~\lineref{acquirecheck}.
Here, we look for two acquire events $a_1 \trord{\tr} a_2$
in the current ideal $I$ that acquire the same lock $\lk$.
How do we efficiently discover two such acquires?
A straightforward but naive way is to enumerate all pairs of events
in $I$ and check if they are acquire events of the above kind.
But this can take $O(\NumEvents^3)$ time, where $\NumEvents = |\events{\tr}|$ because
the number of such pairs can be $O(\NumEvents^2)$ in the worst case and
the number of times the ideal can change is $O(\NumEvents)$.
Instead, we rely on the following observation:
\begin{proposition}
\proplabel{one-acquire-per-thread}
Let $I \subseteq \events{\tr}$ be downward closed with respect to $\tho{\tr}$.
For every $t \in \threads{\tr}$ and every $\lk \in \locks{\tr}$,
there is at most one acquire event $a$ with $\ThreadOf{a} = t$,
$\OpOf{a} = \acq(\lk)$
such that $a \in I$ and $\match{\tr}(a) \not\in I$.
When such an event $a$ exists, then $\match{\tr}(a') \in I$
for every other acquire 
$a' \stricttho{\tr} a$ of the form $a' = \ev{t, \acq(\lk)}$.
\end{proposition}
The above observation can be exploited as follows.
Let $I$ be the current ideal and
let $e^I_{t, \lk}$ be the last acquire on lock $\lk$
performed by thread $t$ such that $e^I_{t, \lk} \in I$.
Let $\textsf{Acq}^I_\lk = \set{e^I_{t, \lk}}_{t \in \threads{\tr}}$ 
be the set of last acquire events in $I$ of every thread.
Let $e^I_\lk$ be the last event (according to trace order $\trord{\tr}$)
in $\textsf{Acq}^I_\lk$.
Then, for every other 
acquire $e' \in \textsf{Acq}^I_\lk \setminus \set{e^I_\lk}$,
the matching release $\match{\tr}(e')$ must be added in $I$ .
Hence, if we can efficiently determine the events $e^I_{t, \lk}$ each time, 
then we can also efficiently determine $e^I_\lk$ and
thus efficiently perform the closure each time.
So, how do we determine $e^I_{t, \lk}$ efficiently each time?
We achieve this by maintaining a FIFO sequence $\AcqLst_{t, \lk}$,
for every thread $t \in \threads{\tr}$ and lock $\lk \in \locks{\tr}$.
% $\set{\AcqLst_{t, \lk}}_{t \in \threads{\tr}, \lk \in \locks{\tr}}$.
For every critical section on lock $\lk$ in thread $t$,
there is a corresponding entry in the list $\AcqLst_{t, \lk}$,
and the order of these entries is the same
as the order in which these critical sections
were performed in the trace $\tr$.
Every entry (corresponding to critical section with acquire $a$) is a pair 
$(\TOOClosure{\tr}(a), \TOOClosure{\tr}(\match{\tr}(a)))$, 
represented as a pair of vector timestamps $(C_a, C_{\match{}(a)})$.

Let us now see how we perform the check in~\lineref{acquirecheck} 
using these data structures $\set{\AcqLst_{t, \lk}}_{t, \lk}$.
For the current ideal $I$, we essentially need to determine the 
last acquire $e^I_{t, \lk}$ of each thread $t$ and lock $\lk$ that belongs to $I$.
This can be done by traversing the list $\AcqLst_{t, \lk}$ starting from the earliest entry,
until we encounter the entry corresponding to the last acquire $e^I_{t, \lk}$ that belongs to $I$ 
(this corresponds to a simple timestamp comparison).
All entries in $\AcqLst_{t, \lk}$ prior to the identified event $e^I_{t, \lk}$
can then be discarded from $\AcqLst_{t, \lk}$,
because the ideal now contains their information and only grows monotonically through the course of the 
fixpoint computation and the discarded entried will not be of use from now on.
% If $\match{\tr}(e^I_{t, \lk})$ must be added to $I$, we just the elements of $\TOOClosure{\tr}(\match{\tr}(a))$ to $I$.
Thus, every entry in the lists $\set{\AcqLst_{t, \lk}}_{t, \lk}$
is traversed only once and the overall fixpoint computation runs in linear time
when the number of $\NumThreads$ is constant.

%!TEX root = main.tex

\subsection{Checking for a Sync-Preserving Race on a Given Event with a Given Thread}
\seclabel{overview-existence}

%!TEX root = main.tex

%%%%%%%%%%%%%%%%%%%%%%%%%%%%%%%%%%%%%%%%%%%%%%%%%%%%%
% Algorithm 2: Check if there is a zero reversal race between e2 and an earlier event of t1
% L_{t_1} = list of events of t1 upto e2 that conflict with e2
% I = \emptyset
% While L is not empty do
%     Let e1 = first event of L
%     I = I(e1,e2,I)
%     If e1 \not\in I
%         Declare race
%         Break
%     Else
%        Remove from L all events that belong to I
%%%%%%%%%%%%%%%%%%%%%%%%%%%%%%%%%%%%%%%%%%%%%%%%%%%%%

%%%%%%%%%%%%%%%%%%%%%%%%%%%%%%%%%%%%%%%%%%%%%%%%%%%%%
% Algorithm 2 modified: Check if there is a zero reversal race between events of t1 and t2
% L1 = list of events of t1
% L2 = list of events of t2
% I = \emptyset
% While L2 is not empty:
% 	e2 = L2.first
% 	While there is an event e1 in L1 such that e1 <tr e2:
% 		e1 = L1.first
% 		I = I(e1, e2, I)
% 		if e1 \not\in I
% 			declare race on e1
% 			break (also break the outer loop)
% 		else:
% 			remove e1 from L1 // Can be optimized to remove all elements from L1 that belong to I
%%%%%%%%%%%%%%%%%%%%%%%%%%%%%%%%%%%%%%%%%%%%%%%%%%%%%
% \small
\begin{savenotes}
\begin{algorithm*}[h]
\Input{Trace $\tr$, Event $e = \ev{t, a(x)}$, Thread $u$}
\BlankLine
\Foreach{$b \in \set{\rd, \wt}$ such that  $b \conf a$\footnote{\text{$b \conf a$ whenever not both $b$ and $a$ are read ($\rd$) operations}}}{
	\Let $L_b$ be the list of events $e'$ of the form $\ev{u, b(x)}$ such that $e' \trord{\tr} e$ \linelabel{access_lists}
}
\Foreach{$b \in \set{\rd, \wt}$ such that $b \conf a$}{
	$I_b \gets \emptyset$ \;
	\Foreach{$e' \in L_b$}{
		$I_b \gets$ \fixpoint{$\tr$, $e'$, $e$, $I_b$} \;
		\If{$e' \not\in I_b$}{
	    \declare `race'
	     and \exit\;
	    % \exit
		}
	}
}
% \vspace*{0.25cm}
\caption{\textit{Checking if there is a sync-preserving race on a given event with a given thread}}
\algolabel{check-given-thread}
\end{algorithm*}
\end{savenotes}
\normalsize

Let us now consider~\algoref{check-given-thread}.
This algorithm takes as input a trace $\tr$, an event $e = \ev{t, a(x)}$ 
and a threads $u \neq t$,
and checks if there is an event $e'$ of thread $u$
such that $e' \trord{\tr} e$ and $(e', e)$ is a sync-preserving race.
\algoref{check-given-thread}~works as follows.
For the sake of simplicity, assume that 
$e$ is a read event, i.e., $a = \rd$.
The algorithm assumes access to the list $L_\wt$
of write events $e' = \ev{u, \wt(x)}$ that appear prior to $e$
(\lineref{access_lists});
these lists can be constructed in linear time, in a single pass
traversal of the trace.
The algorithm simply traverses $L_\wt$ (according to trace order $\trord{\tr}$)
and checks for races with each event in $L_\wt$, 
by computing the fixpoint closure sets as in~\algoref{check-given-pair}.
Instead of computing the ideal from scratch, the algorithm
exploits the monotonicity property outlined earlier in~\lemref{monotonicity}
by reusing the ideal $I_\wt$ computed so far.
As with~\algoref{check-given-pair}, this algorithm
also uses vector timestamps and maintains lists $\set{\AcqLst_{t, \lk}}_{t\in\threads{}, \lk\in\locks{}}$
for computing successive ideals efficiently.
Overall again, each entry in $\set{\AcqLst_{t, \lk}}_{t\in\threads{}, \lk\in\locks{}}$
is visited a constant number of times and thus~\algoref{check-given-thread}
runs in linear time and uses linear space.

%!TEX root = main.tex

\subsection{Algorithm $\ZeroRevAlgo$ for Sync-Preserving Race Prediction}
\seclabel{full-algo}

\small
\begin{algorithm*}[h]
% \BlankLine
\begin{multicols}{2}
\myfun{\init{}}{
	\ForEach{$t \in \threads{}$ $\cdot$}{
		$\Cc_t$ := $\bot$
	}
	\ForEach{$x \in \vars{}$ $\cdot$}{
		$\LW_x$ := $\bot$
	}
	\ForEach{$\lk \in \locks{}$ $\cdot$}{
		$\gId_\lk$ := $0$
	}
	\ForEach{$t_1 {\neq}t_2 \in \threads{}, a_1 {\conf} a_2 \in \set{\rd, \wt}, x \in \vars{}$}{
		$\view{\Ii}{}{\tuple{t_1, t_2, a_1, a_2, x}}$ := $\bot$ \;
		\ForEach{$t \in \threads{}, \lk \in \locks{}$}{
			$\view{\AcqLst}{t, \lk}{\tuple{t_1, t_2, a_1, a_2, x}}$ := $\emptyset$
		}
	}
	\ForEach{$u \in \threads{}$}{
		\ForEach{$t\in \threads{}, a \in \set{\rd, \wt}, x\in \vars{}$}{
			$\view{\AccLst}{t, a, x}{\tuple{u}}$ := $\emptyset$
		}
	}
}
\BlankLine
% {\tiny \tcp{$Lst$ is a list of triplets $(g, C, C')$ in increasing order of $C$}}
% {\tiny \tcp{Returns the last element $(g_{\max}, C_{\max}, C'_{\max})$ s.t. $C_{\max} \cle U$}}
% {\tiny \tcp{Removes all other elements $(g,C, C')$ from $Lst$ with $C \cle U$}}
\myfun{\maxLB{$U$, $Lst$}}{
	$(g_{\max}, C_{\max}, C'_{\max})$ := $(0, \bot, \bot)$ \; 
	\While{$\NOT\, \isEmpty{Lst}$}{
		$(g, C, C')$ := $\first{Lst}$ \;
		\If{$C \cle U$}{
			$(g_{\max}, C_{\max}, C'_{\max})$ := $(g, C, C')$
		}
		\Else{
			\Break
		}
		$\removeFirst{Lst}$
	}
	\Return $(g_{\max}, C_{\max}, C'_{\max})$
}

\BlankLine

\myfun{\fixpoint{$I$, $\tuple{t_1, t_2, a_1, a_2, x}$}}{
	\Repeat{$I$ does not change}{
		\ForEach{$\lk \in \locks{}$}{
			\ForEach{$t \in \threads{}$}{
				($g_{\lk, t}, C_{\lk, t}, C'_{\lk, t}$) := 
				\maxLB{$I$, $\view{\AcqLst}{t, \lk}{\tuple{t_1, t_2, a_1, a_2, x}}$} \;
			}
			$t_{\max}$ := argmax$_{t \in \threads{}}\set{g_{\lk, t}}$ \;
			$I$ := $I \mx \bigsqcup_{t \neq t_{\max} \in \threads{}} C'_{\lk, t}$	
		}
	}
	\Return $I$
}

\myfun{\checkRace{$Lst$, $I$, $\tuple{t_1, t_2, a_1, a_2, x}$}}{
	\While{$\NOT\, \isEmpty{Lst}$}{
		($C_{\prev{}}, C$) := $\first{Lst}$ \;
		% $I$ := $I \mx C_{\prev{}}$ \;
		$I$ := \fixpoint{$I \mx C_{\prev{}}$, $\tuple{t_1, t_2, a_1, a_2, x}$} \;
		\If{$C \not\cle I$}{ \linelabel{check-containment}
			% \declare `($a_1, a_2$) race on $x$ in threads $t_1$ and $t_2$' \;
			\declare `($a_1, a_2$) race on $x$' \;
			\Break
		}
		$\removeFirst{Lst}$
	}
	\Return $I$
}

\myhandler{\rdhandler{$t$, $x$}}{
	$C_{\prev{}}$ := $\Cc_t$ \; 
	$\Cc_t$ := $\Cc_t[t \mapsto \Cc_t(t)+1] \mx \LW_x$ \linelabel{clkrd} \;
	\ForEach{$u\in \threads{}$}{
		$\addLst{\view{\AccLst}{t, \rd, x}{\tuple{u}}}{($C_{\prev{}}$, $\Cc_t$)}$
	}
	\ForEach{$u \neq t \in \threads{}$}{
		$I$ := $\view{\Ii}{}{\tuple{u, t, \wt, \rd, x}} \mx C_{\prev{}}$ \;
		$\view{\Ii}{}{\tuple{u, t, \wt, \rd, x}}$ := \checkRace{$\view{\AccLst}{u, \wt, x}{\tuple{u}}$, $I$, $\tuple{u, t, \wt, \rd, x}$}
	}
}

\myhandler{\wthandler{$t$, $x$}}{
	$C_{\prev{}}$ := $\Cc_t$ \;
	$\Cc_t$ := $\Cc_t[t \mapsto \Cc_t(t)+1]$;  \quad
	$\LW_x$ := $\Cc_t$ \linelabel{clkwt} \;
	\ForEach{$u\in \threads{}$}{
		$\addLst{\view{\AccLst}{t, \wt, x}{\tuple{u}}}{($C_{\prev{}}$, $\Cc_t$)}$
	}
	\ForEach{$u \neq t \in \threads{}, a \in \set{\rd, \wt}$}{
		$I$ := $\view{\Ii}{}{\tuple{u, t, a, \wt, x}} \mx C_{\prev{}}$ \;
		$\view{\Ii}{}{\tuple{u, t, a, \wt, x}}$ := \checkRace{$\view{\AccLst}{u, a, x}{\tuple{u}}$, $I$, $\tuple{u, t, a, \wt, x}$}
	}
}

\myhandler{\acqhandler{$t$, $\lk$}}{
	$\Cc_t$ := $\Cc_t[t \mapsto \Cc_t(t)+1]$; \quad $\gId_\lk$ := $\gId_\lk + 1$ \linelabel{clkacq}\;
	\ForEach{$t_1, t_2 \in \threads{}, a_1, a_2 \in \set{\rd, \wt}, x \in \vars{}$}{
		$\addLst{\view{\AcqLst}{t, \lk}{\tuple{t_1, t_2, a_1, a_2, x}}}{$(\gId_\lk, \Cc_t, \bot)$}$
	}
}

\myhandler{\relhandler{$t$, $\lk$}}{
	$\Cc_t$ := $\Cc_t[t \mapsto \Cc_t(t)+1]$ \linelabel{clkrel} \;
	\ForEach{$t_1, t_2 \in \threads{}, a_1, a_2 \in \set{\rd, \wt}, x \in \vars{}$}{
		$\updateRel{\lst{\view{\AcqLst}{t, \lk}{\tuple{t_1, t_2, a_1, a_2, x}}}}{$\Cc_t$}$
	}
}

\end{multicols}
% \vspace*{0.25cm}
\normalsize
\caption{\textit{Detailed streaming algorithm for checking sync-preserving races}}
\algolabel{overall}
\end{algorithm*}
\normalsize

The pseudo-code for~$\ZeroRevAlgo$ is presented in~\algoref{overall}.
This is a one pass streaming algorithm that processes events as they appear in the trace,
modifying its state and detecting races on the fly.
The algorithm maintains several data structures including vector clocks
and FIFO queues in its state.
We first describe these data structures, 
then discuss how the algorithm initializes and 
modifies them as it processes the trace, and
finally discuss the time and space usage for this algorithm;
many of these details have already been spelt out 
in~\secrangeref{overview}{overview-existence}.

Let us first briefly explain the notion of vector timestamps and 
vector clocks~\cite{Fidge:1991:LTD:112827.112860,Mattern1988}.
A vector timestamp is a mapping $V : \threads{\tr} \to \nats$
from the threads of a trace to natural numbers, and can be represented
as a vector of length $|\threads{\tr}|$.
The \emph{join} of two vector timestamps $V_1$ and $V_2$,
denoted $V_1 \mx V_2$ is the vector timestamp $\lambda t, \max(V_1(t), V_2(t))$.
Vector timestamps can be compared in a pointwise fashion: $V_1 \cle V_2$ iff
$\forall t, V_1(t) \leq V_2(t)$.
The minimum timestamp is denoted by $\bot$ --- $\bot = \lambda t, 0$.
For a scalar $c \in \nats$, we use $V[t\mapsto c]$
to denote the timestamp $\lambda u, \text{ if } u = t \text{ then } c \text{ else } V(u)$.
Vector clocks are variables that take values from 
the space of vector timestamps. 
%All operations on vector timestamps also apply to clocks.
We use normal font for timestamps ($C, C', I \ldots$)
and boldfaced font for vector clocks ($\Cc, \LW, \Ii, \ldots$).

\vspace{-0.1in}
\myparagraph{Data structures and initialization}{
The algorithm maintains the following data structures.
\begin{enumerate}
	\item {\bf Vector clocks}. 
	The algorithm uses vector clocks primarily for two purposes.
	First, we assign timestamps to all events in the trace and use
	the following vector clocks for this purpose ---
	for every thread, a dedicated clock $\Cc_t$, and for every variable $x$,
	a dedicated clock $\LW_x$.
	The timestamp of an event $e$ is 
	essentially a succinct representation
	of the set $\TOOClosure{}(e)$.
	Next, the algorithm computes $\ZRIdeal{}(\cdot, \cdot)$
	sets and represents them as timestamps.
	These are stored in vector clocks $\view{\Ii}{}{\tuple{t_1, t_2, a_1, a_2, x}}$,
	one for every tuple $(t_1, t_2, a_1, a_2, x) \in \threads{} \times \threads \times \set{\rd, \wt} \times \set{\rd, \wt} \times \vars{}$.
	All vector clocks are initialized with the timestamp $\bot = \lambda t, 0.$

	\item {\bf Scalars}.
	For every lock $\lk$, the algorithm maintains a scalar variable $\gId_\lk$
	to record the index of the last acquire on $\lk$ seen in the trace seen so far.
	Each such scalar is initialized with $0$.

	\item {\bf FIFO queues}.
	The algorithm maintains several FIFO queues, whose entries correspond to different events in the trace.
	The algorithm ensures that an event appears only once, and additionally ensures that the entries 
	respect the order of appearance of the corresponding events in the trace.
	The first kind of FIFO queues are used in the fixpoint computation.
	For this, the algorithm maintains queues $\view{\AcqLst}{t, \lk}{\tuple{t_1, t_2, a_1, a_2, x}}$,
	one for every thread $t$, lock $\lk$ and tuple $(t_1, t_2, a_1, a_2, x)$.
	Each entry of $\view{\AcqLst}{t, \lk}{\tuple{t_1, t_2, a_1, a_2, x}}$
	 is a triplet $(g_e, C_e, C'_e)$
	and corresponds to an acquire event $e$ of the form $\ev{t, \acq(\lk)}$ --- here, $g_e$
	is the index of $e$ in the trace, $C_e$ is the timestamp of $e$ and $C'_e$
	is the timestamp of the matching release $\match{}(e)$.
	The second kind of queues are of the form $\view{\AccLst}{t, a, x}{\tuple{u}}$, 
	one for each $t, u \in \threads{}$, $a \in \set{\rd, \wt}$ and $x \in \vars{}$.
	Entry in such a FIFO queue corresponds to access events
	of the form $e = \ev{t, a(x)}$.
	Each entry is of the form $(C_{\prev{}(e)}, C_e)$ where
	$C_{\prev{}(e)}$ is the timestamp of $\prev{}(e)$ 
	(and $\bot$ if $\prev{}(e)$ does not exist)
	 and $C_e$ is the timestamp of $e$.
	 All FIFO queues are empty ($\emptyset$) in the beginning.
	 We write $\first{L}$ and $\lst{L}$ to denote the first
	 and last elements of the FIFO queue $L$.
	 Further, $\isEmpty{L}$, $\addLst{L}{}$ and $\removeFirst{L}$ respectively
	 represent functions that check for emptiness of $L$,
	 add an element at the end of $L$
	and remove the earliest (first) element of $L$.
\end{enumerate}
}

Let us now describe the working of $\ZeroRevAlgo$.
The algorithm works in a streaming fashion and processes each event
$e$ as soon as it appears, by calling the appropriate \textbf{handler}
depending upon the operation performed in $e$ 
(\texttt{read}, \texttt{write}, \texttt{acquire} or \texttt{release}).
The argument for each handler is the thread performing the event
and the object (variable or lock) accessed in the event.
In each handler, the algorithm updates vector clocks to
compute timestamps,
and maintains the invariants of the FIFO queues.
In addition, inside \texttt{read} and \texttt{write} handlers,
the algorithm also checks for races using a fixpoint computation
(function \closurefun).
We explain some of these briefly.

\myparagraph{Computing vector timestamps}{
	The algorithm computes the timestamp of an event $e$,
	denoted $C_e$ when processing the event $e$.
	The timestamps computed in the algorithm are such that
	\[
	\forall t\in \threads{}, C_e(t) = |\setpred{f \in \events{}}{\ThreadOf{f} = t, f \in \TOOClosure{}(e)}|
	\]
	Observe that, with this invariant, we have $C_e \cle C_{e'}$ iff $e \in \TOOClosure{}(e')$.
	The algorithm uses vector clocks $\set{\Cc_t}_{t\in \threads{}}$ and $\set{\LW_x}_{x\in \vars{}}$
	and ensures that after processing an event $e$, 
	\begin{enumerate*}
	\item $\Cc_t$ stores the timestamp $C_{e_t}$, where $e_t$ is the last event by thread $t$
	that occurs before $e$, and
	\item $\LW_x$ stores the timestamp $C_{e_x}$, where $e_x$ is the last event with $\OpOf{e_x} = \wt(x)$
	that occurs before $e$.
	\end{enumerate*}
	The algorithm correctly maintains these values by appropriate
	vector clock operations  on
	Lines~\ref{line:clkrd}, \ref{line:clkwt}, \ref{line:clkacq} and \ref{line:clkrel}.
	%%%% The following is new %%%%

	Let us now describe the invariant for the vector clocks $\view{\Ii}{}{\tuple{t_1, t_2, a_1, a_2, x}}$ (for a given tuple $(t_1, t_2, a_1, a_2, x)$).
	Let $e_{t_i, a_i, x}$ be the last event with $\ThreadOf{e_{t_i, a_i, x}} = t_i$ and $\OpOf{e_{t_i, a_i, x}} = a_i(x)$ ($i \in \set{1, 2}$).
	Let $I \subseteq \events{\tr}$ be defined as follows. 
	If $e_{t_1, a_1, x}$ does not exist, then $I = \ZRClosure{\tr}(\set{\prev{\tr}(e_2)})$.
	Otherwise, let $e$ be be the first event in $\tr$ (according to trace order $\trord{\tr}$) with $\ThreadOf{e} = t_1$ and $\OpOf{e} = a_1(x)$ such that $(e, e_2)$ is a sync-preserving race of $\tr$. If no such event exists, then let $e = e_{t_1, a_1, x}$. Then, $I = \ZRIdeal{\tr}(e_1, e_2)$.
	Then, the timestamp stored in $\view{\Ii}{}{\tuple{t_1, t_2, a_1, a_2, x}}$ is $\bigsqcup_{u \in \threads{\tr}} C_{e^I_u}$, where $e^I_u$ is the last event of thread $u$ which is in $I$.
	%%%%%%%%%%%%%%%%%%%%%%%%%%%%%%
}

\vspace{-0.1in}
\myparagraph{Checking races}{
	When processing an access event $e = \ev{t, a(x)}$, the algorithm checks for a race as follows.
	For every other thread $u$ and for every other conflicting type $b \conf a$,
	the algorithm calls \racefun with the list 
	$Lst = \view{\AccLst}{u, b, x}{\tuple{t}}$ and the
	timestamp representation of the current ideal as argument.
	This function, similar to~\algoref{check-given-thread}, scans $Lst$
	and reports races by repeatedly performing fixpoint computations and checking membership in some set
	(using the timestamp comparison in~\lineref{check-containment}).
	The closure computation is performed using the optimizations discussed in~\secref{overview-given-pair}
	(use of FIFO queues $\view{\AcqLst}{t, \lk}{\tuple{\cdots}}$).
}

\vspace{-0.1in}
\myparagraph{Space optimizations}{
	Observe that, for a given thread $t$ and lock $\lk$, the algorithm,
	as presented, maintains $4\NumThreads^2{\cdot}\NumVariables$ FIFO queues
	$\set{\view{\AcqLst}{t, \lk}{\tuple{t_1, t_2, a_1, a_2, x}}}_{t_1, t_2 \in \threads{}, a_1, a_2 \in \set{\rd, \wt}, x \in \vars{}}$.
	The total number of entries across these queues will then be $O(\NumAcquires \cdot 4\NumThreads^2{\cdot}\NumVariables)$.
	To this end, we observe that all the above data structures essentially have the same content, and
	are suffixes of a common queue corresponding to the critical sections on $\lk$ in thread $t$.
	Indeed, we exploit this redundancy and instead maintain a common
	underlying data-structure that stores 
	all entries corresponding to acquires and releases on $\lk$ in $t$,
	and maintain a pointer for each $\tuple{t_1, t_2, a_1, a_2, x}$.
	Such a pointer keeps track of the
	starting index of the FIFO queue $\view{\AcqLst}{t, \lk}{\tuple{t_1, t_2, a_1, a_2, x}}$.
	With this space optimization, we only store $O(\NumAcquires)$ entries along with
	additional $4{\cdot}\NumThreads^3{\cdot}\NumVariables{\cdot}\NumLocks$ pointers,
	one for every tuple $\tuple{t_1, t_2, a_1, a_2, x}$ and every shared queue indexed by $(t, \lk)$.
	The same observations also apply to the FIFO queues 
	$\set{\view{\AccLst}{t, a, x}{\tuple{u}}}_{u \in \threads{}}$
}

\begin{lemma}[Correctness]
\lemlabel{correctness-algo}
For every access event $e$ in the input trace $\tr$,
\algoref{overall} declares a race when processing $e$ iff there is an event
$e'$ such that $e' \trord{\tr} e$ and $(e', e)$ is a 
sync-preserving race of $\tr$.
\end{lemma}

The proof of~\lemref{correctness-algo} follows directly from the
correctness of the semantics of clocks and other
observations outlined in~\secref{overview}.

The time complexity of the algorithm can be determined as follows.
The algorithm visits each entry in the FIFO queues 
$\view{\AccLst}{t, a, x}{\tuple{u}}$
once, performing constant number of vector clock operations,
each running in $O(\NumThreads)$ time.
The total length of all these queues is $O(\NumThreads{\cdot}\NumEvents)$ (more precisely, the number of access events in the trace).
Similarly, the algorithm visits each entry in the FIFO queues
$\view{\AcqLst}{t, \lk}{\tuple{t_1, t_2, a_1, a_2, x}}$ once,
performing constantly many vector clock operations.
The total number of entries in these queues is $O(\NumThreads^2{\cdot}\NumVariables{\cdot}\NumAcquires)$.
This gives us the following complexity for $\ZeroRevAlgo$.
\begin{lemma}[Complexity]
\lemlabel{complexity-algo}
Let $\tr$ be a trace with
$\NumThreads$ threads, $\NumLocks$ locks, 
$\NumVariables$ variables and $\NumEvents$ events,
of which $\NumAcquires$ are acquire events.
Then, \algoref{overall} runs in time $O(\NumEvents\cdot \NumThreads^2 + \NumAcquires\cdot \NumVariables\cdot \NumThreads^3)$ and uses space $O(\NumEvents+ \NumThreads^3\cdot \NumVariables\cdot \NumLocks)$
on input $\tr$.
\end{lemma}

The proof of~\thmref{zero_reversals} follows from~\lemref{correctness-algo}
 and \lemref{complexity-algo}.

 % \input{invariants}

%!TEX root = main.tex

\subsection{Linear Space Lower Bound}
\label{subsec:space_lowerbound}

In this section we prove the lower-bounds of \cref{thm:space_lowerbound} and \cref{thm:product_lowerbound}, i.e., 
that any streaming algorithm for sync-preserving race prediction must essentially use linear space, while 
the time-space product of any  algorithm for the problem must be quadratic in the length of the input trace.

\Paragraph{The language $\Lang_n$.}
Given a natural number $n$, we define the equality language $\Lang_n =\{u\#v\colon u,v\in \{0,1\}^n \text{ and } u=v\}$,
i.e., it is the language of two $n$-bit strings that are separated by $\#$ and are equal.
%We have the following lemma.

%\smallskip
\begin{restatable}{lemma}{lemlangnlowerbound}\label{lem:langnlowerbound}
Any streaming algorithm that recognizes $\Lang_n$ uses $\Omega(n)$ space.
\end{restatable}

%\input{figures/fig_space_lowerbound_bk}
%!TEX root = ../main.tex

\begin{figure}
\execution{2}{
\figev{1}{\acq( b_1)}
\figev{1}{\acq( b_2)}
\figev{1}{\acq( c)}
\figev{1}{\wt_1(x)}
\figev{1}{\rel( c)}
\figev{1}{\rel( b_2)}
\figev{1}{\rel( b_1)}
\figev{1}{\acq( a_1)}
\figev{1}{\acq( b_2)}
\figev{1}{\rd_2(x)}
\figev{1}{\rel( b_2)}
\figev{1}{\rel( a_1)}
}
\execution{2}{
\figevoffset{12}{1}{\acq( a_2)}
\figevoffset{12}{1}{\acq( b_1)}
\figevoffset{12}{1}{\rd_3(x)}
\figevoffset{12}{1}{\rel( b_1)}
\figevoffset{12}{1}{\rel( a_2)}
\figevoffset{12}{1}{\acq( a_1)}
\figevoffset{12}{1}{\acq( a_2)}
\figevoffset{12}{1}{\acq( c)}
\figevoffset{12}{1}{\wt_4(x)}
\figevoffset{12}{1}{\rel( c)}
\figevoffset{12}{1}{\rel( a_2)}
\figevoffset{12}{1}{\rel( a_1)}
}
\execution{2}{
\figevoffset{24}{2}{\acq( a_1)}
\figevoffset{24}{2}{\acq( a_2)}
\figevoffset{24}{2}{\acq( c)}
\figevoffset{24}{2}{\wt_1(x)}
\figevoffset{24}{2}{\rel( c)}
\figevoffset{24}{2}{\rel( a_2)}
\figevoffset{24}{2}{\rel( a_1)}
\figevoffset{24}{2}{\acq( b_1)}
\figevoffset{24}{2}{\acq( a_2)}
\figevoffset{24}{2}{\rd_2(x)}
\figevoffset{24}{2}{\rel( a_2)}
\figevoffset{24}{2}{\rel( b_1)}
}
\execution{2}{
\figevoffset{36}{2}{\acq( b_2)}
\figevoffset{36}{2}{\acq( a_1)}
\figevoffset{36}{2}{\acq( c)}
\figevoffset{36}{2}{\wt_3(x)}
\figevoffset{36}{2}{\rel( c)}
\figevoffset{36}{2}{\rel( a_1)}
\figevoffset{36}{2}{\rel( b_2)}
\figevoffset{36}{2}{\acq( b_1)}
\figevoffset{36}{2}{\acq( b_2)}
\figevoffset{36}{2}{\acq( c)}
\figevoffset{36}{2}{\wt_4(x)}
\figevoffset{36}{2}{\rel( c)}
\figevoffset{36}{2}{\rel( b_2)}
\figevoffset{36}{2}{\rel( b_1)}
}
\caption{
Construction of the trace $\tr$ on input $s=u\#v$ where $u=1001$ and $v=1011$.
Observe that $(e_{15}, e_{40})$ is a sync-preserving race, which encodes that $e[3]\neq v[3]$.
}
\label{fig:space_lowerbound}
\end{figure}

\Paragraph{Reduction from $\Lang_n$ recognition to sync-preserving race prediction.}
Consider the language $\Lang_n$ for some $n$.
We describe a transducer $\Transducer_n$ such that, on input a string 
$s=u\#v$, the output $\Transducer_n(s)$ is a trace $\tr$ with $2$ threads, $O(n\cdot \log n)$ events, $2\cdot \log n + 1$ locks and a single variable such that the following hold.
\begin{compactenum}
\item If $s\not \in \Lang_n$, then $\tr$ has no predictable race.
\item If $s \in \Lang_n$, then $\tr$ has a single predictable race, which is a sync-preserving race.
\end{compactenum}
Moreover, $\Transducer_n$ uses $O(\log n)$ working space.
The transducer $\Transducer_n$ uses a single variable $x$, two sets of locks $A=\{ a_1,\dots, a_{\log n} \}$ and $B=\{b_1, \dots, b_{\log n} \}$, plus one additional lock $c$.
The trace $\tr$ consists of two local traces $\SeqTrace_1$, $\SeqTrace_2$ of threads $\thread_1$ and $\thread_2$ which encode the bits of $u$ and $v$, respectively.

\begin{compactenum}
\item The local trace $\SeqTrace_1$ is constructed as follows.
For every $i\in [n]$, $\SeqTrace_1$ contains an event $e^1_i$, which is a write event $\wt(x)$ if $u[i]=1$, and a read event $\rd(x)$ otherwise.
The events $e^1_i$ are surrounded by locks from $A$ and $B$ arbitrarily, as long as the following holds.
For any $i<j$, we have
\begin{align*}
\lheld{\tr}(e^1_j)\cap A \not \subseteq \lheld{\tr}(e^1_i) \cap A\quad \text{and}\quad \lheld{\tr}(e^1_i)\cap B \not \subseteq \lheld{\tr}(e^1_j) \cap B.
\end{align*}
Here, the locks held at an event $e$ has the obvious meaning: 
$\lheld{\tr}(e) = \setpred{\lk \in \locks{\tr}}{\exists a \in \acquires{\tr}(\lk) \text{ such that } e \in \crit{\tr}(a)}$.

The above property can be easily met, for example, by making $\Transducer_n$ perform a breadth-first traversal of the subset-lattice of $A$ (resp., $B$)
starting from the top (resp., bottom).
Given the current $i$, the transducer surrounds $e^1_i$ with the locks of the current element in the corresponding lattice.
Finally, every write event $e^1_i$ is surrounded by the lock $c$.
\item  The local trace $\SeqTrace_2$ is similar to $\SeqTrace_1$, i.e., we have an event $e^2_i$ for each $i\in[n]$, which is a write event $\wt(x)$ if $u[i]=1$, otherwise it is a read event $\rd(x)$.
The locks that surround $e^2_i$ are such that
\begin{align*}
\lheld{\tr}(e^2_i)\cap (A\cup B) = A\cup B\setminus\lheld{\tr}(e^1_j).
\end{align*}
Finally, similarly to $\SeqTrace_1$, every write event $e^1_2$ is surrounded by the lock $c$.
\end{compactenum}
See \cref{fig:space_lowerbound} for an illustration.
Observe that $\Transducer_{n}$ uses $O(\log n)$ bits of memory, for storing a bit-set of locks for each set $A$ and $B$ that must surround the current event $e^1_i$ and $e^2_i$.

The key idea of the construction is the following.
Any two events $e^1_i, e^2_j$ are surrounded by a common lock from the set $A\cup B$ iff $i\neq j$.
Hence, $(e^1_i, e^2_j)$ may be a predictable race of $\tr$ only if $i=j$.
In turn, if $u[i]=v[j]$, then either both events are read events, or both are write events.
In the former case the events are not conflicting, while in the latter case the two events are surrounded by lock $c$.
In both cases no race occurs between $e^1_i$ and $e^2_j$.
On the other hand, if $u[i]\neq v[j]$, then one event is a read and the other is a write event. 
Hence, the two events conflicting, and one of them is not surrounded by lock $c$, thereby constituting a predictable race.
The following lemma makes the above insight formal and establishes the correctness of our reduction.

\smallskip
\begin{restatable}{lemma}{lemlangncorrectness}\label{lem:langn_correctness}
The following assertions hold.
\begin{compactenum}
\item If $s \in \Lang_n$, then $\tr$ has no predictable race.
\item If $s \not \in \Lang_n$, then $\tr$ has a single predictable race, which is a sync-preserving race.
\end{compactenum}
\end{restatable}

We are now ready to prove \cref{thm:space_lowerbound}, and refer to 
\appref{app_detection} 
for the proof of \cref{thm:product_lowerbound}.

\begin{proof}[Proof of \cref{thm:space_lowerbound}.]
Consider any algorithm $A_1$ for sync-preserving race prediction, executed in the family of traces $\tr$ constructed in our above reduction.
Let $m=n/\log n$, and assume towards contradiction that $A_1$ uses $o(m)$ space.
Then we can pair $A_1$ with the transducer $\Transducer_{m}$, and obtain a new algorithm $A_2$ for recognizing $\Lang_{m}$.
Since $\Transducer_{m}$ uses $O(\log m)$ space, the space complexity of $A_2$ is $o(m)$.
However, this contradicts \cref{lem:langnlowerbound}.
The desired result follows.
\end{proof}

%!TEX root = main.tex

\section{Beyond Synchronization-Preserving Races}\label{sec:dichotomy}

In this section we explore the problem of dynamic race prediction beyond sync-preservation.
We show  \cref{thm:dichotomy}, i.e., that even when just two critical sections are present in the input trace,
predicting races with witnesses that might reverse the order of the critical sections becomes intractable.
Our reduction is from the realizability problem of Rf-posets, which we present next.

\Paragraph{Rf-posets.}
An rf-poset is a triplet $\RFPoset=(X,P,\RF)$, where 
$X$ is a set of read and write events,
$P$ defines a partial order $\leq_{P}$ over $X$, and $\RF\colon \Reads{X}\to \Writes{X}$
is a reads-from function that maps every read event of $X$ to a write event of $X$.
Given two distinct events $e_1, e_2\in X$, we write $\Unordered{e_1}{}{P}{e_2}$ to denote $e_1\not<_{P} e_2$ and $e_2\not<_{P} e_1$.
Given a set $Y\subseteq X$, we denote by $P\Project Y$ the \emph{projection} of $P$ on $Y$,
i.e., we have $\leq_{P\Project Y}\subseteq Y\times Y$, and for all $e_1,e_2\in Y$, we have $e_1\leq_{P\Project Y} e_2$ iff $e_1\leq_{P} e_2$.
Given a partial order $Q$ over $X$, we say that $Q$ \emph{refines} $P$ denoted $Q\Refines P$
if for every two events $e_1, e_2\in X$, if $e_1\leq_{P} e_2$ then $e_1\leq_{Q} e_2$.
We consider that each event of $X$ belongs to a unique thread, and there is thread order $\tho{}$ that defines a total order on the events of $X$ that belong to the same thread, and $P$ agrees with $\tho{}$.
The number of threads of $\RFPoset$ is the number of threads of the events of $X$.

\SubParagraph{The Realizability problem of rf-posets.}
Given an rf-poset $\RFPoset=(X,P,\RF)$, the \emph{realizability problem} is to decide whether 
$P$ can be linearized to a total order $\tr$ such that $\RFTrace{\tr}=\RF$.
It has long been known that the problem is $\NP$-complete~\cite{Gibbons97}, 
while it was recently shown that it is even $\W{1}$-hard~\cite{Mathur20}.

Our proof of the lower bound of \cref{thm:dichotomy} is by a two-step reduction.
First we define a variant of the realizability problem for rf-posets, namely \emph{reverse rf-realizability}, and show that it is $\W{1}$-hard when parameterized by the number of threads.
Afterwards, we reduce reverse rf-realizability to the decision problem of dynamic race prediction, which concludes the hardness of the latter.

\Paragraph{Rf-triplets.}
Given an RF-poset $\RFPoset=(X,P,\RF)$, an \emph{rf-triplet} of $\RFPoset$ is a tuple $\DistinguishedTriplet=(\wt, \rd, \wt')$ such that 
(i)~$\rd$ is a read event, 
(ii)~$\RF(\rd)=\wt$, and 
(iii)~$\wt\conf\wt'$.
We refer to $\wt$, $\rd$ and $\wt'$ as the
\emph{write}, \emph{read}, and \emph{interfering write} event of $\DistinguishedTriplet$, respectively.
We denote by $\ReadTriplets{\RFPoset}$ the set of rf-triplets of $\RFPoset$.

We next define a variant of rf-poset realizability, and show that, like the original problem, it is $\W{1}$-hard parameterized by the number of threads.

\Paragraph{Reverse rf-poset realizability.}
The input is a tuple $(\RFPoset, \DistinguishedTriplet, \tr)$, where $\RFPoset=(X,P,\RF)$ is an rf-poset,
$\DistinguishedTriplet=(\ov{\wt}, \ov{\rd}, \ov{\wt}')$ is a distinguished triplet of $\RFPoset$, 
and $\tr$ is a witness to the realizability of $\RFPoset$ such that $\ov{\wt}'<_{\tr}\ov{\wt}$.
The task is to determine whether $\RFPoset$ has a linearization $\tr'$ with $\ov{\wt}<_{\tr'} \ov{\wt}'$.
%i.e., the ordering between $\ov{\wt}$ and $\ov{\wt}'$ is reversed in $\tr'$.
In words, $\RFPoset$ is already realizable by a witness that orders $\ov{\wt}'$ before $\ov{\wt}$, and the task is to decide whether $\RFPoset$ also has a witness in which this order is reversed.

\Paragraph{Hardness of reverse rf-poset realizability.}
We show that the problem is $\W{1}$-hard when parameterized by the number of threads of the rf-poset.
Our reduction is from rf-realizability.
We first present the construction and then argue about its correctness.

\SubParagraph{Construction.}
Consider an rf-poset $\RFPoset=(X,P,\RF)$ with $k$ threads, and we construct an instance of reverse rf-poset realizability $(\RFPoset'=(X',P', \RF'), \DistinguishedTriplet, \tr)$ with $k'=O(k^2)$ threads.
We refer to \cref{fig:rscrm} for an illustration.
For simplicity of presentation, we assume wlog that the following hold.
\begin{compactenum}
\item $X$ contains only the events of the triplets of $\RFPoset$.
\item For every read event $\rd$, we have $\ThreadOf{\rd}=\ThreadOf{\wt}$, i.e., every read observes a local write event.
\end{compactenum}
Let $\{X_i\}_{1\leq i \leq k}$ be a partitioning of $X$ such that each $\SeqTrace_i=P\Project X_i$ is a total order containing all events of thread $i$ (i.e., it is the thread order for thread $i$).
We first construct the rf-poset $\RFPoset'=(X', P', \RF')$.
The threads of $\RFPoset'$ are defined implicitly by the sets of events for each thread.
In particular, $X'$ is partitioned in the sets $X_i$ (which are the events of $\RFPoset$), as well as two sets $X_i^j$ and $Y_i^j$ for each $i,j\in[k]$ with $i\neq j$, where each such $X_i^j$ and $Y_i^j$ contains events of a unique thread of $\RFPoset'$.
Finally, we have two threads containing the events of the distinguished triplet $\DistinguishedTriplet$.
Hence, $\RFPoset'$ has  $k'=k+k\cdot (k-1)+2=k^2+2$ threads.

%!TEX root = ../main.tex

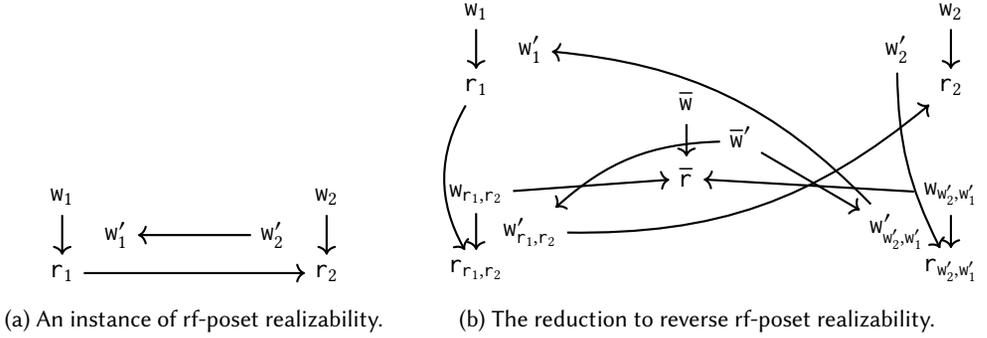
\begin{figure}
\newcommand{\xdisposition}{0}
\newcommand{\ydisposition}{0}
\newcommand{\xtstep}{0.9}
\newcommand{\ytstep}{0.5}
\newcommand{\xstep}{2.8}
\newcommand{\ystep}{1.2}
\newcommand{\xtscale}{0.8}

\def \mycolone {black!20!green}
\def \mycoltwo {black!20!blue}
\def \mycolthree {black!20!orange}
\def \mycolfour {black!20!red}
\def \mycolfive {black!20!black}

\def \mycolone {black}
\def \mycoltwo {black}
\def \mycolthree {black}
\def \mycolfour {black}
\def \mycolfive {black}

%Arguments are name prefix, color, x,y
\newcommand{\drawtriplet}[6]{
\node[] (w#1) at (#4,#5)  {$\textcolor{#3}{\wt_{#2}}$};
\node[] (r#1) at (#4,#5-2*\ytstep)  {$\textcolor{#3}{\rd_{#2}}$};
\node[] (wp#1) at (#4+\xtscale*\xtstep,#5-1*\ytstep)  {$\textcolor{#3}{\wt'_{#2}}$};
\ifthenelse{#6=1}{\draw[->, thick] (w#1) to (r#1);}{}
}
\newcommand{\drawtripletR}[6]{
\node[] (w#1) at (#4+\xtscale*\xtstep,#5)  {$\textcolor{#3}{\wt_{#2}}$};
\node[] (r#1) at (#4+\xtscale*\xtstep,#5-2*\ytstep)  {$\textcolor{#3}{\rd_{#2}}$};
\node[] (wp#1) at (#4,#5-1*\ytstep)  {$\textcolor{#3}{\wt'_{#2}}$};
\ifthenelse{#6=1}{\draw[->, thick] (w#1) to (r#1);}{}
}
\newcommand{\drawtripletdistinguished}[6]{
\node[] (w#1) at (#4,#5)  {$\textcolor{#3}{\ov{\wt}_{#2}}$};
\node[] (r#1) at (#4,#5-2*\ytstep)  {$\textcolor{#3}{\ov{\rd}_{#2}}$};
\node[] (wp#1) at (#4+\xtscale*\xtstep,#5-1*\ytstep)  {$\textcolor{#3}{\ov{\wt}'_{#2}}$};
\ifthenelse{#6=1}{\draw[->, thick] (w#1) to (r#1);}{}
}
\begin{subfigure}[b]{0.45\textwidth}
\centering
\begin{tikzpicture}[thick,
pre/.style={<-,shorten >= 2pt, shorten <=2pt, very thick},
post/.style={->,shorten >= 2pt, shorten <=2pt,  very thick},
seqtrace/.style={->, line width=2},
und/.style={very thick, draw=gray},
event/.style={rectangle, minimum height=0.8mm, minimum width=3mm, fill=black!100,  line width=1pt, inner sep=0},
virt/.style={circle,draw=black!50,fill=black!20, opacity=0}]

\drawtriplet{1}{1}{\mycolone}{0*\xstep}{0*\ystep}{1}
\drawtripletR{2}{2}{\mycoltwo}{1*\xstep}{0*\ystep}{1}
\draw[->, thick] (wp2) to (wp1);
\draw[->, thick] (r1) to (r2);

%\node[] at(0,-6*\ytstep) {};

\end{tikzpicture}
\caption{
An instance of rf-poset realizability.
}
\figlabel{rscrm_reduction1}
\end{subfigure}
\begin{subfigure}[b]{0.5\textwidth}
\centering
\begin{tikzpicture}[thick,
pre/.style={<-,shorten >= 2pt, shorten <=2pt, very thick},
post/.style={->,shorten >= 2pt, shorten <=2pt,  very thick},
seqtrace/.style={->, line width=2},
und/.style={very thick, draw=gray},
event/.style={rectangle, minimum height=0.8mm, minimum width=3mm, fill=black!100,  line width=1pt, inner sep=0},
virt/.style={circle,draw=black!50,fill=black!20, opacity=0}]

\begin{scope}[shift={(-3,0)}]
\drawtriplet{1}{1}{\mycolone}{0*\xstep}{0*\ystep}{1}
\drawtripletR{2}{2}{\mycoltwo}{2*\xstep}{0*\ystep}{1}
\drawtriplet{3}{\rd_{1},\rd_{2}}{\mycolthree}{0*\xstep}{-2*\ystep}{1}
\drawtripletR{4}{\wt'_{2},\wt'_{1}}{\mycolfour}{2*\xstep}{-2*\ystep}{1}
\drawtripletdistinguished{5}{}{\mycolfive}{1*\xstep}{-1*\ystep}{1}
\end{scope}

\draw[->, thick, bend right=30] (r1) to (r3);
\draw[->, thick, bend right=20] (wp3) to (r2);
\draw[->, thick, bend right=13] (wp2) to (r4);
\draw[->, thick, bend right=20] (wp4) to (wp1);
\draw[->, thick, bend right=0] (w3) to (r5);
\draw[->, thick, bend right=0] (w4) to (r5);
\draw[->, thick, bend right=20] (wp5) to (wp3);
\draw[->, thick, bend right=0] (wp5) to (wp4);

\end{tikzpicture}
\caption{
The reduction to reverse rf-poset realizability.
}
\figlabel{rscrm_reduction2}
\end{subfigure}
\caption{
Reduction of rf-poset realizability (\protect\ref{fig:rscrm_reduction1}) to reverse rf-poset realizability (\protect\ref{fig:rscrm_reduction2}).
}
\label{fig:rscrm}
\end{figure}

We first define the set of triplets $\ReadTriplets{\RFPoset}$, which defines the event set $X'$ and the observation function $\RF'$.
We have $X\subseteq X'$ and $\ReadTriplets{\RFPoset}\subseteq \ReadTriplets{\RFPoset'}$.
In addition, we create a distinguished triplet $\DistinguishedTriplet=(\ov{\wt}, \ov{\rd}, \ov{\wt}')$, and all its events are in $X'$.
Finally, for every $i,j\in [k]$ with $i\neq j$, we have $X_i^j, Y_i^j\subseteq X'$, where 
the sets $X_i^j$ and $Y_i^j$ are constructed as follows.
We call a pair of events $(e_1, e_2)\in X_i\times X_j$ with $e_1<_{P}e_2$ \emph{dominant}
if for any pair $(e'_1, e'_2)\in X_i\times X_j$ such that
$e_1\leq_{P} e'_1$, and
$e'_2\leq_{P} e_2$, and
$e'_1<_{P} e'_2$,
we have $e'_i=e_i$ for each $i\in[2]$.
In words, a dominant pair identifies an ordering in $P$ that cannot be inferred transitively by other orderings.
For every dominant pair $(e_1, e_2)\in X_i\times X_j$
we create a triplet $(\wt_{e_1, e_2}, \rd_{e_1, e_2}, \wt'_{e_1, e_2})$,
and let $\wt_{e_1, e_2}, \rd_{e_1, e_2}\in X_i^j$ and $\wt'_{e_1, e_2}\in Y_i^j$.

We now define the partial order $P'$.
For every triplet $(\wt, \rd, \wt')$ of $\RFPoset'$, we have $\wt<_{P'} \rd$.
For every $i,j\in [k]$ with $i\neq j$,
for every two events $e_1, e'_1\in X_i$ such that $e_1<_{P}e'_1$,
for every two events $e_2, e'_2\in X_j$ such that $e_2<_{P}e'_2$,
if $\rd_{e_1, e_2}$ and $\wt_{e'_1, e'_2}$ are events of $X'$ (i.e., $(e_1, e_2)$ and $(e'_1, e'_2)$ are dominant pairs),
we have
(i)~$\rd_{e_1, e_2} <_{P'} \wt_{e'_1, e'_2}$ and
(ii)~$\wt'_{e_1, e_2}<_{P'} \wt'_{e'_1, e'_2}$.
Finally, for every triplet of the form $(\wt_{e_1, e_2}, \rd_{e_1, e_2}, \wt'_{e_1, e_2})$, we have
\begin{align*}
e_1<_{P'}\rd_{e_1, e_2}
\qquad\text{and}\qquad
\wt'_{e_1, e_2}<_{P'}e_2
\qquad\text{and}\qquad
\wt_{e_1, e_2}<_{P'} \ov{\rd}
\qquad\text{and}\qquad
\ov{\wt}'<_{P'} \wt'_{e_1, e_2}\ .
\end{align*}
The following lemma establishes that $P'$ is indeed a partial order.

% \smallskip
\begin{restatable}{lemma}{lemppo}\label{lem:p_po}
$P'$ is a partial order.
\end{restatable}

We now turn our attention to the solution $\tr$ of $\RFPoset'$, which is constructed in two steps.
First, we construct a partial order $Q\Refines P'$ over $X'$ which orders in every triplet the interfering write before the write of the triplet.
That is, for every triplet $(\wt,\rd,\wt')$ of $\RFPoset'$, we have $\wt' <_{Q}\wt$.
Then, we obtain $\tr$ by linearizing $Q$ arbitrarily.
The following lemma states that $\tr$ witnesses the realizability of $\RFPoset'$.

\smallskip
\begin{restatable}{lemma}{lemwitness}\label{lem:witness}
The trace $\tr$ realizes $\RFPoset'$.
\end{restatable}

Observe that the size of $\RFPoset'$ is polynomial  in the size of $\RFPoset$.
The following lemma states the correctness of the reduction.
We refer to \cref{sec:app_dichotomy} for the detailed proof, while here we sketch the argument.

\smallskip
\begin{restatable}{lemma}{lemrscrmhardness}\label{lem:rscrm_hardness}
Reverse rf-poset realizability is $\W{1}$-hard parameterized by the number of threads.
\end{restatable}

\SubParagraph{Correctness.}
We now present the key insight behind the correctness of the reduction.
Consider any dominant pair of events $(e_1, e_2)$ in the initial rf-poset, i.e., we have $e_1<_{P} e_2$.
Observe that the two events are unordered in $P'$.
Now consider any trace $\tr^*$ that solves the reverse rf-poset realizability problem for $\RFPoset'$.
By definition, $\tr^*$ must reverse the order of the two writes of the conflicting triplet, i.e., we must have 
$\ov{\wt}<_{\tr^*}\ov{\wt}'$.
\begin{compactenum}
\item Since $\ov{\wt}<_{\tr^*}\ov{\wt}'$, we  have $\ov{\rd}<_{\tr^*}\ov{\wt}'$, so the last write of $\ov{\rd}$ is not violated in $\tr^*$.
\item Since $\wt_{e_1,e_2} <_{P'} \ov{rd}$, by the previous item we also have transitively $\wt_{e_1,e_2} <_{P'} \ov{\wt}'$, and since $\ov{\wt}'<_{P'}\ov{\wt}'_{e_1, e_2}$,
we have , transitively $\wt_{e_1,e_2}<_{\tr^*} \wt'_{e_1, e_2}$.
\item Since $\wt_{e_1,e_2}<_{\tr^*} \wt'_{e_1, e_2}$, we  have $\rd_{e_1, e_2}<_{\tr^*} \wt'_{e_1, e_2}$, so the last write of $\rd_{e_1, e_2}$ is not violated in $\tr^*$.
\item Finally, since $e_1<{P'} \rd_{e_1, e_2}$ and $\wt'_{e_1, e_2}<e_2$, we also have, transitively, that $e_1<_{\tr^*}e_2$.
\end{compactenum}
Hence, the witness $\tr^*$ also respects the partial order $P$, and thus also serves as a witness of the realizability of $\RFPoset$ (when projected to the set of events $X$).
Thus, if reverse rf-poset realizability holds for $\RFPoset'$, then rf-poset realizability holds for $\RFPoset$.
The inverse direction is similar.

\Paragraph{Hardness of 1-reversal dynamic race prediction.}
We are now ready to prove our second step of the reduction, i.e.,
to establish an FPT reduction from reverse rf-poset realizability to the decision problem of dynamic race prediction.
We first describe the construction.

%!TEX root = ../main.tex

\begin{figure}[t]
% \Andreas{@Umang: please pretify this}
\newcommand{\xdisposition}{0}
\newcommand{\ydisposition}{0}
\newcommand{\xtstep}{0.9}
\newcommand{\ytstep}{0.5}
\newcommand{\xstep}{2.8}
\newcommand{\ystep}{1.2}
\newcommand{\xtscale}{0.75}

\def \mycolone {black!20!green}
\def \mycoltwo {black!20!blue}
\def \mycolthree {black!20!orange}
\def \mycolfour {black!20!red}
\def \mycolfive {black!20!black}

\def \mycolone {black}
\def \mycoltwo {black}
\def \mycolthree {black}
\def \mycolfour {black}
\def \mycolfive {black}

%Arguments are name prefix, color, x,y
\newcommand{\drawtriplet}[6]{
\node[] (w#1) at (#4,#5)  {$\textcolor{#3}{\wt_{#2}}$};
\node[] (r#1) at (#4,#5-2*\ytstep)  {$\textcolor{#3}{\rd_{#2}}$};
\node[] (wp#1) at (#4+\xtscale*\xtstep,#5-1*\ytstep)  {$\textcolor{#3}{\wt'_{#2}}$};
\ifthenelse{#6=1}{\draw[->, thick] (w#1) to (r#1);}{}
}
\newcommand{\drawtripletR}[6]{
\node[] (w#1) at (#4+\xtscale*\xtstep,#5)  {$\textcolor{#3}{\wt_{#2}}$};
\node[] (r#1) at (#4+\xtscale*\xtstep,#5-2*\ytstep)  {$\textcolor{#3}{\rd_{#2}}$};
\node[] (wp#1) at (#4,#5-1*\ytstep)  {$\textcolor{#3}{\wt'_{#2}}$};
\ifthenelse{#6=1}{\draw[->, thick] (w#1) to (r#1);}{}
}
\newcommand{\drawtripletdistinguished}[6]{
\node[] (w#1) at (#4,#5)  {$\textcolor{#3}{\ov{\wt}_{#2}}$};
\node[] (r#1) at (#4,#5-2*\ytstep)  {$\textcolor{#3}{\ov{\rd}_{#2}}$};
\node[] (wp#1) at (#4+\xtscale*\xtstep,#5-1*\ytstep)  {$\textcolor{#3}{\ov{\wt}'_{#2}}$};
\ifthenelse{#6=1}{\draw[->, thick] (w#1) to (r#1);}{}
}

\begin{subfigure}[b]{0.4\textwidth}
\centering
\begin{tikzpicture}
\node[] at (0*\xstep, 0*\ystep){$\DistinguishedTriplet=(\ov{\wt}(x_2), \ov{\rd}(x_2), \ov{\wt}'(x_2))$};
\end{tikzpicture}
\vspace{1cm}

\centering
\execution{4}{
\figev{4}{\ov{\wt}'(x_2)}
\figev{3}{\wt'(x_1)}
\figev{2}{\ov{\wt}(x_2)}
\figev{2}{\ov{\rd}(x_2)}
\figev{1}{\wt(x_1)}
\figev{1}{\rd(x_1)}
%\orderedge{4}{1}{-0.4}{3}{2}{0.4}
\orderedgewithlabel{4}{1}{-0.4}{3}{2}{0.4}{\small $\leq_{P}$}{above}
}
\caption{
An instance of reverse rf-poset realizability $(\RFPoset=(X,P,\RF),\DistinguishedTriplet, \tr)$.
The figure shows $\tr$, and $P$ is defined as the thread order together with the cross-thread ordering $\wt'(x_2)<_{P}\wt'(x_1)$.
}
\figlabel{race_reduction1}
\end{subfigure}
\hspace{0.1in}
\begin{subfigure}[b]{0.5\textwidth}
\centering
\executionlarge{4}{
\figevlarge{4}{\acq_1(\ell)}
\figevlarge{4}{\ov{\wt}'(x_2)}
\figevlarge{4}{\wt\left(x_{\ov{\wt}'(x_2), \wt'(x_1)}\right)}
\figevlarge{4}{\mathbf{\ov{\wt}(y)}}
\figevlarge{4}{\rel_1(\ell)}
\figevlarge{3}{\rd\left(x_{\ov{\wt}'(x_2), \wt'(x_1)}\right)}
\figevlarge{3}{\wt'(x_1)}
\figevlarge{3}{\wt(x^3)}
\figevlarge{2}{\acq_2(\ell)}
\figevlarge{2}{\ov{\wt}(x_2)}
\figevlarge{2}{\rel_2(\ell)}
\figevlarge{2}{\ov{\rd}(x_2)}
\figevlarge{1}{\wt(x_1)}
\figevlarge{1}{\rd(x_1)}
\figevlarge{1}{\wt(x^1)}
\figevlarge{2}{\rd(x^1)}
\figevlarge{2}{\rd(x^3)}
\figevlarge{2}{\mathbf{\ov{\rd}(y)}}
}
\caption{The instance of race prediction using our reduction.}
\figlabel{race_reduction2}
\end{subfigure}
\caption{
Example of our reduction of an instance of reverse rf-poset realizability (\protect\figref{race_reduction1}) to an instance of dynamic data-race prediction (\protect\figref{race_reduction2}) on the event pair $(e_{4}, e_{18})$.
}
\label{fig:race_reduction}
\end{figure}
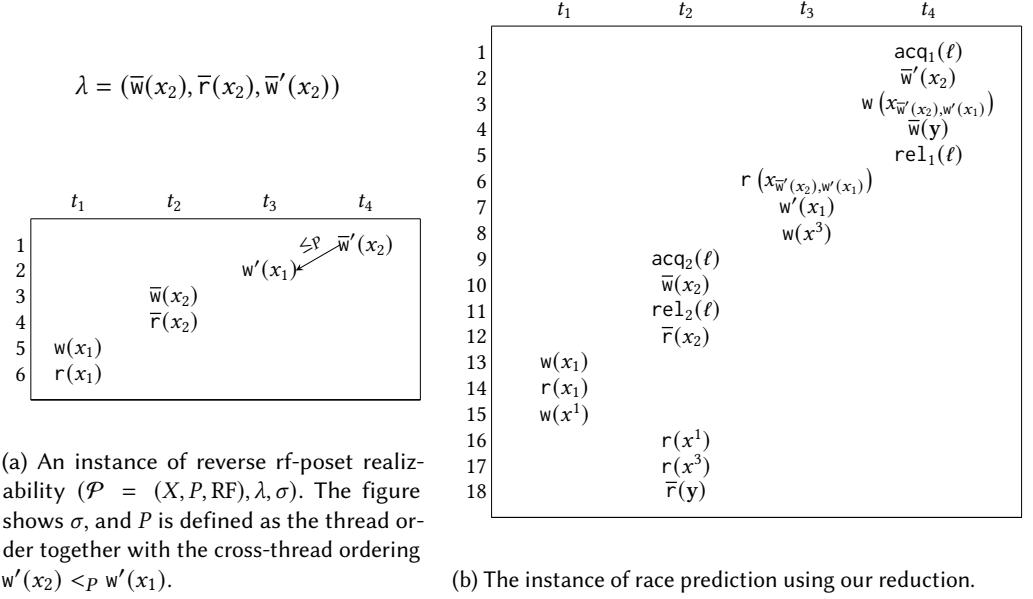

Consider an instance $(\RFPoset=(X,P, \RF), \DistinguishedTriplet=(\ov{\wt},\ov{\rd},\ov{\wt}'), \tr)$ of reverse rf-poset realizability, and we construct a trace $\tr'$ such a specific event pair of $\tr'$ is a predictable race iff $\RFPoset$ is realizable by a witness that reverses $\DistinguishedTriplet$. 
We assume wlog that $X$ contains only events that appear in triplets of $\RFPoset$.
We construct $\tr'$ by inserting various events in $\tr$, as follows.
\cref{fig:race_reduction} provides an illustration.
\begin{compactenum}
\item\label{item:race_red1} For every dominant pair $(e_1, e_2)$ of $\RFPoset$, we introduce a new variable $x_{e_1,e_2}$, and a write event $\wt(x_{e_1, e_2})$ and a read event $\rd(x_{e_1, e_2})$.
We make $\ThreadOf{\wt(x_{e_1, e_2})}=\ThreadOf{e_1}$ and $\ThreadOf{\rd(x_{e_1, e_2})}=\ThreadOf{e_2}$.
Finally, we thread-order $\wt(x_{e_1, e_2})$ after $e_1$ and $\rd(x_{e_1, e_2})$ before $e_2$.
Notice that any correct reordering $\tr^*$ of $\tr'$ must order $\wt(x_{e_1, e_2}) \trord{\tr^*} \rd(x_{e_1, e_2})$, and thus, transitively, also order $e_1 \trord{\tr^*} e_2$.
\item\label{item:race_red2} For every thread $\thread_i\neq \ThreadOf{\ov{\wt}}$, $\thread_i\neq \ThreadOf{\ov{\wt}'}$
we introduce a new variable $x^i$, a write event $\wt(x^i)$, and a read event $\rd(x^i)$.
We make $\ThreadOf{\wt(x^i)}=\thread_i$ and $\ThreadOf{\rd(x^i)}=\ThreadOf{\ov{\rd}}$.
We thread-order each $\wt(x^i)$ as the last event of $\thread_i$, and thread-order all $\rd(x^i)$ as final events of $\ThreadOf{\ov{\rd}}$ so far.
\item\label{item:race_red3} We introduce a new variable $y$, and a write event $\wt(y)$ and a read event $\rd(y)$.
We make $\ThreadOf{\wt(y)}=\ThreadOf{\ov{\wt}'}$ and $\ThreadOf{\rd(y)}=\ThreadOf{\ov{\wt}}$.
Finally, we thread-order $\wt(y)$ and $\rd(y)$ at the end of their respective threads.
In particular, $\rd(y)$ is thread-ordered after the events $\rd(x^i)$ introduced in the previous item.
Notice that because of this ordering and the previous item, any correct reordering $\tr^*$ of $\tr'$ must contain all events of $X'$.
\item\label{item:race_red4} We introduce a lock $\ell$ and two pairs of lock-acquire and lock-release events $(\acq_i(\ell), \rel_i(\ell))$, for each $i\in[2]$.
We make $\ThreadOf{\acq_i(\ell)}=\ThreadOf{\rel_i(\ell)}=\thread_j$, where $\thread_j=\ThreadOf{\ov{\wt}'}$ if $i=1$ and $\thread_j=\ThreadOf{\ov{\wt}}$ otherwise.
Finally, we surround with the critical section of $\acq_1(\ell), \rel_1(\ell)$ all events of the corresponding thread,
and surround with the critical section of $\acq_2(\ell), \rel_2(\ell)$ the event $\ov{\wt}$.
Notice that any correct reordering $\tr^*$ that witnesses a race on $(\wt(y), \rd(y))$ is missing $\rel_1(\ell)$, and thus must order $\rel_2(\ell)\trord{\tr^*} \acq_1(\ell)$. In turn, this leads to a transitive ordering $\ov{\wt}\trord{\tr^*}\ov{\wt}'$, and since the last write of $\ov{\rd}$ must be $\lw{\tr^*}(\ov{\rd})=\ov{\wt}$, we must also have $\ov{\rd}\trord{\tr^*}\ov{\wt}'$.
\end{compactenum}
We now outline the correctness of the reduction (see \cref{sec:app_dichotomy} for the proof).
Consider any correct reordering $\tr^*$ that witnesses a predictable race $(\wt(y), \rd(y))$ on $\tr'$.
\cref{item:race_red2} and \cref{item:race_red3} above guarantee that $X\subseteq  \events{\tr^*}$, while \cref{item:race_red1} guarantees that $\tr^*$ linearizes $P$, and \cref{item:race_red4} guarantees that $\tr^*$ reverses $\DistinguishedTriplet$, i.e., $\ov{\rd}\trord{\tr^*}\ov{\wt}'$.
Finally, note that $\tr'$ has size that is polynomial in $n$, while the number of threads of $\tr'$ equals the number of threads of $\RFPoset$.
This concludes the proof of \cref{thm:dichotomy}.

%!TEX root = main.tex

\section{Experiments}\label{sec:experiments}

In this section we report on an implementation and experimental evaluation of the techniques presented in this work.
Our objective is two-fold.
The first goal is to quantify the practical relevance of sync-preservation, i.e., whether in practice the definition captures races that are missed by the standard notion of happens-before and \textsf{WCP}~\cite{wcp2017}  races.
The second goal is to evaluate the performance of our algorithm $\ZeroRevAlgo$ for detecting sync-preserving races.

\subsection{Experimental Setup}

We have implemented $\ZeroRevAlgo$ (\algoref{overall}) for predicting 
all sync-preserving races in our tool \toolname~\cite{rapid}, written in Java, and evaluated it on a standard set of benchmarks.

\Paragraph{Benchmarks.}
Our benchmark set consists of standard benchmarks found in the recent literature~\cite{rv2014,wcp2017,Yu18,shb2018,Roemer18,PavlogiannisPOPL20}.
It consists of $30$ concurrent programs taken from standard benchmark suites:
\begin{enumerate*}[label=(\roman*)]
\item the IBM Contest benchmark suite~\cite{Farchi2003},
\item the Java Grande forum benchmark suite~\cite{JGF2001},
\item the DaCapo benchmark suite~\cite{DaCapo2006}, 
\item the Software Infrastructure Repository~\cite{doESE05}, and
\item some standalone benchmarks.
\end{enumerate*}
% For each benchmark, we used the tool RV-Predict~\cite{rvpredict} to instrument it and monitor its execution, which created a single trace per program.
% The same trace was used for evaluating all algorithms.
For each benchmark, we generated a single trace using RV-Predict~\cite{rvpredict} and evaluated all methods on the same trace.

\Paragraph{Compared methods.}
We compare our algorithm with state-of-the-art sound race detectors, namely, \textsf{SHB}~\cite{shb2018}, \textsf{WCP}~\cite{wcp2017} and M2~\cite{PavlogiannisPOPL20}.
Recall that \textsf{SHB} and \textsf{WCP} are linear time algorithms that perform a single pass of the input trace $\tr$.
\textsf{SHB} computes happens-before races and is sound even beyond the first race.
On the other hand, \textsf{WCP} is only sound for the first race report.
In order to allow \textsf{WCP} to soundly report more than one race, whenever a race is reported on an event pair $(e_1, e_2)$
(i.e., we have $\Unordered{e_1}{\tr}{\textsf{WCP}}{e_2}$), we force an order $e_1\wcp{\tr} e_2$ before proceeding with the next event of $\tr$.
This is a standard practice that has been followed in other works, e.g.,\cite{Roemer18,PavlogiannisPOPL20}.
Finally, M2 is a more heavyweight algorithm that makes sound reports for all races by design, though its running time is a larger polynomial (of order $n^4$)~\cite{PavlogiannisPOPL20}.

\Paragraph{Optimizations.}
In general, the benchmark traces can be huge and often scale to sizes of order as large as $10^8$. 
A closer inspection shows that many events, even though they perform accesses to shared memory, are non-racy and even totally ordered by fork-join mechanisms and data flows in the trace.
We have implemented a lightweight, linear time, single-pass optimization of the input trace $\tr$ that filters out such events.
The optimization simply identifies memory locations $x$ whose conflicting accesses are totally ordered in $\tr$ by thread and data-flow orderings, and ignores all such accesses in $\tr$.
For a fair comparison, we employ the optimization in all compared methods.
This is similar to FastTrack-like optimizations~\cite{fasttrack}, which identify and ignore thread-local events.
We note that, as each of the compared methods attempts to report as many races as possible, epoch-like optimizations were not applied.

\Paragraph{Reported results.}
Our experiments were conducted on
a 2.6GHz 64-bit Linux machine with Java 1.8 as the JVM and
30GB heap space.
Each of the compared methods is evaluated on the same input trace $\tr$.
For every such input, the respective method reports the following race warnings.
\begin{compactenum}
\item\label{item:race_reports_warnings} \emph{Racy events.} We report the number of events $e_2$ such that there is an event $e_1$ with $e_1 \stricttrord{} e_2$ for which a race $(e_1, e_2)$ is detected.
We remark that this is the standard way of reporting race warnings~\cite{fasttrack,wcp2017,shb2018,Roemer18,Bond2019}, as it allows for one-pass, linear time algorithms that avoid the overhead of testing for races between all possible $\Theta(n^2)$ pairs of events.
%As multiple events $e_2$ may correspond to the same source code line, we also report the number unique source-code locations that are identified as racy.
\item\label{item:race_reports_codelines} \emph{Racy source-code lines.} We report the number of distinct source-code lines which correspond to events $e_2$ that are found as racy in \cref{item:race_reports_warnings}.
This is a meaningful measure, as the same source-code line might be reported by many different events $e_2$.
\item\label{item:race_reports_variables} \emph{Racy memory locations.} We report the number of different memory locations that are accessed by all the events $e_2$ that are found as racy in \cref{item:race_reports_warnings}.
\item\label{item:time_reports} \emph{Running time.} We measure the time of the algorithm required to process each benchmark, while imposing a $1$-hour timeout (TO).
\end{compactenum}

%!TEX root = ../main.tex

\begin{table}
\caption{
Dynamic race reports in our benchmarks.
$\NumEvents$ and $\NumThreads$ denote the number of events and number of threads in the respective trace.
For races, an entry `$r$ ($s$)' denotes the number $r$ of events $e_2$ found to be in race with an earlier event $e_1$,
as well as the number $s$ of unique source-code lines corresponding to such events $e_2$.
%A TO entry denotes a timeout after $1$ hour.
Bold-face entries highlight cases where there are sync-preserving races that are not happens-before races.
}
\label{tab:races}
\setlength\tabcolsep{1.5pt}
\renewcommand{\arraystretch}{1.1}
\footnotesize
\scalebox{1.0}{
\centerline{
\begin{tabular}{|c|c|c|||c|c||c|c||c|c||c|c|}
\hline
\textbf{Benchmark} & $\NumEvents$ & $\NumThreads$ & \multicolumn{2}{c||}{\textsf{SHB}}& \multicolumn{2}{c||}{\textsf{WCP}}& \multicolumn{2}{c||}{\textsf{M2}}& \multicolumn{2}{c|}{$\ZeroRevAlgo$}\\
\hline
&&&\textbf{Races} & \textbf{Time} &\textbf{Races} & \textbf{Time}&\textbf{Races} & \textbf{Time}&\textbf{Races} & \textbf{Time}\\
\hline
\texttt{array} & 51 & 4 & 0 (0) & 0.02s & 0 (0) & 0.03s & 0 (0) & 0.09s & 0 (0) & 0.04s\\
\texttt{critical} & 59 & 5 & 3 (3) & 0.19s & 1 (1) & 0.03s & 3 (3) & 0.11s & 3 (3) & 0.07s\\
\texttt{account} & 134 & 5 & 3 (1) & 0 & 3 (1) & 0.06s & 3 (1) & 0.23s & 3 (1) & 0.09s\\
\texttt{airtickets} & 140 & 5 & 8 (3) & 0.02s & 5 (2) & 0.03s & 8 (3) & 0.13s & 8 (3) & 0.05s\\
\texttt{pingpong} & 151 & 7 & 8 (3) & 1.09s & 8 (3) & 0.04s & 8 (3) & 0.17s & 8 (3) & 0.06s\\
\texttt{twostage} & 193 & 13 & 4 (1) & 0.02s & 4 (1) & 0.09s & 4 (1) & 0.20s & 4 (1) & 0.10s\\
\texttt{wronglock} & 246 & 23 & 12 (2) & 0.02s & 3 (2) & 0.09s & 25 (2) & 0.43s & \textbf{25 (2)} & 0.15s\\
\texttt{bbuffer} & 332 & 3 & 3 (1) & 0.01s & 1 (1) & 0.05s & 3 (1) & 0.11s & 3 (1) & 0.06s\\
\texttt{prodcons} & 658 & 9 & 1 (1) & 0.03s & 1 (1) & 0.09s & 1 (1) & 0.20s & 1 (1) & 0.10s\\
\texttt{clean} & 1.0K & 10 & 59 (4) & 0.04s & 33 (4) & 0.14s & 110 (4) & 0.85s & \textbf{60 (4)} & 0.17s\\
\texttt{mergesort} & 3.0K & 6 & 1 (1) & 11m10s & 1 (1) & 0.12s & 5 (2) & 0.96s & \textbf{3 (1)} & 0.13s\\
\texttt{bubblesort} & 4.0K & 13 & 269 (5) & 0.03s & 100 (5) & 0.27s & 374 (5) & 8.05s & 269 (5) & 0.50s\\
\texttt{lang} & 6.0K & 8 & 400 (1) & 0.10s & 400 (1) & 0.23s & 400 (1) & 1.31s & 400 (1) & 0.31s\\
\texttt{readswrites} & 11K & 6 & 92 (4) & 0.12s & 92 (4) & 0.41s & 228 (4) & 12.74s & \textbf{199 (4)} & 0.77s\\
\texttt{raytracer} & 15K & 4 & 8 (4) & 0.02s & 8 (4) & 0.30s & 8 (4) & 0.40s & 8 (4) & 0.30s\\
\texttt{bufwriter} & 22K & 7 & 8 (4) & 0.10s & 8 (4) & 0.70s & 8 (4) & 2.65s & 8 (4) & 0.84s\\
\texttt{ftpserver} & 49K & 12 & 69 (21) & 6.91s & 69 (21) & 1.34s & 85 (21) & 4.11s & \textbf{85 (21)} & 4.69s\\
\texttt{moldyn} & 200K & 4 & 103 (3) & 0.05s & 103 (3) & 1.83s & 103 (3) & 1m25s & 103 (3) & 1.86s\\
\texttt{linkedlist} & 1.0M & 13 & 5.0K (4) & 7.25s & 5.0K (3) & 27.07s & TO & TO & \textbf{7.0K (4)} & 5m19s\\
\texttt{derby} & 1.0M & 5 & 29 (10) & 0.01s & 28 (10) & 16.48s & 30 (11) & 22.49s & 29 (10) & 24.07s\\
\texttt{jigsaw} & 3.0M & 12 & 4 (4) & 0.41s & 4 (4) & 19.53s & 6 (6) & 11.69s & \textbf{6 (6)} & 17.30s\\
\texttt{sunflow} & 11M & 17 & 84 (6) & 39.66s & 58 (6) & 47.14s & 130 (7) & 50.24s & \textbf{119 (7)} & 55.30s\\
\texttt{cryptorsa} & 58M & 9 & 11 (5) & 3m4s & 11 (5) & 6m35s & TO & TO & \textbf{35 (7)} & 9m42s\\
\texttt{xalan} & 122M & 7 & 31 (10) & 0.15s & 21 (7) & 15m30s & TO & TO & \textbf{37 (12)} & 10m44s\\
\texttt{lufact} & 134M & 5 & 21K (3) & 7m26s & 21K (3) & 14m57s & TO & TO & 21K (3) & 10m38s\\
\texttt{batik} & 157M & 7 & 10 (2) & 9m49s & 10 (2) & 22m56s & TO & TO & 10 (2) & 11m59s\\
\texttt{lusearch} & 217M & 8 & 232 (44) & 12.63s & 119 (27) & 13m40s & 232 (44) & 27m9s & 232 (44) & 14m5s\\
\texttt{tsp} & 307M & 10 & 143 (6) & 15m2s & 140 (6) & 29m10s & TO & TO & 143 (6) & 20m19s\\
\texttt{luindex} & 397M & 3 & 1 (1) & 24m40s & 2 (2) & 31m6s & TO & TO & \textbf{15 (15)} & 31m46s\\
\texttt{sor} & 606M & 5 & 0 (0) & 38m38s & 0 (0) & TO & TO & TO & 0 (0) & 44m36s\\
\hline
\textbf{Totals} & 2.0B & - & 29520 (157) & 1h51m & 29133 (134) & $\geq$ 3h15m & 1846 (131) & $\geq$ 8h30m & \textbf{30862 (178)} & 2h40m \\
\hline
\end{tabular}
}
}
\end{table}
\vspace{-0.2in}

\subsection{Experimental Results}

We now turn our attention to the experimental results.
\cref{tab:races} shows the races and running times reported by each method on each benchmark.

\Paragraph{Coverage of sync-preserving races.}
%We first discuss the coverage of sync-preserving races.
We find that \emph{every} race reported by \textsf{SHB} or \textsf{WCP} is a sync-preserving race, also reported by $\ZeroRevAlgo$.
On the other hand, bold-face entries highlight benchmarks which have sync-preserving races that are not happens-before races.
We see that such races are found in $11$ out of $30$ benchmarks.
Interestingly, in the $5$ most challenging out of these $11$ benchmarks, the same pattern occurs if we focus on source-code lines (i.e., the entries in the parentheses).
Hence, for these benchmarks, sync-preservation is \emph{necessary} to capture many racy source-code lines, which happens-before would completely miss.
We also remark that the more heavyweight analysis \textsf{M2} misses several of these races due to frequent timeouts.
In total, we have $18$ unique source-code lines that are racy but only detected by $\ZeroRevAlgo$.
On the other hand, there are only $2$ source-code lines that are caught by \textsf{M2} but not by $\ZeroRevAlgo$.

\Paragraph{Running times.}
Our experimental times indicate that $\ZeroRevAlgo$ is quite efficient in practice.
Among all algorithms, $\ZeroRevAlgo$ is the second fastest, being about 1.4 times slower that the fastest, lightweight \textsf{SHB}, while at the same time, being able to detect considerably more races the \textsf{SHB} (i.e., $1342$ more racy events, and $21$ more racy source-code lines).
On the other hand, $\ZeroRevAlgo$ detects even more races than \textsf{M2}, due to timeouts, and even has almost equal detection capability with \textsf{M2} on the cases that \textsf{M2} does not time out.
Due to the slow performance of \textsf{M2} (i.e., over 8.5 hours and with several timeouts), we exclude it from the more refined analysis that follows.

%%!TEX root = ../main.tex
%
%\begin{table}
%\caption{
%Numbers of different memory locations that are detected as racy.
%}
%\label{tab:memloc}
%\footnotesize
%\centerline{
%\begin{tabular}{|c||c|c|c|}
%\hline
%\textbf{Benchmark} & \textsf{SHB} & \textsf{WCP} & $\ZeroRevAlgo$ \\
%\hline
%\texttt{ftpserver} & 49 & 49 & 50 \\
%\hline
%\texttt{jigsaw} & 4 & 4 & 5 \\
%\hline
%\texttt{xalan} & 7 & 6 & 9 \\
%\hline
%\texttt{cryptorsa} & 4 & 4 & 5 \\
%\hline
%\texttt{luindex} & 1 & 2 & 9 \\
%\hline
%\texttt{sunflow} & 14 & 10 & 17 \\
%\hline
%\texttt{linkedlist} & 927 & 927 & 932 \\
%%\hline
%%\texttt{hashmap-map-iterator} & 2724 & 2724 & 3771 \\
%\hline
%\hline
%\textbf{Total} & 3730 & 3726 & 4798 \\
%\hline
%\end{tabular}
%}
%\end{table}

\begin{table}
\parbox{.45\linewidth}{
\caption{
Numbers of different memory locations that are detected as racy.
}
\label{tab:memloc}
\footnotesize
\centerline{
\begin{tabular}{|c||c|c|c|}
\hline
\textbf{Benchmark} & \textsf{SHB} & \textsf{WCP} & $\ZeroRevAlgo$ \\
\hline
\texttt{ftpserver} & 49 & 49 & 50 \\
%\hline
\texttt{jigsaw} & 4 & 4 & 5 \\
%\hline
\texttt{xalan} & 7 & 6 & 9 \\
%\hline
\texttt{cryptorsa} & 4 & 4 & 5 \\
%\hline
\texttt{luindex} & 1 & 2 & 9 \\
%\hline
\texttt{sunflow} & 14 & 10 & 17 \\
%\hline
\texttt{linkedlist} & 927 & 927 & 932 \\
%\hline
%\texttt{hashmap-map-iterator} & 2724 & 2724 & 3771 \\
\hline
% \textbf{Total} & 3730 & 3726 & 4798 \\
\textbf{Total} & 1006 & 1002 & 1027 \\
\hline
\end{tabular}
}
}
\hfill
\parbox{.45\linewidth}{
\caption{
Maximum race distances.
}
\label{tab:distances}
\footnotesize
\centerline{
\begin{tabular}{|c|||c|c||c|c||c|c|}
\hline
\textbf{Benchmark} & \multicolumn{1}{c|}{\textsf{SHB}} & \multicolumn{1}{c|}{\textsf{WCP}} & \multicolumn{1}{c|}{$\ZeroRevAlgo$} \\
\hline
\texttt{tsp} & 11K & 11K & 224M\\
%\texttt{lusearch} & 59M & 59M & 125M\\
\texttt{batik} & 1.7M & 1.7M & 4.8M\\
\texttt{cryptorsa} & 7.9M & 7.9M & 8.3M\\
\texttt{jigsaw} & 428 & 428 & 121K\\
\texttt{sunflow} & 10M & 1.0M & 10M\\
\texttt{xalan} & 4K & 4K & 13K\\
%\texttt{bufwriter} & 7.0K & 7.0K & 14K\\
\texttt{ftpserver} & 11K & 11K & 11K\\
\texttt{linkedlist} & 165K & 165K & 165K\\
\texttt{luindex} & 783 & 783 & 6.9K\\
%\texttt{bubblesort} & 2.0K & 2.0K & 3.0K\\
\texttt{mergesort} & 57 & 57 & 1.4K\\
\texttt{clean} & 355 & 47 & 1.2K\\
\texttt{readswrites} & 13 & 13 & 696\\
\texttt{wronglock} & 50 & 6 & 113\\
\hline
\end{tabular}
}
}
\end{table}

\Paragraph{Racy memory locations.}
We next proceed to evaluate the capability of $\ZeroRevAlgo$ in detecting racy memory locations.
As all races detected by \textsf{SHB} or \textsf{WCP} are sync-preserving, the same follows for the racy memory locations, i.e., they are all detected as racy by $\ZeroRevAlgo$.
On the other hand, \cref{tab:memloc} shows a few cases in which $\ZeroRevAlgo$ has discovered racy variables that are missed by \textsf{SHB} and \textsf{WCP}.
Hence, sync-preservation is more adequate to capture not only racy program locations, but also racy memory locations.
We note that, in principle, many different racy memory locations could correspond to the same static race (e.g., if memory is allocated dynamically).
Note, however, that the additional reports of $\ZeroRevAlgo$ in \cref{tab:memloc} occur on benchmarks where it also makes more race reports in \cref{tab:races}.
Together, the two experimental tables give confidence that the new reported races are on entirely different variables.

\Paragraph{Race distances.}
We examine the capability of $\ZeroRevAlgo$ to detect races that are far apart in the input trace.
\cref{tab:distances} shows maximum race distance of races $(e_1, e_2)$ in various benchmarks, including the ones that contains sync-preserving races that are missed by happens-before.
In each case, the distance is counted as the number of events in the input trace between $e_1$ and $e_2$, for every event $e_2$ reported as racy.
We see a sharp contrast between \textsf{SHB}/\textsf{WCP} and $\ZeroRevAlgo$, with the latter being able to detect races that are far more distant in the input.
This is in direct alignment with our theoretical observations already illustrated earlier in \secref{intro} (see \figref{predictable-race-intro}).
Indeed, as partial orders, \textsf{SHB}/\textsf{WCP} can only detect races between conflicting accesses that are successive in the input trace.
On the other hand, sync-preserving races may be interleaved with arbitrarily many conflicting, non-racy accesses,
and our complete algorithm $\ZeroRevAlgo$ is guaranteed to detect them.
Overall, all our experimental observations suggest that sync-preservation is an elegant notion:  it finely characterizes almost all races that are efficiently detectable, while it captures several races that are beyond the standard happens-before relation.

\begin{table}
\caption{
Statistics of the core of the benchmark traces after the lightweight optimization is applied.
$\NumRedEvents$, $\NumRedThreads$, $\NumRedAcquires$, and$\NumRedVariables$ denotes respectively the number of events, threads, lock-acquire events, and variables  in the core trace.
}
\label{tab:stats}
\setlength\tabcolsep{2.5pt}
\renewcommand{\arraystretch}{1.1}
\footnotesize
\centerline{
\begin{tabular}{|c||c|c|c|c|||c||c|c|c|c|||c||c|c|c|c|}
\hline
\textbf{Benchmark} & $\NumRedEvents$ & $\NumRedThreads$ &  $\NumRedAcquires$ &  $\NumRedVariables$ & \textbf{Benchmark}  & $\NumRedEvents$ & $\NumRedThreads$ & $\NumRedAcquires$ &  $\NumRedVariables$  & \textbf{Benchmark}  & $\NumRedEvents$ & $\NumRedThreads$ & $\NumRedAcquires$ &  $\NumRedVariables$ \\
\hline
\texttt{array} & 14 & 4 & 2 & 2 & \texttt{mergesort} & 170 & 6 & 49 & 1 & \texttt{jigsaw} & 3.0K & 12 & 1.0K & 51\\
\texttt{critical} & 14 & 5 & 0 & 1 & \texttt{bubblesort} & 1.0K & 13 & 119 & 25 & \texttt{sunflow} & 3.0K & 17 & 585 & 20\\
\texttt{account} & 18 & 5 & 0 & 1 & \texttt{lang} & 1.0K & 8 & 0 & 100 & \texttt{cryptorsa} & 1.0M & 9 & 156K & 18\\
\texttt{airtickets} & 27 & 5 & 0 & 1 & \texttt{readswrites} & 9.0K & 6 & 1.0K & 6 & \texttt{xalan} & 671K & 7 & 183K & 72\\
\texttt{pingpong} & 38 & 7 & 0 & 2 & \texttt{raytracer} & 529 & 4 & 0 & 3 & \texttt{lufact} & 891K & 5 & 0 & 4\\
\texttt{twostage} & 86 & 13 & 20 & 2 & \texttt{bufwriter} & 10K & 7 & 1.0K & 6 & \texttt{batik} & 132 & 7 & 0 & 5\\
\texttt{wronglock} & 125 & 23 & 20 & 1 & \texttt{ftpserver} & 17K & 12 & 4.0K & 135 & \texttt{lusearch} & 751K & 8 & 53 & 77\\
\texttt{bbuffer} & 13 & 3 & 0 & 1 & \texttt{moldyn} & 21K & 4 & 0 & 2 & \texttt{tsp} & 15M & 10 & 91 & 189\\
\texttt{prodcons} & 248 & 9 & 34 & 3 & \texttt{linkedlist} & 910K & 13 & 1.0K & 932 & \texttt{luindex} & 15K & 3 & 6.0K & 9\\
\texttt{clean} & 871 & 10 & 239 & 2 & \texttt{derby} & 75K & 5 & 21K & 190 & \texttt{sor} & 1.0M & 5 & 633K & 4\\
\hline
\end{tabular}
}
\end{table}
\Paragraph{Complexity of $\ZeroRevAlgo$ and running time.}
Recall the complexity of $\ZeroRevAlgo$ established in \thmref{zero_reversals}.
%i.e., $O(\NumEvents\cdot \NumThreads^2 + \NumAcquires\cdot \NumVariables\cdot \NumThreads^3)$ time
%for $\NumEvents$ events, $\NumAcquires$ acquire events, $\NumThreads$ threads and $\NumVariables$ variables.
We have argued that the complexity is $\Otilde(N)$, i.e., $\NumThreads,\NumVariables=\Otilde(1)$, meaning that the number of threads and variables are much smaller than $N$.
Here we justify this assumption experimentally, by presenting these numbers for the benchmark traces in \cref{tab:stats}.
For each trace $\tr$, 
% we use barred symbols to refer to the corresponding parameters of the core trace $\ov{\tr}$ of $\tr$ as exposed by running the lightweight optimization mentioned in the beginning of this section.
we report the parameters of the core trace $\ov{\tr}$
resulting from our lightweight optimization discussed earlier.
We see that $\ZeroRevAlgo$ (and the other algorithms) is, in reality, executed on the core trace $\ov{\tr}$ where the number of threads $\NumRedThreads$ and variables $\NumRedVariables$ is indeed considerably smaller than $\ov{\NumEvents}$.
Hence, our theoretical treatment of $\NumThreads,\NumVariables=\Otilde(1)$ is justified.

%!TEX root = main.tex

\section{Related Work}
\seclabel{related}

Happens-before (\textsf{HB}) has been the standard approach to sound dynamic race detection.
\textsf{HB} is computable in linear time~\cite{Mattern1988} and has formed the basis of many race detectors in the literature~\cite{Schonberg89,trade,Pozniansky:2003:EOD:966049.781529,fasttrack,bond2010pacer}.
However, \textsf{HB} only characterizes a small subset of predictable races, and recent work 
improves upon this with small increase of computational resources~\cite{cp2012,Roemer18,Bond2019,PavlogiannisPOPL20,wcp2017,Roemer20}.

Another common approach to race prediction is via lockset-based methods.
At a high level, a lockset is a set of locks that guards all accesses to a given memory location.
Such techniques report races when they discover a write access to a location which is not consistently protected (i.e., whose lockset is empty).
They were introduced in~\cite{dinning1991detecting} and equipped in Eraser~\cite{savage1997eraser}.
The lockset criterion is complete but unsound, and various works attempt to reduce false positives by  enhancements such as
random testing~\cite{Sen:2008:RDR:1375581.1375584} and
static analysis~\cite{vonPraun:2001:ORD:504282.504288,Choi02}.
Locksets have also been combined with happens-before techniques~\cite{elmas2007goldilocks,racetrack}.

Another direction to dynamic race prediction is symbolic techniques that typically rely on SAT/SMT encodings of the condition of a correct reordering,
and dispatch such encodings to the respective solver~\cite{Said2011,rv2014,ipa2016,SPA2009}.
The encodings are typically sound and complete in theory, but the solution takes exponential time.
In practice, windowing techniques are used to fragment the trace into chunks and analyze each chunk independently.
This introduces incompleteness, as races between events of different chunks are naturally missed.
Dynamic techniques have also been used for predicting other types of errors, such as deadlocks, atomicity violations and synchronization errors~\cite{Kalhauge18,penelope2010,chen-serbanuta-rosu-2008-icse,sen2005detecting,Farzan2009cav,MathurAtomicity20,velodrome,Farzan2009tacas}.
% ,
% as well as in settings where processes communicate only via locks~\cite{Kahlon05}.

%!TEX root = main.tex

\section{Conclusion}\label{sec:conclusion}

In this work, we have introduced the new notion of synchronization-preserving races.
Conceptually, this is a completion of the principle behind happens-before races, namely that such races can be witnessed without reversing the order in which synchronization operations are observed.
We have shown that sync-preservation strictly subsumes happens-before, and can detect races that are far apart in the input trace.
We have developed an algorithm $\ZeroRevAlgo$ that is sound and complete for sync-preserving races,
and has nearly linear time and space complexity.
%Moreover, we have shown that linear space complexity is necessary for detecting sync-preserving races, and thus our algorithm is optimal.
In addition, we have shown that relaxing our definition even slightly, i.e., by allowing a single synchronization reversal suffices to make the problem $\W{1}$-hard.
Finally, we have performed an extensive experimental evaluation of $\ZeroRevAlgo$.
Our experiments show that sync-preservation is an elegant notion that characterizes almost all races that are efficiently detectable, while it captures several races that are beyond the standard happens-before.
Given the demonstrated relevance of this new notion, we identify as important future work the development of more efficient race detectors for sync-preserving races, in a similar manner that happens-before race detectors have been refined over the years.

%% Acknowledgments
% \begin{acks}                            %% acks environment is optional
% We thank anonymous reviewers for their constructive feedback on an earlier draft of this manuscript.
% Umang Mathur is partially supported by a Google PhD Fellowship.
% Mahesh Viswanathan is partially supported by grants NSF SHF 1901069 and NSF CCF 2007428.
% \end{acks}
\myparagraph{\textsc{Acknowledgments}}{
We thank anonymous reviewers for their constructive feedback on an earlier draft of this manuscript.
Umang Mathur is partially supported by a Google PhD Fellowship.
Mahesh Viswanathan is partially supported by grants NSF SHF 1901069 and NSF CCF 2007428.
}

\clearpage

%% Bibliography
\bibliography{references}


\begin{thebibliography}{62}

%%% the Web address.

\ifx \showCODEN    \undefined \def \showCODEN     #1{\unskip}     \fi
\ifx \showDOI      \undefined \def \showDOI       #1{#1}\fi
\ifx \showISBNx    \undefined \def \showISBNx     #1{\unskip}     \fi
\ifx \showISBNxiii \undefined \def \showISBNxiii  #1{\unskip}     \fi
\ifx \showISSN     \undefined \def \showISSN      #1{\unskip}     \fi
\ifx \showLCCN     \undefined \def \showLCCN      #1{\unskip}     \fi
\ifx \shownote     \undefined \def \shownote      #1{#1}          \fi
\ifx \showarticletitle \undefined \def \showarticletitle #1{#1}   \fi
\ifx \showURL      \undefined \def \showURL       {\relax}        \fi
% invisible to TeX
\providecommand\bibfield[2]{#2}
\providecommand\bibinfo[2]{#2}
\providecommand\natexlab[1]{#1}
\providecommand\showeprint[2][]{arXiv:#2}

\bibitem[\protect\citeauthoryear{Aguado, Mendler, Pouzet, Roop, and von
  Hanxleden}{Aguado et~al\mbox{.}}{2018}]%
        {Aguado18}
\bibfield{author}{\bibinfo{person}{Joaqu{\'i}n Aguado},
  \bibinfo{person}{Michael Mendler}, \bibinfo{person}{Marc Pouzet},
  \bibinfo{person}{Partha Roop}, {and} \bibinfo{person}{Reinhard von
  Hanxleden}.} \bibinfo{year}{2018}\natexlab{}.
\newblock \showarticletitle{Deterministic Concurrency: A Clock-Synchronised
  Shared Memory Approach}. In \bibinfo{booktitle}{\emph{Programming Languages
  and Systems}}, \bibfield{editor}{\bibinfo{person}{Amal Ahmed}} (Ed.).
  \bibinfo{publisher}{Springer International Publishing},
  \bibinfo{address}{Cham}, \bibinfo{pages}{86--113}.
\newblock
\showISBNx{978-3-319-89884-1}


\bibitem[\protect\citeauthoryear{Blackburn, Garner, Hoffmann, Khang, McKinley,
  Bentzur, Diwan, Feinberg, Frampton, Guyer, Hirzel, Hosking, Jump, Lee, Moss,
  Phansalkar, Stefanovi\'{c}, VanDrunen, von Dincklage, and
  Wiedermann}{Blackburn et~al\mbox{.}}{2006}]%
        {DaCapo2006}
\bibfield{author}{\bibinfo{person}{Stephen~M. Blackburn},
  \bibinfo{person}{Robin Garner}, \bibinfo{person}{Chris Hoffmann},
  \bibinfo{person}{Asjad~M. Khang}, \bibinfo{person}{Kathryn~S. McKinley},
  \bibinfo{person}{Rotem Bentzur}, \bibinfo{person}{Amer Diwan},
  \bibinfo{person}{Daniel Feinberg}, \bibinfo{person}{Daniel Frampton},
  \bibinfo{person}{Samuel~Z. Guyer}, \bibinfo{person}{Martin Hirzel},
  \bibinfo{person}{Antony Hosking}, \bibinfo{person}{Maria Jump},
  \bibinfo{person}{Han Lee}, \bibinfo{person}{J.~Eliot~B. Moss},
  \bibinfo{person}{Aashish Phansalkar}, \bibinfo{person}{Darko Stefanovi\'{c}},
  \bibinfo{person}{Thomas VanDrunen}, \bibinfo{person}{Daniel von Dincklage},
  {and} \bibinfo{person}{Ben Wiedermann}.} \bibinfo{year}{2006}\natexlab{}.
\newblock \showarticletitle{The DaCapo Benchmarks: Java Benchmarking
  Development and Analysis}. In \bibinfo{booktitle}{\emph{Proceedings of the
  21st Annual ACM SIGPLAN Conference on Object-oriented Programming Systems,
  Languages, and Applications}} (Portland, Oregon, USA)
  \emph{(\bibinfo{series}{OOPSLA '06})}. \bibinfo{publisher}{ACM},
  \bibinfo{address}{New York, NY, USA}, \bibinfo{pages}{169--190}.
\newblock
\showISBNx{1-59593-348-4}
\urldef\tempurl%
\url{https://doi.org/10.1145/1167473.1167488}
\showDOI{\tempurl}


\bibitem[\protect\citeauthoryear{Bocchino, Adve, Adve, and Snir}{Bocchino
  et~al\mbox{.}}{2009}]%
        {Bocchino09}
\bibfield{author}{\bibinfo{person}{Robert~L. Bocchino},
  \bibinfo{person}{Vikram~S. Adve}, \bibinfo{person}{Sarita~V. Adve}, {and}
  \bibinfo{person}{Marc Snir}.} \bibinfo{year}{2009}\natexlab{}.
\newblock \showarticletitle{Parallel Programming Must Be Deterministic by
  Default}. In \bibinfo{booktitle}{\emph{Proceedings of the First USENIX
  Conference on Hot Topics in Parallelism}} (Berkeley, California)
  \emph{(\bibinfo{series}{HotPar’09})}. \bibinfo{publisher}{USENIX
  Association}, \bibinfo{address}{USA}, \bibinfo{pages}{4}.
\newblock


\bibitem[\protect\citeauthoryear{Boehm}{Boehm}{2011}]%
        {boehmbenign2011}
\bibfield{author}{\bibinfo{person}{Hans-J. Boehm}.}
  \bibinfo{year}{2011}\natexlab{}.
\newblock \showarticletitle{How to Miscompile Programs with “Benign” Data
  Races}. In \bibinfo{booktitle}{\emph{Proceedings of the 3rd USENIX Conference
  on Hot Topic in Parallelism}} (Berkeley, CA)
  \emph{(\bibinfo{series}{HotPar’11})}. \bibinfo{publisher}{USENIX
  Association}, \bibinfo{address}{USA}, \bibinfo{pages}{3}.
\newblock


\bibitem[\protect\citeauthoryear{Boehm}{Boehm}{2012}]%
        {evil2012}
\bibfield{author}{\bibinfo{person}{Hans-J. Boehm}.}
  \bibinfo{year}{2012}\natexlab{}.
\newblock \showarticletitle{Position Paper: Nondeterminism is Unavoidable, but
  Data Races Are Pure Evil}. In \bibinfo{booktitle}{\emph{Proceedings of the
  2012 ACM Workshop on Relaxing Synchronization for Multicore and Manycore
  Scalability}} (Tucson, Arizona, USA) \emph{(\bibinfo{series}{RACES ’12})}.
  \bibinfo{publisher}{Association for Computing Machinery},
  \bibinfo{address}{New York, NY, USA}, \bibinfo{pages}{9–14}.
\newblock
\showISBNx{9781450316323}
\urldef\tempurl%
\url{https://doi.org/10.1145/2414729.2414732}
\showDOI{\tempurl}


\bibitem[\protect\citeauthoryear{Bond, Coons, and McKinley}{Bond
  et~al\mbox{.}}{2010}]%
        {bond2010pacer}
\bibfield{author}{\bibinfo{person}{Michael~D. Bond},
  \bibinfo{person}{Katherine~E. Coons}, {and} \bibinfo{person}{Kathryn~S.
  McKinley}.} \bibinfo{year}{2010}\natexlab{}.
\newblock \showarticletitle{PACER: Proportional Detection of Data Races}. In
  \bibinfo{booktitle}{\emph{Proceedings of the 31st ACM SIGPLAN Conference on
  Programming Language Design and Implementation}} (Toronto, Ontario, Canada)
  \emph{(\bibinfo{series}{PLDI '10})}. \bibinfo{publisher}{ACM},
  \bibinfo{address}{New York, NY, USA}, \bibinfo{pages}{255--268}.
\newblock
\showISBNx{978-1-4503-0019-3}
\urldef\tempurl%
\url{https://doi.org/10.1145/1806596.1806626}
\showDOI{\tempurl}


\bibitem[\protect\citeauthoryear{Chen, {\c S}erb{\u a}nu{\c t}{\u a}, and Ro{\c
  s}u}{Chen et~al\mbox{.}}{2008}]%
        {chen-serbanuta-rosu-2008-icse}
\bibfield{author}{\bibinfo{person}{Feng Chen}, \bibinfo{person}{Traian~Florin
  {\c S}erb{\u a}nu{\c t}{\u a}}, {and} \bibinfo{person}{Grigore Ro{\c s}u}.}
  \bibinfo{year}{2008}\natexlab{}.
\newblock \showarticletitle{{{jPredictor}: a predictive runtime analysis tool
  for {Java}}}. In \bibinfo{booktitle}{\emph{ICSE '08: Proceedings of the 30th
  International Conference on Software Engineering}} (Leipzig, Germany).
  \bibinfo{publisher}{ACM}, \bibinfo{address}{New York, NY, USA},
  \bibinfo{pages}{221--230}.
\newblock


\bibitem[\protect\citeauthoryear{Choi, Lee, Loginov, O'Callahan, Sarkar, and
  Sridharan}{Choi et~al\mbox{.}}{2002}]%
        {Choi02}
\bibfield{author}{\bibinfo{person}{Jong-Deok Choi}, \bibinfo{person}{Keunwoo
  Lee}, \bibinfo{person}{Alexey Loginov}, \bibinfo{person}{Robert O'Callahan},
  \bibinfo{person}{Vivek Sarkar}, {and} \bibinfo{person}{Manu Sridharan}.}
  \bibinfo{year}{2002}\natexlab{}.
\newblock \showarticletitle{Efficient and Precise Datarace Detection for
  Multithreaded Object-oriented Programs}. In
  \bibinfo{booktitle}{\emph{Proceedings of the ACM SIGPLAN 2002 Conference on
  Programming Language Design and Implementation}} (Berlin, Germany)
  \emph{(\bibinfo{series}{PLDI '02})}. \bibinfo{publisher}{ACM},
  \bibinfo{address}{New York, NY, USA}, \bibinfo{pages}{258--269}.
\newblock
\showISBNx{1-58113-463-0}
\urldef\tempurl%
\url{https://doi.org/10.1145/512529.512560}
\showDOI{\tempurl}


\bibitem[\protect\citeauthoryear{Christiaens and Bosschere}{Christiaens and
  Bosschere}{2001}]%
        {trade}
\bibfield{author}{\bibinfo{person}{Mark Christiaens} {and}
  \bibinfo{person}{Koenraad~De Bosschere}.} \bibinfo{year}{2001}\natexlab{}.
\newblock \showarticletitle{{TRaDe: Data Race Detection for Java}}. In
  \bibinfo{booktitle}{\emph{Proceedings of the International Conference on
  Computational Science-Part II}} \emph{(\bibinfo{series}{ICCS '01})}.
  \bibinfo{publisher}{Springer-Verlag}, \bibinfo{address}{London, UK, UK},
  \bibinfo{pages}{761--770}.
\newblock


\bibitem[\protect\citeauthoryear{Cui, Gu, Liu, Chen, and Yang}{Cui
  et~al\mbox{.}}{2015}]%
        {Cui15}
\bibfield{author}{\bibinfo{person}{Heming Cui}, \bibinfo{person}{Rui Gu},
  \bibinfo{person}{Cheng Liu}, \bibinfo{person}{Tianyu Chen}, {and}
  \bibinfo{person}{Junfeng Yang}.} \bibinfo{year}{2015}\natexlab{}.
\newblock \showarticletitle{Paxos Made Transparent}. In
  \bibinfo{booktitle}{\emph{Proceedings of the 25th Symposium on Operating
  Systems Principles}} (Monterey, California) \emph{(\bibinfo{series}{SOSP
  ’15})}. \bibinfo{publisher}{Association for Computing Machinery},
  \bibinfo{address}{New York, NY, USA}, \bibinfo{pages}{105–120}.
\newblock
\showISBNx{9781450338349}
\urldef\tempurl%
\url{https://doi.org/10.1145/2815400.2815427}
\showDOI{\tempurl}


\bibitem[\protect\citeauthoryear{Dinning and Schonberg}{Dinning and
  Schonberg}{1991}]%
        {dinning1991detecting}
\bibfield{author}{\bibinfo{person}{Anne Dinning} {and} \bibinfo{person}{Edith
  Schonberg}.} \bibinfo{year}{1991}\natexlab{}.
\newblock \showarticletitle{Detecting Access Anomalies in Programs with
  Critical Sections}. In \bibinfo{booktitle}{\emph{Proceedings of the 1991
  ACM/ONR Workshop on Parallel and Distributed Debugging}} (Santa Cruz,
  California, USA) \emph{(\bibinfo{series}{PADD '91})}.
  \bibinfo{publisher}{ACM}, \bibinfo{address}{New York, NY, USA},
  \bibinfo{pages}{85--96}.
\newblock
\showISBNx{0-89791-457-0}
\urldef\tempurl%
\url{https://doi.org/10.1145/122759.122767}
\showDOI{\tempurl}


\bibitem[\protect\citeauthoryear{Do, Elbaum, and Rothermel}{Do
  et~al\mbox{.}}{2005}]%
        {doESE05}
\bibfield{author}{\bibinfo{person}{Hyunsook Do}, \bibinfo{person}{Sebastian~G.
  Elbaum}, {and} \bibinfo{person}{Gregg Rothermel}.}
  \bibinfo{year}{2005}\natexlab{}.
\newblock \showarticletitle{Supporting Controlled Experimentation with Testing
  Techniques: An Infrastructure and its Potential Impact.}
\newblock \bibinfo{journal}{\emph{Empirical Software Engineering: An
  International Journal}} \bibinfo{volume}{10}, \bibinfo{number}{4}
  (\bibinfo{year}{2005}), \bibinfo{pages}{405--435}.
\newblock


\bibitem[\protect\citeauthoryear{Elmas, Qadeer, and Tasiran}{Elmas
  et~al\mbox{.}}{2007}]%
        {elmas2007goldilocks}
\bibfield{author}{\bibinfo{person}{Tayfun Elmas}, \bibinfo{person}{Shaz
  Qadeer}, {and} \bibinfo{person}{Serdar Tasiran}.}
  \bibinfo{year}{2007}\natexlab{}.
\newblock \showarticletitle{Goldilocks: A Race and Transaction-aware Java
  Runtime}. In \bibinfo{booktitle}{\emph{Proceedings of the 28th ACM SIGPLAN
  Conference on Programming Language Design and Implementation}} (San Diego,
  California, USA) \emph{(\bibinfo{series}{PLDI '07})}.
  \bibinfo{publisher}{ACM}, \bibinfo{address}{New York, NY, USA},
  \bibinfo{pages}{245--255}.
\newblock
\showISBNx{978-1-59593-633-2}
\urldef\tempurl%
\url{https://doi.org/10.1145/1250734.1250762}
\showDOI{\tempurl}


\bibitem[\protect\citeauthoryear{Farchi, Nir, and Ur}{Farchi
  et~al\mbox{.}}{2003}]%
        {Farchi2003}
\bibfield{author}{\bibinfo{person}{Eitan Farchi}, \bibinfo{person}{Yarden Nir},
  {and} \bibinfo{person}{Shmuel Ur}.} \bibinfo{year}{2003}\natexlab{}.
\newblock \showarticletitle{{Concurrent Bug Patterns and How to Test Them}}. In
  \bibinfo{booktitle}{\emph{Proceedings of the 17th International Symposium on
  Parallel and Distributed Processing}} \emph{(\bibinfo{series}{IPDPS '03})}.
  \bibinfo{publisher}{IEEE Computer Society}, \bibinfo{address}{Washington, DC,
  USA}, \bibinfo{pages}{286.2--}.
\newblock


\bibitem[\protect\citeauthoryear{Farzan and Madhusudan}{Farzan and
  Madhusudan}{2009}]%
        {Farzan2009tacas}
\bibfield{author}{\bibinfo{person}{Azadeh Farzan} {and} \bibinfo{person}{P.
  Madhusudan}.} \bibinfo{year}{2009}\natexlab{}.
\newblock \showarticletitle{The Complexity of Predicting Atomicity Violations}.
  In \bibinfo{booktitle}{\emph{Proceedings of the 15th International Conference
  on Tools and Algorithms for the Construction and Analysis of Systems: Held As
  Part of the Joint European Conferences on Theory and Practice of Software,
  ETAPS 2009,}} (York, UK) \emph{(\bibinfo{series}{TACAS '09})}.
  \bibinfo{publisher}{Springer-Verlag}, \bibinfo{address}{Berlin, Heidelberg},
  \bibinfo{pages}{155--169}.
\newblock
\showISBNx{978-3-642-00767-5}
\urldef\tempurl%
\url{https://doi.org/10.1007/978-3-642-00768-2_14}
\showDOI{\tempurl}


\bibitem[\protect\citeauthoryear{Farzan, Madhusudan, and Sorrentino}{Farzan
  et~al\mbox{.}}{2009}]%
        {Farzan2009cav}
\bibfield{author}{\bibinfo{person}{Azadeh Farzan}, \bibinfo{person}{P.
  Madhusudan}, {and} \bibinfo{person}{Francesco Sorrentino}.}
  \bibinfo{year}{2009}\natexlab{}.
\newblock \showarticletitle{Meta-analysis for Atomicity Violations Under Nested
  Locking}. In \bibinfo{booktitle}{\emph{Proceedings of the 21st International
  Conference on Computer Aided Verification}} (Grenoble, France)
  \emph{(\bibinfo{series}{CAV '09})}. \bibinfo{publisher}{Springer-Verlag},
  \bibinfo{address}{Berlin, Heidelberg}, \bibinfo{pages}{248--262}.
\newblock
\showISBNx{978-3-642-02657-7}
\urldef\tempurl%
\url{https://doi.org/10.1007/978-3-642-02658-4_21}
\showDOI{\tempurl}


\bibitem[\protect\citeauthoryear{Fidge}{Fidge}{1991}]%
        {Fidge:1991:LTD:112827.112860}
\bibfield{author}{\bibinfo{person}{Colin Fidge}.}
  \bibinfo{year}{1991}\natexlab{}.
\newblock \showarticletitle{Logical Time in Distributed Computing Systems}.
\newblock \bibinfo{journal}{\emph{Computer}} \bibinfo{volume}{24},
  \bibinfo{number}{8} (\bibinfo{date}{Aug.} \bibinfo{year}{1991}),
  \bibinfo{pages}{28--33}.
\newblock
\showISSN{0018-9162}
\urldef\tempurl%
\url{https://doi.org/10.1109/2.84874}
\showDOI{\tempurl}


\bibitem[\protect\citeauthoryear{Flanagan and Freund}{Flanagan and
  Freund}{2009}]%
        {fasttrack}
\bibfield{author}{\bibinfo{person}{Cormac Flanagan} {and}
  \bibinfo{person}{Stephen~N. Freund}.} \bibinfo{year}{2009}\natexlab{}.
\newblock \showarticletitle{FastTrack: Efficient and Precise Dynamic Race
  Detection}. In \bibinfo{booktitle}{\emph{Proceedings of the 30th ACM SIGPLAN
  Conference on Programming Language Design and Implementation}} (Dublin,
  Ireland) \emph{(\bibinfo{series}{PLDI '09})}. \bibinfo{publisher}{ACM},
  \bibinfo{address}{New York, NY, USA}, \bibinfo{pages}{121--133}.
\newblock
\showISBNx{978-1-60558-392-1}
\urldef\tempurl%
\url{https://doi.org/10.1145/1542476.1542490}
\showDOI{\tempurl}


\bibitem[\protect\citeauthoryear{Flanagan, Freund, and Yi}{Flanagan
  et~al\mbox{.}}{2008}]%
        {velodrome}
\bibfield{author}{\bibinfo{person}{Cormac Flanagan},
  \bibinfo{person}{Stephen~N. Freund}, {and} \bibinfo{person}{Jaeheon Yi}.}
  \bibinfo{year}{2008}\natexlab{}.
\newblock \showarticletitle{Velodrome: {A} Sound and Complete Dynamic Atomicity
  Checker for Multithreaded Programs}. In \bibinfo{booktitle}{\emph{Proceedings
  of the 29th ACM SIGPLAN Conference on Programming Language Design and
  Implementation}} (Tucson, AZ, USA) \emph{(\bibinfo{series}{PLDI '08})}.
  \bibinfo{publisher}{ACM}, \bibinfo{address}{New York, NY, USA},
  \bibinfo{pages}{293--303}.
\newblock
\showISBNx{978-1-59593-860-2}
\urldef\tempurl%
\url{https://doi.org/10.1145/1375581.1375618}
\showDOI{\tempurl}


\bibitem[\protect\citeauthoryear{Gen\c{c}, Roemer, Xu, and Bond}{Gen\c{c}
  et~al\mbox{.}}{2019}]%
        {Bond2019}
\bibfield{author}{\bibinfo{person}{Kaan Gen\c{c}}, \bibinfo{person}{Jake
  Roemer}, \bibinfo{person}{Yufan Xu}, {and} \bibinfo{person}{Michael~D.
  Bond}.} \bibinfo{year}{2019}\natexlab{}.
\newblock \showarticletitle{Dependence-Aware, Unbounded Sound Predictive Race
  Detection}.
\newblock \bibinfo{journal}{\emph{Proc. ACM Program. Lang.}}
  \bibinfo{volume}{3}, \bibinfo{number}{OOPSLA}, Article
  \bibinfo{articleno}{179} (\bibinfo{date}{Oct.} \bibinfo{year}{2019}),
  \bibinfo{numpages}{30}~pages.
\newblock
\urldef\tempurl%
\url{https://doi.org/10.1145/3360605}
\showDOI{\tempurl}


\bibitem[\protect\citeauthoryear{Gen\c{c}, Xu, and Bond}{Gen\c{c}
  et~al\mbox{.}}{2020}]%
        {sdpUnsound2020}
\bibfield{author}{\bibinfo{person}{Kaan Gen\c{c}}, \bibinfo{person}{Yufan Xu},
  {and} \bibinfo{person}{Michael~D. Bond}.} \bibinfo{year}{2020}\natexlab{}.
\newblock \showarticletitle{Personal Communication}.
\newblock  (\bibinfo{year}{2020}).
\newblock


\bibitem[\protect\citeauthoryear{Gibbons and Korach}{Gibbons and
  Korach}{1997}]%
        {Gibbons97}
\bibfield{author}{\bibinfo{person}{Phillip~B. Gibbons} {and}
  \bibinfo{person}{Ephraim Korach}.} \bibinfo{year}{1997}\natexlab{}.
\newblock \showarticletitle{Testing Shared Memories}.
\newblock \bibinfo{journal}{\emph{SIAM J. Comput.}} \bibinfo{volume}{26},
  \bibinfo{number}{4} (\bibinfo{date}{Aug.} \bibinfo{year}{1997}),
  \bibinfo{pages}{1208--1244}.
\newblock
\showISSN{0097-5397}
\urldef\tempurl%
\url{https://doi.org/10.1137/S0097539794279614}
\showDOI{\tempurl}


\bibitem[\protect\citeauthoryear{Gorogiannis, O’Hearn, and
  Sergey}{Gorogiannis et~al\mbox{.}}{2019}]%
        {racerdx2019}
\bibfield{author}{\bibinfo{person}{Nikos Gorogiannis},
  \bibinfo{person}{Peter~W. O’Hearn}, {and} \bibinfo{person}{Ilya Sergey}.}
  \bibinfo{year}{2019}\natexlab{}.
\newblock \showarticletitle{A True Positives Theorem for a Static Race
  Detector}.
\newblock \bibinfo{journal}{\emph{Proc. ACM Program. Lang.}}
  \bibinfo{volume}{3}, \bibinfo{number}{POPL}, Article \bibinfo{articleno}{57}
  (\bibinfo{date}{Jan.} \bibinfo{year}{2019}), \bibinfo{numpages}{29}~pages.
\newblock
\urldef\tempurl%
\url{https://doi.org/10.1145/3290370}
\showDOI{\tempurl}


\bibitem[\protect\citeauthoryear{Herlihy and Wing}{Herlihy and Wing}{1990}]%
        {HerlihyWing90}
\bibfield{author}{\bibinfo{person}{Maurice~P. Herlihy} {and}
  \bibinfo{person}{Jeannette~M. Wing}.} \bibinfo{year}{1990}\natexlab{}.
\newblock \showarticletitle{Linearizability: A Correctness Condition for
  Concurrent Objects}.
\newblock \bibinfo{journal}{\emph{ACM Trans. Program. Lang. Syst.}}
  \bibinfo{volume}{12}, \bibinfo{number}{3} (\bibinfo{date}{July}
  \bibinfo{year}{1990}), \bibinfo{pages}{463–492}.
\newblock
\showISSN{0164-0925}
\urldef\tempurl%
\url{https://doi.org/10.1145/78969.78972}
\showDOI{\tempurl}


\bibitem[\protect\citeauthoryear{Huang, Meredith, and Rosu}{Huang
  et~al\mbox{.}}{2014}]%
        {rv2014}
\bibfield{author}{\bibinfo{person}{Jeff Huang}, \bibinfo{person}{Patrick~O'Neil
  Meredith}, {and} \bibinfo{person}{Grigore Rosu}.}
  \bibinfo{year}{2014}\natexlab{}.
\newblock \showarticletitle{Maximal Sound Predictive Race Detection with
  Control Flow Abstraction}. In \bibinfo{booktitle}{\emph{Proceedings of the
  35th ACM SIGPLAN Conference on Programming Language Design and
  Implementation}} (Edinburgh, United Kingdom) \emph{(\bibinfo{series}{PLDI
  '14})}. \bibinfo{publisher}{ACM}, \bibinfo{address}{New York, NY, USA},
  \bibinfo{pages}{337--348}.
\newblock
\showISBNx{978-1-4503-2784-8}
\urldef\tempurl%
\url{https://doi.org/10.1145/2594291.2594315}
\showDOI{\tempurl}


\bibitem[\protect\citeauthoryear{Huang and Rajagopalan}{Huang and
  Rajagopalan}{2016}]%
        {Huang2016}
\bibfield{author}{\bibinfo{person}{Jeff Huang} {and} \bibinfo{person}{Arun~K.
  Rajagopalan}.} \bibinfo{year}{2016}\natexlab{}.
\newblock \showarticletitle{Precise and Maximal Race Detection from Incomplete
  Traces}. In \bibinfo{booktitle}{\emph{Proceedings of the 2016 ACM SIGPLAN
  International Conference on Object-Oriented Programming, Systems, Languages,
  and Applications}} (Amsterdam, Netherlands) \emph{(\bibinfo{series}{OOPSLA
  2016})}. \bibinfo{publisher}{ACM}, \bibinfo{address}{New York, NY, USA},
  \bibinfo{pages}{462--476}.
\newblock
\showISBNx{978-1-4503-4444-9}
\urldef\tempurl%
\url{https://doi.org/10.1145/2983990.2984024}
\showDOI{\tempurl}


\bibitem[\protect\citeauthoryear{Kalhauge and Palsberg}{Kalhauge and
  Palsberg}{2018}]%
        {Kalhauge18}
\bibfield{author}{\bibinfo{person}{Christian~Gram Kalhauge} {and}
  \bibinfo{person}{Jens Palsberg}.} \bibinfo{year}{2018}\natexlab{}.
\newblock \showarticletitle{Sound Deadlock Prediction}.
\newblock \bibinfo{journal}{\emph{Proc. ACM Program. Lang.}}
  \bibinfo{volume}{2}, \bibinfo{number}{OOPSLA}, Article
  \bibinfo{articleno}{146} (\bibinfo{date}{Oct.} \bibinfo{year}{2018}),
  \bibinfo{numpages}{29}~pages.
\newblock
\urldef\tempurl%
\url{https://doi.org/10.1145/3276516}
\showDOI{\tempurl}


\bibitem[\protect\citeauthoryear{Kasikci, Zamfir, and Candea}{Kasikci
  et~al\mbox{.}}{2013}]%
        {racemob}
\bibfield{author}{\bibinfo{person}{Baris Kasikci}, \bibinfo{person}{Cristian
  Zamfir}, {and} \bibinfo{person}{George Candea}.}
  \bibinfo{year}{2013}\natexlab{}.
\newblock \showarticletitle{{RaceMob: Crowdsourced Data Race Detection}}. In
  \bibinfo{booktitle}{\emph{Proceedings of the Twenty-Fourth ACM Symposium on
  Operating Systems Principles}} (Farminton, Pennsylvania)
  \emph{(\bibinfo{series}{SOSP '13})}. \bibinfo{publisher}{ACM},
  \bibinfo{address}{New York, NY, USA}, \bibinfo{pages}{406--422}.
\newblock


\bibitem[\protect\citeauthoryear{Kini, Mathur, and Viswanathan}{Kini
  et~al\mbox{.}}{2017}]%
        {wcp2017}
\bibfield{author}{\bibinfo{person}{Dileep Kini}, \bibinfo{person}{Umang
  Mathur}, {and} \bibinfo{person}{Mahesh Viswanathan}.}
  \bibinfo{year}{2017}\natexlab{}.
\newblock \showarticletitle{Dynamic Race Prediction in Linear Time}. In
  \bibinfo{booktitle}{\emph{Proceedings of the 38th ACM SIGPLAN Conference on
  Programming Language Design and Implementation}} (Barcelona, Spain)
  \emph{(\bibinfo{series}{PLDI 2017})}. \bibinfo{publisher}{ACM},
  \bibinfo{address}{New York, NY, USA}, \bibinfo{pages}{157--170}.
\newblock
\showISBNx{978-1-4503-4988-8}
\urldef\tempurl%
\url{https://doi.org/10.1145/3062341.3062374}
\showDOI{\tempurl}


\bibitem[\protect\citeauthoryear{Lamport}{Lamport}{1978}]%
        {lamport1978time}
\bibfield{author}{\bibinfo{person}{Leslie Lamport}.}
  \bibinfo{year}{1978}\natexlab{}.
\newblock \showarticletitle{{Time, Clocks, and the Ordering of Events in a
  Distributed System}}.
\newblock \bibinfo{journal}{\emph{Commun. ACM}} \bibinfo{volume}{21},
  \bibinfo{number}{7} (\bibinfo{date}{July} \bibinfo{year}{1978}),
  \bibinfo{pages}{558--565}.
\newblock


\bibitem[\protect\citeauthoryear{Liu, Tripp, and Zhang}{Liu
  et~al\mbox{.}}{2016}]%
        {ipa2016}
\bibfield{author}{\bibinfo{person}{Peng Liu}, \bibinfo{person}{Omer Tripp},
  {and} \bibinfo{person}{Xiangyu Zhang}.} \bibinfo{year}{2016}\natexlab{}.
\newblock \showarticletitle{IPA: Improving Predictive Analysis with Pointer
  Analysis}. In \bibinfo{booktitle}{\emph{Proceedings of the 25th International
  Symposium on Software Testing and Analysis}} (Saarbr\&\#252;cken, Germany)
  \emph{(\bibinfo{series}{ISSTA 2016})}. \bibinfo{publisher}{ACM},
  \bibinfo{address}{New York, NY, USA}, \bibinfo{pages}{59--69}.
\newblock


\bibitem[\protect\citeauthoryear{Lu, Park, Seo, and Zhou}{Lu
  et~al\mbox{.}}{2008}]%
        {lpsz08}
\bibfield{author}{\bibinfo{person}{Shan Lu}, \bibinfo{person}{Soyeon Park},
  \bibinfo{person}{Eunsoo Seo}, {and} \bibinfo{person}{Yuanyuan Zhou}.}
  \bibinfo{year}{2008}\natexlab{}.
\newblock \showarticletitle{Learning from Mistakes: A Comprehensive Study on
  Real World Concurrency Bug Characteristics}. In
  \bibinfo{booktitle}{\emph{Proceedings of the 13th International Conference on
  Architectural Support for Programming Languages and Operating Systems}}
  (Seattle, WA, USA) \emph{(\bibinfo{series}{ASPLOS XIII})}.
  \bibinfo{publisher}{ACM}, \bibinfo{address}{New York, NY, USA},
  \bibinfo{pages}{329--339}.
\newblock
\showISBNx{978-1-59593-958-6}
\urldef\tempurl%
\url{https://doi.org/10.1145/1346281.1346323}
\showDOI{\tempurl}


\bibitem[\protect\citeauthoryear{Mathur}{Mathur}{2020}]%
        {rapid}
\bibfield{author}{\bibinfo{person}{Umang Mathur}.}
  \bibinfo{year}{2020}\natexlab{}.
\newblock \bibinfo{booktitle}{\emph{{{RAPID}}}}.
\newblock
\urldef\tempurl%
\url{https://github.com/umangm/rapid}
\showURL{%
\tempurl}
\newblock
\shownote{Accessed: 2020-10-25.}


\bibitem[\protect\citeauthoryear{Mathur, Kini, and Viswanathan}{Mathur
  et~al\mbox{.}}{2018}]%
        {shb2018}
\bibfield{author}{\bibinfo{person}{Umang Mathur}, \bibinfo{person}{Dileep
  Kini}, {and} \bibinfo{person}{Mahesh Viswanathan}.}
  \bibinfo{year}{2018}\natexlab{}.
\newblock \showarticletitle{What Happens-after the First Race? Enhancing the
  Predictive Power of Happens-before Based Dynamic Race Detection}.
\newblock \bibinfo{journal}{\emph{Proc. ACM Program. Lang.}}
  \bibinfo{volume}{2}, \bibinfo{number}{OOPSLA}, Article
  \bibinfo{articleno}{145} (\bibinfo{date}{Oct.} \bibinfo{year}{2018}),
  \bibinfo{numpages}{29}~pages.
\newblock
\showISSN{2475-1421}
\urldef\tempurl%
\url{https://doi.org/10.1145/3276515}
\showDOI{\tempurl}


\bibitem[\protect\citeauthoryear{Mathur, Pavlogiannis, and Viswanathan}{Mathur
  et~al\mbox{.}}{2020}]%
        {Mathur20}
\bibfield{author}{\bibinfo{person}{Umang Mathur}, \bibinfo{person}{Andreas
  Pavlogiannis}, {and} \bibinfo{person}{Mahesh Viswanathan}.}
  \bibinfo{year}{2020}\natexlab{}.
\newblock \showarticletitle{The Complexity of Dynamic Data Race Prediction}. In
  \bibinfo{booktitle}{\emph{Proceedings of the 35th Annual ACM/IEEE Symposium
  on Logic in Computer Science}} (Saarbr\"{u}cken, Germany)
  \emph{(\bibinfo{series}{LICS ’20})}. \bibinfo{publisher}{Association for
  Computing Machinery}, \bibinfo{address}{New York, NY, USA},
  \bibinfo{pages}{713–727}.
\newblock
\showISBNx{9781450371049}
\urldef\tempurl%
\url{https://doi.org/10.1145/3373718.3394783}
\showDOI{\tempurl}


\bibitem[\protect\citeauthoryear{Mathur and Viswanathan}{Mathur and
  Viswanathan}{2020}]%
        {MathurAtomicity20}
\bibfield{author}{\bibinfo{person}{Umang Mathur} {and} \bibinfo{person}{Mahesh
  Viswanathan}.} \bibinfo{year}{2020}\natexlab{}.
\newblock \showarticletitle{Atomicity Checking in Linear Time Using Vector
  Clocks}. In \bibinfo{booktitle}{\emph{Proceedings of the Twenty-Fifth
  International Conference on Architectural Support for Programming Languages
  and Operating Systems}} (Lausanne, Switzerland)
  \emph{(\bibinfo{series}{ASPLOS ’20})}. \bibinfo{publisher}{Association for
  Computing Machinery}, \bibinfo{address}{New York, NY, USA},
  \bibinfo{pages}{183–199}.
\newblock
\showISBNx{9781450371025}
\urldef\tempurl%
\url{https://doi.org/10.1145/3373376.3378475}
\showDOI{\tempurl}


\bibitem[\protect\citeauthoryear{Mattern}{Mattern}{1988}]%
        {Mattern1988}
\bibfield{author}{\bibinfo{person}{Friedemann Mattern}.}
  \bibinfo{year}{1988}\natexlab{}.
\newblock \showarticletitle{{Virtual Time and Global States of Distributed
  Systems}}. In \bibinfo{booktitle}{\emph{Parallel and Distributed
  Algorithms}}. \bibinfo{publisher}{North-Holland}, \bibinfo{pages}{215--226}.
\newblock


\bibitem[\protect\citeauthoryear{Narayanasamy, Wang, Tigani, Edwards, and
  Calder}{Narayanasamy et~al\mbox{.}}{2007}]%
        {Narayanasamy2007}
\bibfield{author}{\bibinfo{person}{Satish Narayanasamy},
  \bibinfo{person}{Zhenghao Wang}, \bibinfo{person}{Jordan Tigani},
  \bibinfo{person}{Andrew Edwards}, {and} \bibinfo{person}{Brad Calder}.}
  \bibinfo{year}{2007}\natexlab{}.
\newblock \showarticletitle{Automatically Classifying Benign and Harmful Data
  Races Using Replay Analysis}. In \bibinfo{booktitle}{\emph{Proceedings of the
  28th ACM SIGPLAN Conference on Programming Language Design and
  Implementation}} (San Diego, California, USA) \emph{(\bibinfo{series}{PLDI
  ’07})}. \bibinfo{publisher}{Association for Computing Machinery},
  \bibinfo{address}{New York, NY, USA}, \bibinfo{pages}{22–31}.
\newblock
\showISBNx{9781595936332}
\urldef\tempurl%
\url{https://doi.org/10.1145/1250734.1250738}
\showDOI{\tempurl}


\bibitem[\protect\citeauthoryear{Pavlogiannis}{Pavlogiannis}{2019}]%
        {PavlogiannisPOPL20}
\bibfield{author}{\bibinfo{person}{Andreas Pavlogiannis}.}
  \bibinfo{year}{2019}\natexlab{}.
\newblock \showarticletitle{Fast, Sound, and Effectively Complete Dynamic Race
  Prediction}.
\newblock \bibinfo{journal}{\emph{Proc. ACM Program. Lang.}}
  \bibinfo{volume}{4}, \bibinfo{number}{POPL}, Article \bibinfo{articleno}{17}
  (\bibinfo{date}{Dec.} \bibinfo{year}{2019}), \bibinfo{numpages}{29}~pages.
\newblock
\urldef\tempurl%
\url{https://doi.org/10.1145/3371085}
\showDOI{\tempurl}


\bibitem[\protect\citeauthoryear{Pozniansky and Schuster}{Pozniansky and
  Schuster}{2003}]%
        {Pozniansky:2003:EOD:966049.781529}
\bibfield{author}{\bibinfo{person}{Eli Pozniansky} {and} \bibinfo{person}{Assaf
  Schuster}.} \bibinfo{year}{2003}\natexlab{}.
\newblock \showarticletitle{Efficient On-the-fly Data Race Detection in
  Multithreaded C++ Programs}. In \bibinfo{booktitle}{\emph{Proceedings of the
  Ninth ACM SIGPLAN Symposium on Principles and Practice of Parallel
  Programming}} (San Diego, California, USA) \emph{(\bibinfo{series}{PPoPP
  '03})}. \bibinfo{publisher}{ACM}, \bibinfo{address}{New York, NY, USA},
  \bibinfo{pages}{179--190}.
\newblock
\showISBNx{1-58113-588-2}
\urldef\tempurl%
\url{https://doi.org/10.1145/781498.781529}
\showDOI{\tempurl}


\bibitem[\protect\citeauthoryear{Roemer and Bond}{Roemer and Bond}{2019}]%
        {roemeronline2019}
\bibfield{author}{\bibinfo{person}{Jake Roemer} {and}
  \bibinfo{person}{Michael~D. Bond}.} \bibinfo{year}{2019}\natexlab{}.
\newblock \showarticletitle{Online Set-Based Dynamic Analysis for Sound
  Predictive Race Detection}.
\newblock \bibinfo{journal}{\emph{CoRR}}  \bibinfo{volume}{abs/1907.08337}
  (\bibinfo{year}{2019}).
\newblock
\showeprint[arxiv]{1907.08337}
\urldef\tempurl%
\url{http://arxiv.org/abs/1907.08337}
\showURL{%
\tempurl}


\bibitem[\protect\citeauthoryear{Roemer, Gen\c{c}, and Bond}{Roemer
  et~al\mbox{.}}{2018}]%
        {Roemer18}
\bibfield{author}{\bibinfo{person}{Jake Roemer}, \bibinfo{person}{Kaan
  Gen\c{c}}, {and} \bibinfo{person}{Michael~D. Bond}.}
  \bibinfo{year}{2018}\natexlab{}.
\newblock \showarticletitle{High-coverage, Unbounded Sound Predictive Race
  Detection}. In \bibinfo{booktitle}{\emph{Proceedings of the 39th ACM SIGPLAN
  Conference on Programming Language Design and Implementation}} (Philadelphia,
  PA, USA) \emph{(\bibinfo{series}{PLDI 2018})}. \bibinfo{publisher}{ACM},
  \bibinfo{address}{New York, NY, USA}, \bibinfo{pages}{374--389}.
\newblock
\showISBNx{978-1-4503-5698-5}
\urldef\tempurl%
\url{https://doi.org/10.1145/3192366.3192385}
\showDOI{\tempurl}


\bibitem[\protect\citeauthoryear{Roemer, Gen\c{c}, and Bond}{Roemer
  et~al\mbox{.}}{2020}]%
        {Roemer20}
\bibfield{author}{\bibinfo{person}{Jake Roemer}, \bibinfo{person}{Kaan
  Gen\c{c}}, {and} \bibinfo{person}{Michael~D. Bond}.}
  \bibinfo{year}{2020}\natexlab{}.
\newblock \showarticletitle{SmartTrack: Efficient Predictive Race Detection}.
  In \bibinfo{booktitle}{\emph{Proceedings of the 41st ACM SIGPLAN Conference
  on Programming Language Design and Implementation}} (London, UK)
  \emph{(\bibinfo{series}{PLDI 2020})}. \bibinfo{publisher}{Association for
  Computing Machinery}, \bibinfo{address}{New York, NY, USA},
  \bibinfo{pages}{747–762}.
\newblock
\showISBNx{9781450376136}
\urldef\tempurl%
\url{https://doi.org/10.1145/3385412.3385993}
\showDOI{\tempurl}


\bibitem[\protect\citeauthoryear{Rosu}{Rosu}{2018}]%
        {rvpredict}
\bibfield{author}{\bibinfo{person}{Grigore Rosu}.}
  \bibinfo{year}{2018}\natexlab{}.
\newblock \bibinfo{booktitle}{\emph{{{RV-Predict, Runtime Verification}}}}.
\newblock
\newblock
\shownote{Accessed: 2018-04-01.}


\bibitem[\protect\citeauthoryear{Said, Wang, Yang, and Sakallah}{Said
  et~al\mbox{.}}{2011}]%
        {Said2011}
\bibfield{author}{\bibinfo{person}{Mahmoud Said}, \bibinfo{person}{Chao Wang},
  \bibinfo{person}{Zijiang Yang}, {and} \bibinfo{person}{Karem Sakallah}.}
  \bibinfo{year}{2011}\natexlab{}.
\newblock \showarticletitle{{Generating Data Race Witnesses by an SMT-based
  Analysis}}. In \bibinfo{booktitle}{\emph{Proceedings of the Third
  International Conference on NASA Formal Methods}} (Pasadena, CA)
  \emph{(\bibinfo{series}{NFM'11})}. \bibinfo{publisher}{Springer-Verlag},
  \bibinfo{address}{Berlin, Heidelberg}, \bibinfo{pages}{313--327}.
\newblock


\bibitem[\protect\citeauthoryear{Savage, Burrows, Nelson, Sobalvarro, and
  Anderson}{Savage et~al\mbox{.}}{1997}]%
        {savage1997eraser}
\bibfield{author}{\bibinfo{person}{Stefan Savage}, \bibinfo{person}{Michael
  Burrows}, \bibinfo{person}{Greg Nelson}, \bibinfo{person}{Patrick
  Sobalvarro}, {and} \bibinfo{person}{Thomas Anderson}.}
  \bibinfo{year}{1997}\natexlab{}.
\newblock \showarticletitle{{Eraser: A Dynamic Data Race Detector for
  Multi-threaded Programs}}.
\newblock \bibinfo{journal}{\emph{SIGOPS Oper. Syst. Rev.}}
  \bibinfo{volume}{31}, \bibinfo{number}{5} (\bibinfo{date}{Oct.}
  \bibinfo{year}{1997}), \bibinfo{pages}{27--37}.
\newblock


\bibitem[\protect\citeauthoryear{Schonberg}{Schonberg}{1989}]%
        {Schonberg89}
\bibfield{author}{\bibinfo{person}{D. Schonberg}.}
  \bibinfo{year}{1989}\natexlab{}.
\newblock \showarticletitle{On-the-fly Detection of Access Anomalies}. In
  \bibinfo{booktitle}{\emph{Proceedings of the ACM SIGPLAN 1989 Conference on
  Programming Language Design and Implementation}} (Portland, Oregon, USA)
  \emph{(\bibinfo{series}{PLDI '89})}. \bibinfo{publisher}{ACM},
  \bibinfo{address}{New York, NY, USA}, \bibinfo{pages}{285--297}.
\newblock
\showISBNx{0-89791-306-X}
\urldef\tempurl%
\url{https://doi.org/10.1145/73141.74844}
\showDOI{\tempurl}


\bibitem[\protect\citeauthoryear{Sen}{Sen}{2008}]%
        {Sen:2008:RDR:1375581.1375584}
\bibfield{author}{\bibinfo{person}{Koushik Sen}.}
  \bibinfo{year}{2008}\natexlab{}.
\newblock \showarticletitle{Race Directed Random Testing of Concurrent
  Programs}. In \bibinfo{booktitle}{\emph{Proceedings of the 29th ACM SIGPLAN
  Conference on Programming Language Design and Implementation}} (Tucson, AZ,
  USA) \emph{(\bibinfo{series}{PLDI '08})}. \bibinfo{publisher}{ACM},
  \bibinfo{address}{New York, NY, USA}, \bibinfo{pages}{11--21}.
\newblock
\showISBNx{978-1-59593-860-2}
\urldef\tempurl%
\url{https://doi.org/10.1145/1375581.1375584}
\showDOI{\tempurl}


\bibitem[\protect\citeauthoryear{Sen, Ro\c{s}u, and Agha}{Sen
  et~al\mbox{.}}{2005}]%
        {sen2005detecting}
\bibfield{author}{\bibinfo{person}{Koushik Sen}, \bibinfo{person}{Grigore
  Ro\c{s}u}, {and} \bibinfo{person}{Gul Agha}.}
  \bibinfo{year}{2005}\natexlab{}.
\newblock \showarticletitle{{Detecting Errors in Multithreaded Programs by
  Generalized Predictive Analysis of Executions}}. In
  \bibinfo{booktitle}{\emph{Proceedings of the 7th IFIP WG 6.1 International
  Conference on Formal Methods for Open Object-Based Distributed Systems}}
  (Athens, Greece) \emph{(\bibinfo{series}{FMOODS'05})}.
  \bibinfo{publisher}{Springer-Verlag}, \bibinfo{address}{Berlin, Heidelberg},
  \bibinfo{pages}{211--226}.
\newblock


\bibitem[\protect\citeauthoryear{{\c{S}}erb{\u{a}}nu{\c{t}}{\u{a}}, Chen, and
  Ro{\c{s}}u}{{\c{S}}erb{\u{a}}nu{\c{t}}{\u{a}} et~al\mbox{.}}{2012}]%
        {maxcausalmodels}
\bibfield{author}{\bibinfo{person}{Traian~Florin
  {\c{S}}erb{\u{a}}nu{\c{t}}{\u{a}}}, \bibinfo{person}{Feng Chen}, {and}
  \bibinfo{person}{Grigore Ro{\c{s}}u}.} \bibinfo{year}{2012}\natexlab{}.
\newblock \showarticletitle{{Maximal causal models for sequentially consistent
  systems}}. In \bibinfo{booktitle}{\emph{International Conference on Runtime
  Verification}}. Springer, \bibinfo{pages}{136--150}.
\newblock


\bibitem[\protect\citeauthoryear{Serebryany and Iskhodzhanov}{Serebryany and
  Iskhodzhanov}{2009}]%
        {threadsanitizer}
\bibfield{author}{\bibinfo{person}{Konstantin Serebryany} {and}
  \bibinfo{person}{Timur Iskhodzhanov}.} \bibinfo{year}{2009}\natexlab{}.
\newblock \showarticletitle{{ThreadSanitizer: Data Race Detection in
  Practice}}. In \bibinfo{booktitle}{\emph{Proceedings of the Workshop on
  Binary Instrumentation and Applications}} (New York, New York, USA)
  \emph{(\bibinfo{series}{WBIA '09})}. \bibinfo{publisher}{ACM},
  \bibinfo{address}{New York, NY, USA}, \bibinfo{pages}{62--71}.
\newblock


\bibitem[\protect\citeauthoryear{Sergey}{Sergey}{2019}]%
        {soundness-dynamic-analysis}
\bibfield{author}{\bibinfo{person}{Ilya Sergey}.}
  \bibinfo{year}{2019}\natexlab{}.
\newblock \bibinfo{booktitle}{\emph{{{What Does It Mean for a Program Analysis
  to Be Sound?}}}}
\newblock
\newblock
\shownote{Accessed: 2019-08-07.}


\bibitem[\protect\citeauthoryear{Smaragdakis, Evans, Sadowski, Yi, and
  Flanagan}{Smaragdakis et~al\mbox{.}}{2012}]%
        {cp2012}
\bibfield{author}{\bibinfo{person}{Yannis Smaragdakis}, \bibinfo{person}{Jacob
  Evans}, \bibinfo{person}{Caitlin Sadowski}, \bibinfo{person}{Jaeheon Yi},
  {and} \bibinfo{person}{Cormac Flanagan}.} \bibinfo{year}{2012}\natexlab{}.
\newblock \showarticletitle{Sound Predictive Race Detection in Polynomial
  Time}. In \bibinfo{booktitle}{\emph{Proceedings of the 39th Annual ACM
  SIGPLAN-SIGACT Symposium on Principles of Programming Languages}}
  (Philadelphia, PA, USA) \emph{(\bibinfo{series}{POPL '12})}.
  \bibinfo{publisher}{ACM}, \bibinfo{address}{New York, NY, USA},
  \bibinfo{pages}{387--400}.
\newblock
\showISBNx{978-1-4503-1083-3}
\urldef\tempurl%
\url{https://doi.org/10.1145/2103656.2103702}
\showDOI{\tempurl}


\bibitem[\protect\citeauthoryear{Smith and Bull}{Smith and Bull}{2001}]%
        {JGF2001}
\bibfield{author}{\bibinfo{person}{Lorna~A Smith} {and} \bibinfo{person}{J~Mark
  Bull}.} \bibinfo{year}{2001}\natexlab{}.
\newblock \showarticletitle{{A multithreaded java grande benchmark suite}}. In
  \bibinfo{booktitle}{\emph{Proceedings of the third workshop on Java for high
  performance computing}}.
\newblock


\bibitem[\protect\citeauthoryear{Sorrentino, Farzan, and Madhusudan}{Sorrentino
  et~al\mbox{.}}{2010}]%
        {penelope2010}
\bibfield{author}{\bibinfo{person}{Francesco Sorrentino},
  \bibinfo{person}{Azadeh Farzan}, {and} \bibinfo{person}{P. Madhusudan}.}
  \bibinfo{year}{2010}\natexlab{}.
\newblock \showarticletitle{PENELOPE: Weaving Threads to Expose Atomicity
  Violations}. In \bibinfo{booktitle}{\emph{Proceedings of the Eighteenth ACM
  SIGSOFT International Symposium on Foundations of Software Engineering}}
  (Santa Fe, New Mexico, USA) \emph{(\bibinfo{series}{FSE '10})}.
  \bibinfo{publisher}{ACM}, \bibinfo{address}{New York, NY, USA},
  \bibinfo{pages}{37--46}.
\newblock
\showISBNx{978-1-60558-791-2}
\urldef\tempurl%
\url{https://doi.org/10.1145/1882291.1882300}
\showDOI{\tempurl}


\bibitem[\protect\citeauthoryear{von Praun and Gross}{von Praun and
  Gross}{2001}]%
        {vonPraun:2001:ORD:504282.504288}
\bibfield{author}{\bibinfo{person}{Christoph von Praun} {and}
  \bibinfo{person}{Thomas~R. Gross}.} \bibinfo{year}{2001}\natexlab{}.
\newblock \showarticletitle{Object Race Detection}. In
  \bibinfo{booktitle}{\emph{Proceedings of the 16th ACM SIGPLAN Conference on
  Object-oriented Programming, Systems, Languages, and Applications}} (Tampa
  Bay, FL, USA) \emph{(\bibinfo{series}{OOPSLA '01})}.
  \bibinfo{publisher}{ACM}, \bibinfo{address}{New York, NY, USA},
  \bibinfo{pages}{70--82}.
\newblock
\showISBNx{1-58113-335-9}
\urldef\tempurl%
\url{https://doi.org/10.1145/504282.504288}
\showDOI{\tempurl}


\bibitem[\protect\citeauthoryear{Wang, Kundu, Ganai, and Gupta}{Wang
  et~al\mbox{.}}{2009}]%
        {SPA2009}
\bibfield{author}{\bibinfo{person}{Chao Wang}, \bibinfo{person}{Sudipta Kundu},
  \bibinfo{person}{Malay Ganai}, {and} \bibinfo{person}{Aarti Gupta}.}
  \bibinfo{year}{2009}\natexlab{}.
\newblock \showarticletitle{Symbolic Predictive Analysis for Concurrent
  Programs}. In \bibinfo{booktitle}{\emph{Proceedings of the 2nd World Congress
  on Formal Methods}} (Eindhoven, The Netherlands) \emph{(\bibinfo{series}{FM
  '09})}. \bibinfo{publisher}{Springer-Verlag}, \bibinfo{address}{Berlin,
  Heidelberg}, \bibinfo{pages}{256--272}.
\newblock


\bibitem[\protect\citeauthoryear{Yao}{Yao}{1979}]%
        {Yao79}
\bibfield{author}{\bibinfo{person}{Andrew Chi-Chih Yao}.}
  \bibinfo{year}{1979}\natexlab{}.
\newblock \showarticletitle{Some Complexity Questions Related to Distributive
  Computing(Preliminary Report)}. In \bibinfo{booktitle}{\emph{Proceedings of
  the Eleventh Annual ACM Symposium on Theory of Computing}} (Atlanta, Georgia,
  USA) \emph{(\bibinfo{series}{STOC ’79})}. \bibinfo{publisher}{Association
  for Computing Machinery}, \bibinfo{address}{New York, NY, USA},
  \bibinfo{pages}{209–213}.
\newblock
\showISBNx{9781450374385}
\urldef\tempurl%
\url{https://doi.org/10.1145/800135.804414}
\showDOI{\tempurl}


\bibitem[\protect\citeauthoryear{Yu, Lee, and Bae}{Yu et~al\mbox{.}}{2018}]%
        {Yu18}
\bibfield{author}{\bibinfo{person}{Misun Yu}, \bibinfo{person}{Joon-Sang Lee},
  {and} \bibinfo{person}{Doo-Hwan Bae}.} \bibinfo{year}{2018}\natexlab{}.
\newblock \showarticletitle{AdaptiveLock: Efficient Hybrid Data Race Detection
  Based on Real-World Locking Patterns}.
\newblock \bibinfo{journal}{\emph{International Journal of Parallel
  Programming}} (\bibinfo{date}{04 Jun} \bibinfo{year}{2018}).
\newblock
\showISSN{1573-7640}
\urldef\tempurl%
\url{https://doi.org/10.1007/s10766-018-0579-5}
\showDOI{\tempurl}


\bibitem[\protect\citeauthoryear{Yu, Rodeheffer, and Chen}{Yu
  et~al\mbox{.}}{2005}]%
        {racetrack}
\bibfield{author}{\bibinfo{person}{Yuan Yu}, \bibinfo{person}{Tom Rodeheffer},
  {and} \bibinfo{person}{Wei Chen}.} \bibinfo{year}{2005}\natexlab{}.
\newblock \showarticletitle{{RaceTrack: Efficient Detection of Data Race
  Conditions via Adaptive Tracking}}.
\newblock \bibinfo{journal}{\emph{SIGOPS Oper. Syst. Rev.}}
  \bibinfo{volume}{39}, \bibinfo{number}{5} (\bibinfo{date}{Oct.}
  \bibinfo{year}{2005}), \bibinfo{pages}{221--234}.
\newblock


\bibitem[\protect\citeauthoryear{Zhao, Qiu, and Jin}{Zhao
  et~al\mbox{.}}{2019}]%
        {Zhao19}
\bibfield{author}{\bibinfo{person}{Qi Zhao}, \bibinfo{person}{Zhengyi Qiu},
  {and} \bibinfo{person}{Guoliang Jin}.} \bibinfo{year}{2019}\natexlab{}.
\newblock \showarticletitle{Semantics-Aware Scheduling Policies for
  Synchronization Determinism}. In \bibinfo{booktitle}{\emph{Proceedings of the
  24th Symposium on Principles and Practice of Parallel Programming}}
  (Washington, District of Columbia) \emph{(\bibinfo{series}{PPoPP ’19})}.
  \bibinfo{publisher}{Association for Computing Machinery},
  \bibinfo{address}{New York, NY, USA}, \bibinfo{pages}{242–256}.
\newblock
\showISBNx{9781450362252}
\urldef\tempurl%
\url{https://doi.org/10.1145/3293883.3295731}
\showDOI{\tempurl}


\bibitem[\protect\citeauthoryear{Zhivich and Cunningham}{Zhivich and
  Cunningham}{2009}]%
        {SoftwareErrors2009}
\bibfield{author}{\bibinfo{person}{M. Zhivich} {and} \bibinfo{person}{R.~K.
  Cunningham}.} \bibinfo{year}{2009}\natexlab{}.
\newblock \showarticletitle{The Real Cost of Software Errors}.
\newblock \bibinfo{journal}{\emph{IEEE Security and Privacy}}
  \bibinfo{volume}{7}, \bibinfo{number}{2} (\bibinfo{date}{March}
  \bibinfo{year}{2009}), \bibinfo{pages}{87–90}.
\newblock
\showISSN{1540-7993}
\urldef\tempurl%
\url{https://doi.org/10.1109/MSP.2009.56}
\showDOI{\tempurl}


\end{thebibliography}

\clearpage
%% Appendix
\appendix
% \section{Appendix}

%!TEX root = main.tex

\section{Proofs from~\secref{detection}}
\applabel{app_detection}

\lemtraceordersuffices*

\begin{proof}
Let $\pi$ be a sync-preserving correct reordering of $\tr$.
Let $E = \events{\pi}$.
Observe that $E$ is $(\tho{\tr}, \lw{\tr})$-closed because $\pi$
is a correct reordering of $\tr$.
We note a few observations about $E$.
First, $E \subseteq \events{\tr}$.
Second, $E$ is downward closed with respect to $\tho{\tr}$.
Third, for every read event $e \in E$, we have $\lw{\tr}(e) \in E$.
Fourth, for every lock $\lk$, there is at most one acquire $a$
of $\lk$ such that $a \in E$ but $\match{\tr}(a) \not\in E$.

Now, consider the sequence $\rho$
obtained by linearizing $E$ according to the total order of 
$\trord{\tr}$ (i.e., $\rho$ is the projection of $\tr$ onto the set $E$).
We first argue that $\rho$ is a well-formed trace.
This follows because for every lock $\lk$, 
there is at most one unmatched acquire event $e$ in $E$ of lock $\lk$, 
and all the other acquires are earlier than $e$ in $\tr$ (and hence in $\rho$).
Next, $\rho$ respects $\tho{\tr}$ because $E$ is downward closed with respect to $\tho{\tr}$
and respects $\trord{\tr}$.
For the same reason, for every read event $e \in E$,
we have $\lw{\rho}(e) = \lw{\tr}(e)$.
Thus, $\rho$ is a correct reordering of $\tr$.
Further, $\rho$ is trivially a sync-preserving correct reordering of $\tr$.
\end{proof}

% \begin{proposition}
% \proplabel{only-before}
% For an 
% \end{proposition}
% \begin{proof}
% Follows from~\defref{zr-closure}
% \end{proof}

\zridealdisjointness*
\begin{proof}
Without loss of generality, we let $e_1 \trord{\tr} e_2$.

($\Rightarrow$) Assume $\ZRIdeal{\tr}(e_1, e_2)\cap\set{e_1, e_2} \neq \emptyset$.
Then we must have $e_1 \in \ZRIdeal{\tr}(e_1, e_2)$
and in particular $e_1 \in \ZRClosure{\tr}(\prev{\tr}(e_2))$
(and thus $e = \prev{\tr}(e_2) \neq \bot$).
Then, either $e_1 \in \TOOClosure{\tr}(\set{e})$ or
there is a release event $e'$ such that
$e_1 \tho{\tr} e'$ and $e' \in \ZRClosure{\tr}(e)$.
In either case, $e_1$ cannot be enabled in
a sync-preserving correct reordering containing $e$.

($\Leftarrow$)
Let $E = \ZRIdeal{\tr}(e_1, e_2)$ and let $\rho$
be the sequence obtained by linearizing $E$ as per $\trord{\tr}$.
Observe that $\rho$ is well formed, respects $\tho{\tr}$
and $\lw{\tr}$ and thus is a correct reordering of $\tr$.
Further, the order of all critical sections is the same.
Also, $e_1$ and $e_2$ are both enabled in $\rho$.
\end{proof}

\lemmonotonicity*

\begin{proof}
The proof follows from the following observations.
$e_1 \tho{\tr} \prev{\tr}(e'_1)$,
$e_2 \tho{\tr} \prev{\tr}(e'_2)$
and $\ZRClosure{\tr}(S)$ is downward closed with respect to $(\tho{\tr})$
and sync-preserving-closed.
\end{proof}

\lemlangnlowerbound*
\begin{proof}
Assume towards contradiction otherwise, i.e., there is a streaming algorithm that uses $o(n)$ space.
Hence the state space of the algorithm is $o(2^n)$.
Then, there exist two distinct $n$-bit strings $u_1\neq u_2$, such that the algorithm is in the same state after parsing $u_1$ and $u_2$.
Hence, for any $n$-bit string $v$, the algorithm gives the same answer on inputs $u_1\#v$ and $u_2\#v$.
Since the algorithm is correct, it reports that $u_1\#u_1$ belongs to $\Lang_n$.
But then the algorithm reports that $u_2\#u_1$ also belongs to $\Lang_n$, a contradiction.
The desired result follows.
\end{proof}

\lemlangncorrectness*
\begin{proof}
We prove each item separately.

\noindent{\em 1. $s \in \EQ_n$.}
First, notice that by construction, for every $i\in [n]$, the pair $(e^1_i,e^2_i)$ is not a predictable race of $\tr$, as we have $u[i]=v[i]$ and thus either both events are read events or both are write events and thus each is protected by the lock $c$.
It remains to argue that $e(^1_i, e^2_j)$ is not a predictable race for any $i,j\in [n]$ with $i\neq j$.
It suffices to show that $\lheld{\tr}(e^1_i)) \cap \lheld{\tr}(e^2_j) \neq \emptyset$.
We split cases based on the relation between $i$ and $j$.
\begin{compactenum}
\item[$i<j$.] By construction, we have $\lheld{\tr}(e^1_j)\cap B \not \subseteq \lheld{\tr}(e^1_i) \cap B$.
Moreover, we have $\lhead{\tr}(e^2_j)\cap B = B \setminus \lheld{\tr}(e^1_j)$.
Thus, we have $\lheld{\tr}(e^1_i)) \cap \lheld{\tr}(e^2_j) \cap B \neq \emptyset$, as desired.
\item[$j<i$.] By construction, we have $\lheld{\tr}(e^1_i)\cap A \not \subseteq \lheld{\tr}(e^1_j) \cap A$.
Moreover, we have $\lhead{\tr}(e^2_i)\cap A = A \setminus \lheld{\tr}(e^1_i)$.
Thus, we have $\lheld{\tr}(e^1_i)) \cap \lheld{\tr}(e^2_j) \cap A \neq \emptyset$, as desired.
\end{compactenum}

\noindent{\em 2. $s \not \in \EQ_n$.}
First, notice that by construction, one of $e_i^1$ and $e_i^2$ is a read event and the other is a write event.
Hence, at least one of them is not surrounded by lock $c$.
Finally, by construction we have $\lhead{\tr}(e^2_i)\cap (A\cup B) = A\cup B \setminus\lheld{\tr}(e^1_j)$, and thus $\lheld{\tr}(e^1_i)) \cap \lheld{\tr}(e^2_j) = \emptyset$.

The desired result follows.
\end{proof}

\Paragraph{The set-equality problem.}
Finally, we turn our attention to \cref{thm:product_lowerbound}.
Our proof uses the set-equality problem, i.e., given two bit-sets $u,v\in \{0,1\}^n$, the task is to decide whether $u=v$.
The problem has a $\Omega(n)$ lower-bound for communication complexity~\cite{Yao79},
i.e., if $u$ and $v$ is given separately to Alice and Bob, respectively, the two parties need to exchange $\Omega(n)$ bits of information in order to decide whether $u=v$.

\begin{proof}[Proof of \cref{thm:product_lowerbound}.]
Consider the language $\EQ_n=\{u\#^n v\colon  u,v\in \{0,1\}^n \text{ and } u=v\}$, and any Turing Machine $M$ that decides $\EQ_n$ using $T_M(n)$ time and $S_M(n)$ space.
Let $K_M(n)$ be an upper-bound on the number of ``passes'' that $M$ makes over the sequence $\#^n$ as it decides membership in $\EQ_n$.
In each pass, $M$ ``communicates'' at most $S_M(n)$ bits of information.
Since the set-equality problem has communication complexity $\Omega(n)$~\cite{Yao79},
we have $K_M(n)\cdot S_M(n)=\Omega(n)$, i.e., $M$ makes $\Omega(n/S_M(n))$ passes.
Since each pass has to traverse $n$ symbols $\#$, each pass costs time $\Omega(n)$.
Hence, the total time is $T_M(n)=\Omega(n\cdot K_M(n))$, and thus $T_M(n)\cdot S_M(n)\geq n^2$.

We now turn our attention to sync-preserving race prediction.
Let $m=n/\log n$
Using essentially the same reduction as above, we reduce the membership problem for $\EQ_{m}$ to the sync-preserving race prediction problem on a trace $\tr$ with $2$ threads, $n$ events and $O(\log n)$ locks.
\cref{lem:langn_correctness} guarantees that $\tr$ has a predictable race iff $u\neq v$, and if so, then it is a sync-preserving race.
Hence, any algorithm that solves sync-preserving race prediction on $\tr$ in time $T(n)$ and space $S(n)$ must satisfy that $T(n)\cdot S(n)=\Omega (m^2)=\Omega(n^2/\log^2 n)$.
The desired result follows.
\end{proof}

\section{Proofs from~\secref{dichotomy}}\label{sec:app_dichotomy}

In this section we present the proofs of \cref{sec:dichotomy}, i.e., 
\cref{lem:p_po}, \cref{lem:witness} and \cref{lem:rscrm_hardness}.

\smallskip
\lemppo*
\begin{proof}
Assume towards contradiction otherwise, hence $P'$ has a cycle $e_1 <_{P'} e_2<_{P'}\dots<_{P'} e_1$.
Let $C=\{ e_i\colon e_i \in X'\setminus X \}$, and note that $C\neq \emptyset$ as $P\Refines P'\Project X$.
Observe that $C$ cannot contain any event of the distinguished triplet $\DistinguishedTriplet$, as every event of $\DistinguishedTriplet$ is either minimal or maximal in $P'$.
On the other hand, $C$ cannot contain any event of any set $Y_i^j$, as every such event only has predecessors that are either in $Y_i^j$ or in $\DistinguishedTriplet$, and clearly $P'\Project Y_i^j$ is acyclic.
Finally, $C$ cannot contain any event of any set $X_i^j$, as every such event only has successors that are either in $X_i^j$ or in $\DistinguishedTriplet$,
and clearly $P'\Project X_i^j$ is acyclic.
Hence $C=\emptyset$, a contradiction.

The desired result follows.
\end{proof}

\smallskip
\lemwitness*
\begin{proof}
We argue that $Q$ is indeed a partial order, from which follows that $\tr$ is a witness of the realizability of $\RFPoset'$.
Observe that for every interfering write event $\wt'$ of each triplet, every predecessor $\wt''$ of $\wt'$ is also an interfering write event of some triplet. 
On the other hand, every new successor of $\wt'$ in $Q$ is not an interfering write event of any triplet.
Hence, since $P'$ is acyclic, $Q$ is also acyclic.

The desired result follows.
\end{proof}

\smallskip
\lemrscrmhardness
\begin{proof}
Here we argue that $\RFPoset$ has a witness $\tr_1$ iff $\RFPoset'$ has a witness $\tr_2$.

\noindent{($\Rightarrow$).}
Consider the witness $\tr_1$ for $\RFPoset$, and we show how to obtain the witness $\tr_2$ for $\RFPoset'$.
We construct a partial order $Q$ over $X'$ such that
(i)~$Q\Refines P'$,
(ii)~for every triplet $(\wt, \rd, \wt')\in \ReadTriplets{\RFPosetS}\setminus\ReadTriplets{\RFPoset}$, we have $\rd <_{Q} \wt'$, and 
(iii)~$Q\Project X = \tr_1$ (i.e., $Q$ totally orders the events of $X$ according to $\tr_1$).
Afterwards, we construct $\tr_2$ by linearizing $Q$ arbitrarily.
Note that (ii) makes $Q\Project X=P$, hence the linearization in (iii) is well defined.

\noindent{($\Leftarrow$).}
Consider the witness $\tr_2$ for $\RFPoset'$, and we show how to obtain the witness $\tr_1$ for $\RFPoset$.
We construct $\tr_1$ simply as $\tr_1=\tr_2 \Project X$.
To see that $\tr_1$ is a linearization of $(X,P)$, consider any two events $e, e'\in X$ such that $e<_{P}e'$.
Consider the triplet $(\wt_{e, e'}, \rd_{e, e'}, \wt_{e, e'}')$ of $\RFPoset'$,
and we have $\wt_{e, e'}<_{\tr_2}\wt'_{e, e'}$, thus $\rd_{e, e'}<_{\tr_2}\wt'_{e, e'}$.
This leads to $e<_{\tr_2}e'$, and thus $e<_{\tr_1}e'$, as desired.

The desired result follows.
\end{proof}

\thmwonehard*
\begin{proof}
We argue that $\RFPoset$ has a witness $\tr_1$ iff $(\wt(y), \rd(y))$ is a predictable data race of $\tr'$, witnessed by a trace $\tr_2$.

\noindent{($\Rightarrow$).}
Consider the witness $\tr_1$, hence $\ov{\rd}<_{\tr_1}\ov{\wt}'$ for the distinguished triplet $\DistinguishedTriplet=(\ov{\wt}, \ov{\rd}, \ov{\wt}')$.
The witness trace $\tr_2$ is constructed as follows.
\begin{compactenum}
\item We insert in $\tr_1$ all events that where inserted in $\tr$ to produce $\tr'$, in the same order.
\item We remove from $\tr_1$ the events $\{ \rel_1, \wt(y), \rd(y) \}$.
\end{compactenum}
It follows easily that $\tr_2$ is a correct reordering of $\tr'$, in which $\wt(y)$ and $\rd(y)$ are enabled, hence $(\wt(y), \rd(y))$ is a predictable data race of $\tr'$.

\noindent{($\Leftarrow$).}
Consider the witness $\tr_2$, and it is straightforward that $\events{\tr_2} = \events{\tr'}\setminus \{ \rel_2 , \ov{\wt}, \ov{\rd} \}$.
Hence $\rel_1<_{\tr_2}\acq_2$ and thus $\ov{\rd}<\ov{\wt}'$ for the distinguished triplet $\DistinguishedTriplet=(\ov{\wt}, \ov{\rd}, \ov{\wt}')$.
The witness $\tr_1$ is constructed as $\tr_1=\tr_2\Project X$, i.e., by removing from $\tr_2$ all the events that we inserted when we constructed $\tr'$ from $\tr$.
It follows easily that $\tr_1$ is a linearization of the rf-poset $(X,P,\RF)$ and $\rd<_{\tr_1}\wt'$.

The desired result follows.
\end{proof}
%!TEX root = main.tex

%% REMARK
%% Order of reads and writes need not be same
%% RHO may not event contain some critical sections
%% All SHB races are zero reversal races
%% Some WCP races may not be zero reversal races and vice versa.
%% Soundness guarantee is strong soundness, unlike WCP or CP.

\section{Comparison with Other Race Prediction Algorithms}
\applabel{comparison}

Here, we discuss recent advances in data race prediction
and characterize their prediction power with respect to
sync-preserving races.
We focus our attention to \emph{sound} race prediction techniques.
Specifically, we compare sync-preserving race detection with 
race prediction based on 
the \emph{happens-before} (\textsf{HB}) partial order, 
the \emph{schedulable-happens-before} (\textsf{SHB}) partial order~\cite{shb2018},
the \emph{causally precedes partial order} (\textsf{CP})~\cite{cp2012}, 
the \emph{weak causally precedes partial order} (\textsf{WCP})~\cite{wcp2017}, and
the \emph{does not commute partial order} (\textsf{DC})~\cite{Roemer18}.
%\text{M2}~\cite{PavlogiannisPOPL20}.
% The \emph{doesn't commute} (\textsf{DC}) order~\cite{Roemer18}
% is not a sound partial order (reports false alarms, even when there is no predictable race)
% and we do not compare against it.
The recently introduced partial order \emph{strong-dependently-precedes} 
(\textsf{SDP})~\cite{Bond2019}, while claimed to be sound in that paper,
is actually unsound. 
In \cref{sec:app_sdp}, we show a counter-example
to the soundness theorem of \textsf{SDP}.

\subsection{Comparison with \textsf{HB} and \textsf{SHB}}

The happens-before ($\hb{}$) order is a classic partial order
employed in popular race detectors like ThreadSanitizer~\cite{threadsanitizer}
and FastTrack~\cite{fasttrack}.
The main idea behind race detectors based on $\hb{}$ 
is to determine the existence of
conflicting pairs of events that are unordered by $\hb{}$,
defined as follows.

\begin{definition}[Happens-Before]
The happens-before order defined given a trace $\tr$ is the smallest partial order
$\hb{\tr}$ on $\events{\tr}$ such that $\tho{\tr} \subseteq \hb{\tr}$
and for every lock $\lk \in \locks{\tr}$ and for every two events
$e_1 \trord{\tr} e_2$ such that $e_1 \in \releases{\tr}(\lk)$ and
$e_2 \in \acquires{\tr}(\lk)$, we have $e_1 \hb{\tr} e_2$.
\end{definition}

A pair $(e_1, e_2)$ of conflicting events of $\tr$
is said to be an $\hb{}$-race if $\Unordered{e_1}{\tr}{\textsf{HB}}{e_2}$.
The soundness guarantee of $\hb{}$ states that if $\tr$
has an $\hb{}$-race, then $\tr$ also has a predictable data race.
The partial order $\hb{}$ is known to miss the existence of predictable data races,
or, in other words, it is incomplete.
Further, as noted in as noted in~\cite{shb2018}, 
while $\hb{}$ is sound for checking the existence of a data race,
the soundness guarantee only applies to the first such race identified,
beyond which, conflicting pairs of events unordered by $\hb{}$
may not be predictable races.
The partial order $\shb{}$~\cite{shb2018}, defined below, 
overcomes this problem.

\begin{definition}[Schedulable-Happens-Before~\cite{shb2018}]
The schedulable-happens-before order defined given a trace $\tr$ is the smallest partial order
$\shb{\tr}$ on $\events{\tr}$ such that $\hb{\tr} \subseteq \shb{\tr}$
and for every variable $x \in \locks{\tr}$ and for every read event
$e_1 \trord{\tr} e_2$ such that $e_1 \in \releases{\tr}(\lk)$ and
$e_2 \in \acquires{\tr}(\lk)$, we have $e_1 \hb{\tr} e_2$.
\end{definition}

A pair of conflicting events $(e_1, e_2)$ in $\tr$ 
is an $\shb{}$-race if
either $\prev{\tr}(e_2) = \bot$ or $\Unordered{e_1}{\tr}{\textsf{SHB}}{\prev{\tr}(e_2)}$.
The soundness theorem for \textsf{SHB} states
that every $\shb{}$-race of a trace $\tr$ is also a 
predictable data race of $\tr$~\cite{shb2018}, and further
the first race identified by $\hb{}$ (for which soundness of $\hb{}$ holds)
is also reported by $\shb{}$.

We next make the following observation.
The proof follows directly from the soundness proof of \textsf{SHB}~\cite{shb2018}.
\begin{restatable}{lemma}{hbraceimplieszeroreversal}\lemlabel{hbraceimplieszeroreversal}
For a trace $\tr$ and a conflicting pair of events $(e_1, e_2)$ of $\tr$,
if $(e_1, e_2)$ is an \textsf{SHB}-race, then $(e_1, e_2)$ is 
a sync-preserving race of $\tr$.
\end{restatable}

\begin{proof}
Let $\tr$ be a trace and let $(e_1, e_2)$ be an \textsf{SHB}-race of $\tr$
such that $e_1 \trord{\tr} e_2$.
Let $S_i = \setpred{e \in \events{\tr}}{e \strictshb{\tr} e_i}$
and let $S = S_1 \cup S_2$.
We observe that $\set{e_1, e_2} \cap S = \emptyset$.
Next, consider the sequence $\rho'$ obtained by linearizing the events in $S$
as per $\trord{\tr}$ and let $\rho = \rho' \cdot e_1 \cdot e_2$.
We remark that $S$ is downward closed with respect to $\shb{\tr}$
and thus $\rho$  is a correct reordering of $\tr$.
Further, since $\trord{\rho} \subseteq \trord{\tr}$, $\rho$
is also a sync-preserving correct reordering of $\tr$.
\end{proof}

\begin{example}
\exlabel{hbracesmissedbyzeroreversal}
Consider the trace $\tr_\two$ in~\figref{hb-misses-race}.
%\ucomment{$e_i$ denotes the $i^\text{th}$ event, etc.,}
Both $\hb{\tr_\two}$ and $\shb{\tr_\two}$ order all events in $\tr_\two$, and thus there is
no \textsf{HB} or \textsf{SHB} race in $\tr$.
Nevertheless, the correct reordering $\tr_\two^\cre$
is a sync-preserving correct reordering that witnesses the race $(e_1, e_6)$.
Notice that the earlier critical section in $\tr_\two$ on lock $\lk$
(performed in thread $t_1$) is not present in the correct reordering.
\end{example}

Based on \lemref{hbraceimplieszeroreversal} and \exref{hbracesmissedbyzeroreversal},
we have the following.
\begin{observation}
The prediction power of sync-preserving race prediction is
strictly better than \textsf{SHB} race prediction.
\end{observation}

\subsection{Comparison with \textsf{CP} and \textsf{WCP}}

The partial orders \emph{causally precedes} (denoted $\cp{}$)~\cite{cp2012} and
\emph{weak causally precedes} (denoted $\wcp{}$)~\cite{wcp2017} 
are two recent partial orders, proposed for data race prediction. 
Both partial orders are sound and can report races missed by $\hb{}$ (and $\shb{}$).
More precisely, $\cp{}$ can predict races on strictly more traces than $\hb{}$,
and $\wcp{}$ can, in turn, predict strictly more races than $\cp{}$.
Further, $\wcp{}$ has a more efficient race detection algorithm than $\cp{}$.
We define $\wcp{}$ below; the precise definition of $\cp{}$ is similar
to that of $\wcp{}$ and not important for our discussion here.
Here, we say that two acquire events $a_1$ and $a_2$
are conflicting if there is a common lock $\lk$ such that
$\OpOf{a_1} = \OpOf{a_2} = \acq(\lk)$.

\begin{definition}[Weak Causal Precedence]
\deflabel{wcp}
For a trace $\tr$, $\wcp{\tr} = \strictwcp{\tr} \cup \tho{\tr}$,
where $\strictwcp{\tr}$ is the smallest transitive order such that
the following hold:
\begin{compactenum}[(a)]
\item\itmlabel{wcp-rule-a} For two conflicting acquire events
$a_1, a_2$, if there are events $e_1 \in \crit{\tr}(a_1)$
and $e_2 \in \crit{\tr}(a_2)$ such that $e_1 \conf e_2$, then
$\match{\tr}(a_1) \strictwcp{\tr} e_2$.
\item\itmlabel{wcp-rule-b} For two conflicting acquire
$a_1, a_2$ events, if $a_1 \strictwcp{\tr} a_2$,
then $\match{\tr}(a_1) \strictwcp{\tr} \match{\tr}(a_2)$.
\item\itmlabel{wcp-rule-c} For any three events $e_1, e_2, e_3$, if either
$e_1 \strictwcp{\tr} e_2 \hb{\tr} e_3$, or
$e_1 \hb{\tr} e_2 \strictwcp{\tr} e_3$, then
$e_1 \strictwcp{\tr} e_3$.
\end{compactenum}
\end{definition}
The partial order $\cp{}$ is defined in a similar manner, 
except that it orders the second
conflicting acquire event $a_2$ in rules \itmref{wcp-rule-a}
and \itmref{wcp-rule-b} in~\defref{wcp}.
For a trace $\tr$, a $\wcp{}$-race (resp. $\cp{}$-race) is a pair of conflicting events 
$(e_1, e_2)$ in $\tr$ unordered by $\wcp{\tr}$ (resp. $\cp{\tr}$).

The soundness guarantee of $\wcp{}$ (and $\cp{}$) is
that of \emph{weak soundness}---if a trace $\tr$ has a $\wcp{}$-race (or $\cp{}$-race), then
$\tr$ has a predictable data race or a predictable deadlock\footnote{A trace $\tr$ has a predictable deadlock, if there is a correct reordering of $\tr$ that witnesses a deadlock.}.
We remark that, as with $\hb{}$, not every $\wcp{}$-race
is a predictable race, and the weak soundness guarantee applies only
to the first $\wcp{}$-race identified.

Let us now consider how $\wcp{}$-race prediction compares with
sync-preserving race prediction.

%!TEX root = ../main.tex

% \begin{figure}[t]
% \execution{2}{
% \figev{1}{\wt(x)}
% \figev{1}{\acq(\lk)}
% \figev{1}{\wt(z)}
% \figev{1}{\rel(\lk)}
% \figev{2}{\acq(\lk)}
% \figev{2}{\wt(z)}
% \figev{2}{\rel(\lk)}
% \figev{2}{\wt(x)}
% }
% \execution{2}{
% \figev{1}{\acq(\lk)}
% \figev{1}{\wt(x)}
% \figev{1}{\rel(\lk)}
% \figev{2}{\acq(\lk)}
% \figev{2}{\rel(\lk)}
% \figev{2}{\wt(x)}
% }
% \caption{
% Traces $\tr_2$ (left) and $\tr_3$ (right).
% \zeroreversal~race $(e_1, e_8)$ in $\tr_2$ missed by $\wcp{}$ and $\cp{}$.
% \textsf{WCP}-race $(e_2, e_6)$ in $\tr_3$ which is not a \zeroreversal-race.
% }
% \figlabel{compare-wcp}
% \end{figure}

\begin{figure}[t]
\centering
\begin{subfigure}{.5\textwidth}
	\centering
	\execution{2}{
		\figev{1}{\mathbf{\wt(x)}}
		\figev{1}{\acq(\lk)}
		\figev{1}{\wt(z)}
		\figev{1}{\rel(\lk)}
		\figev{2}{\acq(\lk)}
		\figev{2}{\wt(z)}
		\figev{2}{\rel(\lk)}
		\figev{2}{\mathbf{\wt(x)}}
		\orderedgewithlabel{1}{4}{0.5}{2}{6}{-0.3}{\small $\strictwcp{}$}{below}
	}
	\caption{$\tr_\three$ with no $\wcp{}$-race.}
	\figlabel{wcp-misses-race-rule-a}
\end{subfigure}%
\begin{subfigure}{.5\textwidth}
	\centering
	\execution{3}{
		\figev{1}{\mathbf{\wt(x)}}
		\figev{1}{\acq(\lk)}
		\figev{1}{\wt(z)}
		\figev{1}{\rel(\lk)}
		\figev{2}{\acq(\lk)}
		\figev{2}{\rd(z)}
		\figev{2}{\rel(\lk)}
		\figev{3}{\acq(\lk)}
		\figev{3}{\rel(\lk)}
		\figev{3}{\mathbf{\wt(x)}}
		\orderedgewithlabel{1}{4}{0.5}{2}{6}{-0.3}{\small $\strictwcp{}$}{below}
		\orderedgewithlabel{2}{7}{0.5}{3}{8}{-0.5}{\small $\hb{}$}{above}
	}
	\caption{Trace $\tr_\four$ with no $\wcp{}$-race.}
	\figlabel{wcp-misses-race-rule-c}
\end{subfigure}%
% \begin{subfigure}{.3\textwidth}
%   \centering
% 	\execution{2}{
% 		\figev{1}{\acq(\lk)}
% 		\figev{1}{\wt(x)}
% 		\figev{1}{\rel(\lk)}
% 		\figev{2}{\acq(\lk)}
% 		\figev{2}{\rel(\lk)}
% 		\figev{2}{\wt(x)}
% 		}
%   \caption{\Zeroreversal~correct-reordering $\tr_1^\cre$ of $\tr_1$ witnessing a data race.}
%   \figlabel{ZR-race-missed-by-HB}
% \end{subfigure}
\caption{Traces $\tr_\three$ and $\tr_\four$ have sync-preserving races.
But $\wcp{}$ does not report any races.}
\figlabel{compare-wcp-missed}
\end{figure}
\begin{example}
\exlabel{wcp-zr-rule-a}
Consider trace $\tr_\three$ in \figref{wcp-misses-race-rule-a}.
Here, we have $e_4 \strictwcp{\tr_\three} e_6$ due to rule~\itmref{wcp-rule-a}.
Together with composition with $\hb{\tr_\three}$ (rule~\itmref{wcp-rule-c}),
we have that $e_1 \wcp{\tr_\three} e_8$ and thus there is no $\wcp{}$-race in $\tr_\three$.
However, the trace $\tr_\three^\cre = e_5{\cdot}e_6{\cdot}e_7{\cdot}e_1{\cdot}e_8$
is a sync-preserving correct reordering of $\tr_\three$ that exposes the predictable race $(e_1, e_8)$.
\end{example}
We remark that $\wcp{}$ misses the race in \exref{wcp-zr-rule-a}
because of the ordering $e_4 \strictwcp{\tr_\three} e_6$, which is a 
spurious ordering and correct reorderings may not necessarily respect it.
We next highlight another source of imprecision in $\wcp{}$ arising
due to the $\hb{}$-composition rule of $\wcp{}$ (rule~\itmref{wcp-rule-c} in \defref{wcp}).

\begin{example}
Consider trace $\tr_\four$ in \figref{wcp-misses-race-rule-c}.
Here, due to rule~\itmref{wcp-rule-a}, we have $e_4 \strictwcp{\tr_\four} e_6$.
Further, we have $e_1 \hb{\tr_\four} e_4$ and $e_6 \hb{\tr_\four} e_{10}$, giving
us $e_1 \strictwcp{\tr_\four} e_{10}$ due to rule~\itmref{wcp-rule-c}
As a result, there is no $\wcp{}$-race in $\tr_\four$.
However, the pair $(e_1, e_{10})$ is, in fact, a sync-preserving race
witnessed by the correct reordering $\tr_\four^\cre = e_8{\cdot}e_9{\cdot}e_1{\cdot}e_{10}$
that completely drops the critical sections in $t_1$ and $t_3$.
\end{example}

Of course, there are predictable races that are neither $\wcp{}$-races,
nor sync-preserving races.
\begin{example}
The trace $\tr_{\six}$ in~\figref{compare-neither-wcp-nor-zr} has a predictable race
$(e_2, e_7)$ which is witnessed by the (only) correct reordering 
$\tr_{\six}^\cre = e_4{\cdot}e_5{\cdot}e_6{\cdot}e_1$.
Notice that this is not a sync-preserving correct ordering.
Further, $\wcp{}$ misses this race as well ---
$e_3 \strictwcp{\tr_{\six}} e_5$ (rule~\itmref{wcp-rule-a}),
giving $e_1 \wcp{\tr_{\six}} e_7$.
\end{example}

%!TEX root = ../main.tex

% \begin{figure}[h]
% \execution{2}{
% \figev{1}{\acq(\lk)}
% \figev{1}{\wt(x)}
% \figev{1}{\rel(\lk)}
% \figev{2}{\acq(\lk)}
% \figev{2}{\rel(\lk)}
% \figev{2}{\wt(x)}
% }
% \caption{
% Traces $\tr_\five$ with a $\wcp{}$-race but no \zeroreversal~race.
% }
% \figlabel{compare-wcp-more}
% \end{figure}

%!TEX root = ../main.tex

\begin{figure}[t]
\centering
\begin{subfigure}{.5\textwidth}
	\centering
	\vspace{0.1in}
	\execution{2}{
\figev{1}{\acq(\lk)}
\figev{1}{\mathbf{\wt(x)}}
\figev{1}{\rel(\lk)}
\figev{2}{\acq(\lk)}
\figev{2}{\rel(\lk)}
\figev{2}{\mathbf{\wt(x)}}
}
	\caption{Trace $\tr_\five$ with predictable race reported by $\wcp{}$}
	\figlabel{compare-wcp-more}
\end{subfigure}%
\begin{subfigure}{.5\textwidth}
  \centering
		\execution{2}{
\figev{1}{\acq(\lk)}
\figev{1}{\wt(x)}
\figev{1}{\rel(\lk)}
\figev{2}{\acq(\lk)}
\figev{2}{\wt(x)}
\figev{2}{\rel(\lk)}
\figev{2}{\wt(x)}
}
	% \vspace{0.2in}
  \caption{Trace $\tr_{\six}$ with predictable race missed by $\wcp{}$}
  \figlabel{compare-neither-wcp-nor-zr}
\end{subfigure}
\caption{Traces with no sync-preserving races. $\wcp{}$ predicts race in $\tr_\five$ but misses in $\tr_{\six}$}
\figlabel{compare-no-zr}
\end{figure}

In the next example, we illustrate that $\wcp{}$
can predict races that are not sync-preserving races.
\begin{example}
Consider trace $\tr_\five$ in~\figref{compare-wcp-more}.
Here, $(e_2, e_6)$ is a predictable data race witnessed by
the (only) correct reordering $\tr_\five^\cre = e_4{\cdot}e_5{\cdot}e_1{\cdot}e_2{\cdot}e_6$.
Observe that $\tr_\five^\cre$ is not a sync-preserving correct reordering
of $\tr_\five$ and thus $\tr_\five$ does not have any sync-preserving race.
At the same time, $\Unordered{e_2}{\tr}{\textsf{WCP}}{e_6}$
and thus this predictable race is identified by $\wcp{}$.
\end{example}
We summarize our comparison with $\wcp{}$ as follows.
\begin{observation}
The prediction power of $\wcp{}$-race prediction
and sync-preserving race prediction are incomparable.
Further, $\wcp{}$ offers a weak soundness guarantee, i.e.,
a $\wcp{}$-race may sometimes imply no predictable race but only a predictable deadlock,
whereas sync-preserving race prediction is strongly sound.
Finally, $\wcp{}$ is sound only until the first race, whereas
all sync-preserving races reported are true races.
\end{observation}

\subsection{Comparison with \textsf{DC} and \textsf{WDP}}

Finally, we briefly outline how sync-preserving races compare with the methods \textsf{DC} and \textsf{WDP}.
Both methods are unsound, and  they rely on a second vindication phase to filter out false positives.

The \textsf{DC} partial order was introduced in~\cite{Roemer18} as an unsound weakening to \textsf{WCP}.
The difference between the two is that \textsf{DC} does not compose with \textsf{HB} as in rule (\itmref{wcp-rule-c}) in the definition of \textsf{WCP},
and instead only composes with the thread order.
Due to its similarity with \textsf{WCP}, \textsf{DC} also misses sync-preserving races.
For example, in the trace of \figref{wcp-misses-race-rule-a}, \textsf{DC} forces the same ordering as \textsf{WCP}, and thus misses the sync-preserving race $(e_1, e_8)$.

The \textsf{WDP} partial order was introduced in~\cite{Bond2019} as a further unsound weakening of \textsf{DC}.
In high level, \textsf{WDP} operates on traces that also include branching events $\br$, and relies on static analysis to identify whether a read event $\rd$ affects $\br$. For every such $\rd$, \textsf{WDP} orders $\rd$ after the critical section that contains $\lw{\tr}(\rd)$.
%Note that in our setting, where such branch information is not present, \textsf{WDP} trivially coincides with the thread order.
While the \textsf{WDP} partial order captures all predictable data races, it also reports false alarms.
In order to eliminate false positives, \textsf{WDP}, like \text{DC} employs an additional post-processing step called \emph{vindication} phase.
In this phase, a graph-based check is employed for every pair of events reported to be a race, to check if there is a correct reordering in which the two events are simultaneously enabled.
This check however is not complete, and can rule out even true races.
As a result, even true sync-preserving races can be missed.
%!TEX root = main.tex

\section{A Note on the Soundness of SDP}\label{sec:app_sdp}

The \textsf{SDP} partial order was recently introduce in~\cite{Bond2019} for dynamic race prediction.
\cite[Theorem~p12]{Bond2019} states that \textsf{SDP} is sound, i.e., if a trace $\tr$ has an \textsf{SDP}-race and $\tr$ has a predictable race.
In this section we construct a counterexample to soundness.
\begin{figure}[h]
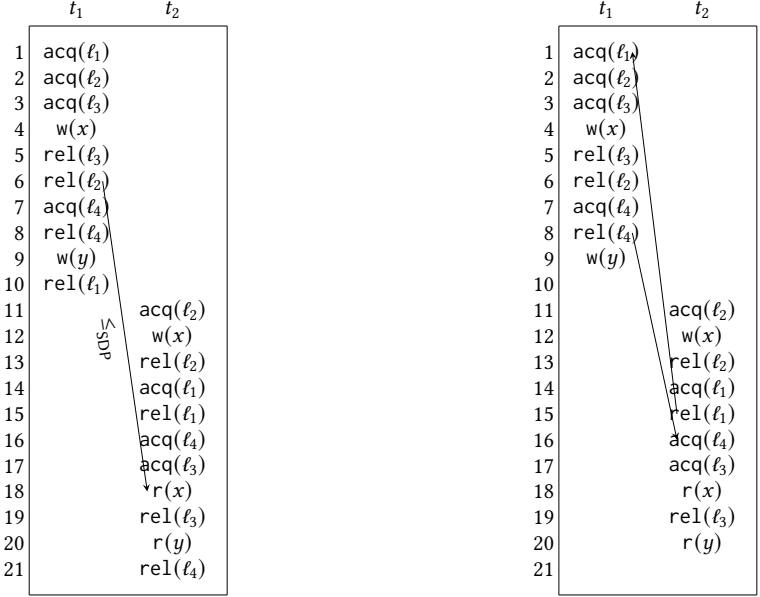

\begin{subfigure}[b]{0.45\textwidth}
\centering
\execution{2}{
\figev{1}{\acq(\ell_1)}
\figev{1}{\acq(\ell_2)}
\figev{1}{\acq(\ell_3)}
\figev{1}{\wt(x)}
\figev{1}{\rel(\ell_3)}
\figev{1}{\rel(\ell_2)}
\figev{1}{\acq(\ell_4)}
\figev{1}{\rel(\ell_4)}
\figev{1}{\wt(y)}
\figev{1}{\rel(\ell_1)}
\figev{2}{\acq(\ell_2)}
\figev{2}{\wt(x)}
\figev{2}{\rel(\ell_2)}
\figev{2}{\acq(\ell_1)}
\figev{2}{\rel(\ell_1)}
\figev{2}{\acq(\ell_4)}
\figev{2}{\acq(\ell_3)}
\figev{2}{\rd(x)}
\figev{2}{\rel(\ell_3)}
\figev{2}{\rd(y)}
\figev{2}{\rel(\ell_4)}
%\orderedge{1}{6}{0.4}{2}{18}{-0.4}
\orderedgewithlabel{1}{6}{0.4}{2}{18}{-0.4}{\small $\sdp{}$}{below}
}
\caption{
A trace $\tr$ with an \textsf{SDP}-race but no predictable race.
}
\figlabel{sdp1}
\end{subfigure}
\qquad
\begin{subfigure}[b]{0.45\textwidth}
\centering
\execution{2}{
\figev{1}{\acq(\ell_1)}
\figev{1}{\acq(\ell_2)}
\figev{1}{\acq(\ell_3)}
\figev{1}{\wt(x)}
\figev{1}{\rel(\ell_3)}
\figev{1}{\rel(\ell_2)}
\figev{1}{\acq(\ell_4)}
\figev{1}{\rel(\ell_4)}
\figev{1}{\wt(y)}
\figev{1}{}
\figev{2}{\acq(\ell_2)}
\figev{2}{\wt(x)}
\figev{2}{\rel(\ell_2)}
\figev{2}{\acq(\ell_1)}
\figev{2}{\rel(\ell_1)}
\figev{2}{\acq(\ell_4)}
\figev{2}{\acq(\ell_3)}
\figev{2}{\rd(x)}
\figev{2}{\rel(\ell_3)}
\figev{2}{\rd(y)}
\figev{1}{}
\orderedge{1}{8}{0.4}{2}{16}{-0.4}
\orderedge{2}{15}{-0.4}{1}{1}{0.4}
}
\caption{
Attempt for a correct reordering of $\tr$ with a race on $(e_{9}, e_{20})$.
}
\figlabel{sdp2}
\end{subfigure}
\caption{Counterexample to the soundness of \textsf{SDP}.}
\figlabel{sdp}
\end{figure}

\Paragraph{Counterexample to \textsf{SDP} soundness.}
Our counterexample is shown in \figref{sdp}.
First, we argue that the trace $\tr$ has an \textsf{SDP}-race.
Second, we argue that $\tr$ has no predictable race.

\begin{compactenum}
\item Observe that $e_9$ and $e_{20}$ are conflicting and are not protected by the same lock.
Hence, it suffices to argue that $e_9 \notsdp{\tr} e_{20}$.
Since the two critical sections on $\ell_2$ contain the $\wt(x)$ conflicting events $e_4$ and $e_{12}$,
\textsf{SDP} will order $e_6\sdp{\tr} e_{18}$, as $e_{18}$ is a $\rd(x)$ event that is thread-ordered after $e_{12}$.
At this point, \textsf{SDP} will insert no orderings, hence $e_9 \notsdp{\tr} e_{20}$, and $(e_9, e_{20})$ is an SDP-race.

\item There are three conflicting event pairs that may constitute a predictable data race, namely, 
(i)~$(e_4,e_{12})$, 
(ii)~$(e_4, e_{17})$, and
(iii)~$(e_9,e_{20})$.
Observe that $(e_4,e_{12})$ and $(e_4, e_{17})$ cannot yield a predictable race, as the event pairs are protected by the same locks $\ell_2$ and $\ell_3$, respectively.
For $(e_9,e_{20})$, consider an attempt for constructing a correct reordering $\tr^*$ that witnesses the race, as shown in \figref{sdp2}.
Observe that $\tr^*$ is missing the $\rel(\ell_1)$ event $e_{10}$ and $\rel(\ell_4)$ event $e_{21}$.
Since $\tr^*$ must respect lock semantics, it must satisfy the two orderings shown in \figref{sdp2}.
Note, however, that these two orderings necessarily violate the observation of the $\rd(x)$ event $e_{18}$.
Thus, $\RF_{\tr}\neq \RF_{\tr^*}$, and $\tr^*$ cannot be a correct reordering of $\tr$.
\end{compactenum}

\bigskip

We thank Casper Abild Larsen and Simon Sataa-Yu Larsen for helpful discussions on \textsf{SDP}.

\end{document}